\documentclass[11pt,twoside]{article}
\pdfoutput=1
\usepackage[margin=1in]{geometry}
\usepackage[T1]{fontenc}
\usepackage[USenglish]{babel}
\usepackage{amsmath,amssymb,amsthm}
\usepackage{enumitem}
\usepackage{caption}
\usepackage[hypertexnames=false]{hyperref}
\usepackage{thm-restate}
\usepackage[capitalise,noabbrev]{cleveref}
\usepackage[scaled]{helvet}
\theoremstyle{definition}
\newtheorem{theorem}{Theorem}[section]
\newtheorem{lemma}[theorem]{Lemma}
\newtheorem{corollary}[theorem]{Corollary}
\newtheorem{proposition}[theorem]{Proposition}
\newtheorem{definition}[theorem]{Definition}

\newtheorem*{claim*}{Claim}
\newtheorem*{remark*}{Remark}

\crefname{algocf}{Algorithm}{Algorithms}
\Crefname{algocf}{Algorithm}{Algorithms}
\usepackage{setspace}
\setdisplayskipstretch{}
\usepackage{etoolbox}
\AtBeginEnvironment{align}{\vspace{-6pt}}
\AtEndEnvironment{align}{\vspace{-6pt}}
\AtBeginEnvironment{align*}{\vspace{-6pt}}
\AtEndEnvironment{align*}{\vspace{-6pt}}
\AtBeginEnvironment{equation}{\vspace{-6pt}}
\AtEndEnvironment{equation}{\vspace{-6pt}}
\AtBeginEnvironment{equation*}{\vspace{-6pt}}
\AtEndEnvironment{equation*}{\vspace{-6pt}}
\hypersetup{hidelinks}

\usepackage{commath,mathtools,mathrsfs,stmaryrd,dsfont}
\usepackage{graphicx,booktabs,etoolbox}
\usepackage[noend]{algorithmic}
\usepackage{algorithm}
\usepackage[algo2e,linesnumbered,ruled,vlined]{algorithm2e}
\DontPrintSemicolon
\usepackage{tikz}
\usetikzlibrary{decorations.pathmorphing,calc,shapes,arrows,positioning,fit,backgrounds}
\usepackage[normalem]{ulem}
\usepackage{xcolor}

\newcommand{\mc}[1]{\mathcal{#1}}

\newcommand{\ignore}[1]{}

\newcommand{\calX}{\mathcal X}

\newcommand{\F}{\mathbb F} 
\newcommand{\D}{\mathbf D} 

\newcommand{\batya}[1]{{\texttt{\color{blue} Batya: [{#1}]}}}

\newcommand*{\rom}[1]{\expandafter\@slowromancap\romannumeral #1@}
\newcommand{\RNum}[1]{\uppercase\expandafter{\romannumeral #1\relax}}

\newtheorem{defn}{Definition}[section]

\newcommand{\eat}[1]{}

\eat{

\newtheorem{corollary}[defn]{Corollary}

\newtheorem{proposition}[defn]{Proposition}

\newcommand{\sel}[1]{{\sigma}}
}

\newcommand{\cut}[1]{}

\newcommand{\defeq}{\stackrel{\text{def}}{=}}

\def\set#1{\mathord{\{#1\}}}

\def\abs#1{\mathord{\lvert#1\rvert}}

\def\eqdef{\mathrel{\stackrel{\textsf{\tiny def}}{=}}}

\def\D{\mathcal{D}}

\def\e#1{\emph{#1}}
\newenvironment{citedtheorem}[1]
{\begin{theorem}{\it\e{(#1)}}\,\,}
	{\end{theorem}}
\newenvironment{citedlemma}[1]
{\begin{lemma}{\it\e{(#1)}}\,\,}
	{\end{lemma}}
\newenvironment{citeddefinition}[1]
{\begin{definition}{\it\e{(#1)}}\,\,}
	{\end{definition}}

\newenvironment{repeatresult}[2]
{\vskip0.5em\par\textsc{#1} #2.}
{\vskip0.5em}

\newenvironment{reptheorem}[1]{\begin{repeatresult}{Theorem}{#1}}{\end{repeatresult}}
\newenvironment{replemma}[1]{\begin{repeatresult}{Lemma}{#1}}{\end{repeatresult}}
\newenvironment{repcorollary}[1]{\begin{repeatresult}{Corollary}{#1}}{\end{repeatresult}}

\def\eqdef{\mathrel{\stackrel{\textsf{\tiny def}}{=}}}

\def\e#1{\emph{#1}}

\newcommand{\algname}[1]{{\sf #1}}
\def\myrulewidth{2.80in}
\def\therule{\rule{\myrulewidth}{0.2pt}}

\def\myrulewidthwide{6in}
\def\therulewide{\rule{\myrulewidthwide}{0.2pt}}

\newenvironment{insidecode}[3]
{
	\begin{tabular}{p{\myrulewidth}}
		\multicolumn{1}{c}{\rule{0mm}{3mm}{\bf #3} $\algname{#1}(\mbox{#2})$\vspace{-0.6em}}\\
		\therule\vskip-0.8em\therule
		\vspace{-1em}
		\begin{algorithmic}[1]}
		{\end{algorithmic}
		\vskip-0.4em\therule
\end{tabular}}

\newenvironment{insidecodewide}[3]
{
	\begin{tabular}{p{\myrulewidthwide}}
		\multicolumn{1}{c}{\rule{0mm}{3mm}{\bf #3} $\algname{#1}(\mbox{#2})$\vspace{-0.6em}}\\
		\therulewide\vskip-0.8em\therulewide
		\vspace{-1em}
		\begin{algorithmic}[1]}
		{\end{algorithmic}
		\vskip-0.3em\therulewide
\end{tabular}}

\newcommand{\minlsep}[2]{\mathcal{S}_{#1}(#2)}

\newcommand{\minlsepst}[1]{\mathcal{S}_{s,t}(#1)}
\newcommand{\minsepst}[1]{\mathcal{L}_{s,t}(#1)}
\newcommand{\minlsepImp}[2]{\mathcal{S}^{*}_{#1}(#2)}

\def\minsep{\mathcal{L}}

\newcommand{\edges}{\texttt{E}}

\newcommand{\nodes}{\texttt{V}}

\newcommand{\minus}{\scalebox{0.75}[1.0]{$-$}}

\definecolor{mygreen}{rgb}{0,0.6,0}
\definecolor{mygray}{rgb}{0.5,0.5,0.5}
\definecolor{mymauve}{rgb}{0.58,0,0.82}
\definecolor{cadmiumgreen}{rgb}{0.0, 0.42, 0.24}

\def\gq1{{\geq}1}

\def\cc{\mathcal{C}}

\newcommand{\minstVertices}[1]{\mathcal{U}^{\min}_{#1}}

\newif \ifnonplanar
\nonplanarfalse 

\newcommand{\sminus}{\scalebox{0.75}[1.0]{$-$}}

\newcommand{\safeSeps}[2]{\mathcal{F}_{#1,#2}}
\newcommand{\safeSepsk}[3]{\mathcal{F}_{#1,#2,#3}}
\newcommand{\safeSepsMin}[2]{\mathcal{F}_{#1,#2,\min}}
\newcommand{\safeSepsImp}[2]{\mathcal{F}^{*}_{#1,#2}}
\newcommand{\safeSepskImp}[3]{\mathcal{F}^{*}_{#1,#2,#3}}

\def\MHS{\mathrm{MHS}}

\newcommand{\closestMinSep}[1]{L^{*}_{#1}}

\def\CleanSubmission{}
\newif\ifshowchanges
\showchangestrue
\ifdefined\CleanSubmission\showchangesfalse\fi
\newcommand{\rev}[1]{{\ifshowchanges\color{red}\else\color{black}\fi #1}}
\newenvironment{revision}{\ifshowchanges\color{red}\else\color{black}\fi}{}

\newcommand{\apprev}[1]{{\ifshowchanges\color{blue}\else\color{black}\fi #1}}
\newenvironment{applicationrevision}{\ifshowchanges\color{blue}\else\color{black}\fi}{}

\newif\ifshowclarityadditions
\showclarityadditionsfalse
\definecolor{claritypurple}{RGB}{128,0,128}
\newcommand{\clarityadd}[1]{{\ifshowclarityadditions\color{claritypurple}\fi #1}}

\makeatletter
\expandafter\def\csname appbib@MixedCut2016\endcsname{}
\expandafter\def\csname appbib@MultiBudgetedCuts2020\endcsname{}
\expandafter\def\csname appbib@JaffkeMorelleSauThilikos2023\endcsname{}
\expandafter\def\csname appbib@OCTLP2014\endcsname{}
\expandafter\def\csname appbib@TreewidthApprox2016\endcsname{}
\expandafter\def\csname appbib@StrongBipartiteDeletion2016\endcsname{}

\let\application@originalbibitem\@bibitem
\def\@bibitem#1{%
  \ifshowchanges
    \ifcsname appbib@#1\endcsname\color{blue}\else\color{black}\fi
  \fi
  \application@originalbibitem{#1}%
}
\makeatother

\newcommand{\mediumbigcup}{\mathop{\scalebox{1.2}{$\bigcup$}}}
\newcommand{\mediumbigcap}{\mathop{\scalebox{1.2}{$\bigcap$}}}
\newcommand{\mediumbiguplus}{\mathop{\scalebox{1.2}{$\biguplus$}}}
\title{Connectivity-Preserving Important Separators: A Framework for Cut-Uncut Problems}
\author{Batya Kenig}
\date{}
\begin{document}
\hypersetup{pdftitle={Connectivity-Preserving Important Separators: A Framework for Cut-Uncut Problems},pdfauthor={Batya Kenig},pdfsubject={}}
\pagenumbering{gobble}
\maketitle
\begin{abstract}
	Important separators are a cornerstone of parameterized algorithms for graph separation:
	they reduce an a priori enormous search space of separators to a small, structured family that
	can be enumerated efficiently.  This principle has been remarkably successful for parameterized separation problems. We extend its enumeration approach to cut-uncut problems, where one must cut some connections while preserving the connectivity of a given set of terminals. These connectivity-preservation requirements create a qualitatively different type of structure, and the classical important-separator machinery no longer gives the right objects to enumerate.
	
	We introduce \emph{connectivity-preserving important separators} in undirected
graphs. These are inclusion-minimal vertex separators that disconnect $s$
from $t$, keep a prescribed terminal set connected to $s$, and admit no
such separator of no greater size whose source component is properly
contained in their own. Our main result shows that, despite the additional connectivity	constraints, the number of such separators of size at most $k$ is bounded by $2^{O(k^2\rev{\log k})}$, and they can be enumerated in	$O(2^{O(k^2\rev{\log k})}\cdot n\cdot T(n,m))$ time, where $n$ and $m$ are the numbers of vertices and edges, and $T(n,m)$ is the time, including input reading, for computing a minimum-cardinality $s,t$-separator on a graph with $n$ vertices and $m$ edges.
	This gives a systematic extension of the important-separator method with connectivity constraints.
	
	\apprev{For \textsc{Mixed Multiway Cut-Uncut} with any fixed number of terminal
	classes, we improve the published
	$2^{O(p^3\log p)}$ parameter dependence to $2^{O(p^2\log p)}$,
	where $p$ is the sum of the vertex and edge deletion budgets.
	For \textsc{Annotated Odd Cycle Transversal}, we add connectivity of an unbounded
	terminal set whose vertices cannot be deleted, in $2^{O(k^2\log k)}(n+m+1)^{O(1)}$ time,
	improving the doubly exponential dependence of a direct treewidth-reduction
	approach.}
We also obtain a fixed-parameter tractable (FPT) algorithm for finding a minimum-cardinality minimal
$s,t$-separator whose source component must contain a prescribed
set $A$ and avoid a prescribed set $B$.

\end{abstract}

\begin{quote}\small
\noindent\textbf{2012 ACM Subject Classification}
Theory of computation $\rightarrow$ Parameterized complexity and exact algorithms;
Theory of computation $\rightarrow$ Graph algorithms analysis

\end{quote}
\clearpage
\pagenumbering{arabic}
\section{Introduction}
\label{sec:introduction}
Graph separation is a cornerstone of parameterized algorithm design. A particularly powerful tool is the notion of \emph{important separators}~\cite{DBLP:journals/tcs/Marx06}. For a given source set $X$ and target set $Y$, an important $X,Y$-separator is a minimal separator that is ``pushed'' as far towards $Y$ as possible. Formally, it maximizes the vertices reachable from $X$: no separator of equal or smaller size leaves a strict superset of these vertices connected to $X$. The number of important separators of size at most $k$ is bounded by $4^k$, allowing efficient enumeration~\rev{\cite{DBLP:books/sp/CyganFKLMPPS15}}. This framework is standard in fixed-parameter tractable (FPT) algorithms for problems like \textsc{Multiway Cut}, \textsc{Directed Feedback Vertex Set}, \textsc{Almost 2-SAT}, and more~\cite{DBLP:journals/tcs/Marx06,DBLP:journals/siamcomp/BousquetDT18,DBLP:journals/siamcomp/MarxR14,DBLP:journals/siamcomp/ChitnisHM13,DBLP:journals/jacm/ChenLLOR08,DBLP:journals/corr/abs-0801-1300,DBLP:journals/iandc/LokshtanovM13}.

However, many  separation problems include \emph{connectivity constraints}: one must separate designated terminals while keeping specified terminal sets connected. This tension has been studied under several names, including \emph{connectivity-preserving minimum cut}, \emph{network diversion}, and, in parameterized complexity, \emph{cut-uncut} problems~\cite{DBLP:journals/jcss/DuanX14,DBLP:journals/networks/CullenbineWN13,DBLP:conf/icalp/BentertDFGK24,DBLP:conf/iwpec/BentertFH024}. A fundamental variant is \textsc{Two-Sets Cut-Uncut}: one must separate two terminal sets $S$ and $T$ such that $S$ and $T$ each remain internally connected after the cut~\cite{DBLP:conf/wg/TelleV13,DBLP:conf/iwpec/BentertFH024}. 
\rev{Deciding whether an edge-deletion Two-Sets Cut-Uncut solution exists is NP-complete~\cite{DBLP:conf/icalp/BentertDFGK24}; connectivity-preserving minimum node cut is hard to approximate~\cite{DBLP:journals/jcss/DuanX14}; and edge-deletion Two-Sets Cut-Uncut is W[1]-hard parameterized by $|S|+|T|$~\cite{DBLP:conf/icalp/BentertDFGK24}.}
While recent work has successfully tackled restricted graph classes~\cite{DBLP:journals/tcs/GolovachKP13,kenig2025connectivitypreserving,DBLP:conf/icalp/BentertDFGK24}, developing efficient parameterized algorithms for general graphs remains a highly active research frontier~\cite{DBLP:conf/iwpec/BentertFH024,DBLP:conf/icalp/BentertDFGK24,bentert2025planarnetworkdiversion,DBLP:conf/wg/TelleV13}.

Let $G$ be a finite simple undirected graph, with vertex set $\nodes(G)$ and edge set $\edges(G)$. For disjoint $X,Y\subseteq\nodes(G)$, an $X,Y$-\e{separator} is a set $S\subseteq\nodes(G){\setminus}(X\cup Y)$ intersecting every path from $X$ to $Y$. For a separator $S$ and a vertex set $X$, we denote by $G\sminus S$ the graph that results from $G$ by removing vertices $S$ and their adjacent edges, and by $C_X(G\sminus S)$ the vertices reachable from $X$ in $G\sminus S$, namely the union of all connected components of $G\sminus S$ that contain a vertex of $X$. An $X,Y$-separator $S$ is \e{minimal} if no strict subset of $S$ separates $X$ from $Y$. 
Under the convention used in this paper, a minimal $X,Y$-separator $S$ is \e{important} if there is no minimal $X,Y$-separator $S'$ such that $C_X(G\sminus S')\subsetneq C_X(G\sminus S)$ and $|S'|\le |S|$. Thus, $S$ is \e{undominated} in the sense that one cannot obtain a strictly smaller reachable source set without using a larger separator. This is symmetric to the standard convention, which maximizes the source side (i.e., $C_X(G\sminus S)$) among separators of no larger size~\cite{DBLP:journals/tcs/Marx06,DBLP:conf/wg/Marx11,DBLP:books/sp/CyganFKLMPPS15}: applying the standard convention to $Y,X$ gives the orientation used here (see Lemma~\ref{lem:inclusionCsCt}). \eat{For fixed $X$ and $Y$, there are at most $4^k$ important $X,Y$-separators of size at most $k$, and they can be enumerated efficiently~\cite{DBLP:books/sp/CyganFKLMPPS15}. This bounded representative family of separators is a standard local primitive in parameterized algorithms~\cite{DBLP:journals/tcs/Marx06,DBLP:conf/wg/Marx11,DBLP:journals/siamcomp/BousquetDT18,	DBLP:books/sp/CyganFKLMPPS15,DBLP:journals/siamcomp/MarxR14,DBLP:journals/siamcomp/ChitnisHM13,DBLP:journals/iandc/LokshtanovM13}.
}

The main contribution of this paper is to recover the important-separator
paradigm when connectivity constraints are introduced: separators must be
nondominated not only for separating $s$ from $t$, but also for preserving connectivity among a predefined set of terminals. We introduce \emph{connectivity-preserving important separators}. Fix vertices $s,t$ and a terminal set $A$ with $s\in A$ and $t\notin A$. A minimal $s,t$-separator $S$ is connectivity-preserving, or CP, with respect to $A$ if all vertices of $A$ remain in the component of $G\sminus S$ containing $s$: $A\subseteq C_s(G\sminus S)$. A CP $s,t$-separator $S$ is CP-important if there is no CP $s,t$-separator $S'$ such that
\[
A\subseteq C_s(G\sminus S')\subsetneq C_s(G\sminus S)
\qquad\text{and}\qquad
|S'|\le |S|.
\]
Our applications need a connected set containing $A$ that lies inside the
source component of an unknown solution. This containment keeps the set
disjoint from that solution's deleted vertices. Shrinking the source component until its
separator is CP-important supplies such a set; maximizing it does not give
this containment guarantee.

\paragraph*{Our contribution.}
Enumerating CP-important separators is more intricate than in the standard setting. Classical important separators rely on the existence of a unique, canonical ``rightmost'' (or symmetrically ``leftmost'') minimum separator that anchors the entire search process~\cite{DBLP:journals/tcs/Marx06}. Among minimum-cardinality CP separators, there need not be a unique separator with an inclusionwise smallest source component. To navigate this lack of a single ``leftmost'' anchor, we develop a branching strategy based on repairing connectivity constraint violations. We prove that connectivity violations in the unconstrained canonical separator can be systematically ``repaired'' by forcing the violating components into the correct connected component. We encode this repair step as a  hitting-set problem on the boundaries of these components, justified by a novel structural property concerning the union of all minimum separators. Using this new structural tool, we prove that despite the loss of a unique anchor, the number of CP important separators of size at most $k$ remains bounded by \rev{$2^{O(k^2\rev{\log k})}$}. Furthermore, we provide an explicit algorithm to enumerate all such separators in time $O(2^{O(k^2\rev{\log k})} \cdot n \cdot T(n,m))$, where $T(n,m)$ is the time required to compute a standard minimum $s,t$-separator~\cite{Chen2022,DBLP:journals/corr/abs-2309-16629,DBLP:journals/siamcomp/EvenT75}. \eat{This provides the missing separator theory for cut-uncut problems, showing that separator-based methods can simultaneously enforce separation and connectivity.}

\rev{In what follows, $T(n,m)$ denotes the time to compute a minimum-cardinality $s,t$-separator in a graph with $n$ vertices and $m$ edges. Deterministic exact maximum flow~\cite{DBLP:journals/corr/abs-2309-16629} gives $T(n,m)=(n+m)^{1+o(1)}$, which is $m^{1+o(1)}$ on connected nontrivial inputs.}

\eat{Our main theorem is the missing representative-family theorem for separation under connectivity-preservation constraints, showing that all nondominated separators of size at most $k$ are captured by a bounded enumerable family.}
\def\singleExponentialSafeImportantOfSizek{
	Let $G$ be a simple, undirected graph with $n$ vertices and $m$ edges, let $A\subseteq \nodes(G)$ with $s\in A$ and $t\notin A$, and let $k \in \mathbb N_{\geq 0}$. The number of connectivity-preserving important $s,t$-separators with respect to $A$ of size at most $k$ is
	$2^{O(k^2\rev{\log k})}$, and they can be enumerated in time	$O(2^{O(k^2\rev{\log k})}\cdot n\cdot T(n,m))$.
}

\begin{theorem}
	\label{thm:singleExponentialSafeImportantOfSizek}
	\singleExponentialSafeImportantOfSizek
\end{theorem}

\eat{For a prescribed tuple $(G,s,t,A)$, the theorem gives an explicit
representative family: it enumerates the relevant separators together with
their source components $C_s(G\sminus S)$. This yields the following fixed-terminal isolation primitive.}
\subsection{\apprev{Applications}}
\label{sec:intro-applications}

\apprev{We first state a source-component optimization consequence of
Theorem~\ref{thm:singleExponentialSafeImportantOfSizek}.}
\def\targetedCPIsolation{
	Let $G$ be a simple, undirected graph with $n$ vertices and $m$ edges, let
	$s,t\in\nodes(G)$ with $s\neq t$, $A,B\subseteq\nodes(G){\setminus}\set{s,t}$, and
	$k\in \mathbb N$. In time
	$O(2^{O(k^2\rev{\log k})}\cdot n\cdot \rev{(n+m)^{1+o(1)}})$, one can either find a
	minimum-cardinality, minimal $s,t$-separator $S$ of size at most $k$ satisfying $A\subseteq C_s(G\sminus S)$ and $B\cap C_s(G\sminus S)=\emptyset$, or report that no such separator exists.
}

\begin{corollary}
	\label{cor:targetedCPIsolation}
\targetedCPIsolation
\end{corollary}
\rev{Vertices of $B$ may be deleted and need not be connected to $t$ or
	one another; filtering source components enforces their exclusion.
	Every deleted vertex must be necessary to separate $s$ from $t$:
	extra deletions solely to exclude $B$ are forbidden.
	Standard cut-uncut formulations do not
	impose this requirement for a designated terminal pair~\cite{DBLP:journals/siamcomp/ChitnisCHPP16,ChuLLPTX2026}.
	Secluded-subgraph enumeration permits inclusion and exclusion constraints,
	but does not impose this minimal-separator condition~\cite{DBLP:journals/jcss/JansenKW25}.}
	
\apprev{The enumeration theorem (Theorem~\ref{thm:singleExponentialSafeImportantOfSizek}) gives a common way to impose connectivity
constraints on deletion problems: identify a bounded number of connected subgraphs such that \clarityadd{whenever a feasible solution exists, at least one of them is
preserved by some feasible solution}, then forbid deletion of its vertices and, when edge deletion
is allowed, its edges. We call this
\e{protecting} the subgraph. Its size need not be bounded by the deletion budget.
Both applications concern simple undirected graphs.}

\begin{applicationrevision}
\paragraph*{Mixed Multiway Cut-Uncut.}
Given disjoint nonempty terminal classes $A_1,\ldots,A_r$, delete at most
$k$ nonterminal vertices and at most $\ell$ edges so that each class lies
in one connected component and different classes lie in different components.
The classes may be arbitrarily large, and $G[A_i]$ need not be connected.
Write $p=k+\ell$.
Rai, Ramanujan, and Saurabh~\cite[Theorem~3]{MixedCut2016} solve this problem
in $2^{O(p^3\log p)}n^4\log n$ time, independently of $r$.
For every fixed $r$, we improve the parameter exponent from
$O(p^3\log p)$ to $O(p^2\log p)$ while respecting both budgets separately.

\begin{theorem}
\label{thm:mixed-two-sets}
\textsc{Mixed Multiway Cut-Uncut} with $r\geq1$ terminal classes can be
solved deterministically in
\[
2^{O(rp^2\log p)}(n+m+1)^{O(1)},\qquad p=k+\ell,
\]
where $n=|\nodes(G)|$ and $m=|\edges(G)|$.
For every fixed $r$, the parameter dependence is
$2^{O(p^2\log p)}$, independently of the sizes of the terminal classes.
The same bound permits additional vertices forbidden for vertex deletion
and edges forbidden for edge deletion.
\end{theorem}

To our knowledge, this is the first algorithm with this parameter dependence
for every fixed number of terminal classes. The displayed dependence on $r$
is essential to the comparison: we do not improve the published bound for
unrestricted $r$.
Enumerating CP-important separators
(Theorem~\ref{thm:singleExponentialSafeImportantOfSizek}) gives, for each
terminal class, a bounded family of candidate connected subgraphs
containing that class.
We try every choice of pairwise disjoint candidates, one per class.
For each choice, multi-budgeted cut techniques~\cite{MultiBudgetedCuts2020}
separate the classes while keeping the chosen subgraphs intact and
respecting the vertex- and edge-deletion budgets. Section~\ref{sec:mixed-two-sets} explains the
idea using two classes; Appendix~\ref{sec:mixed-two-sets-appendix} gives
the general proof.

\paragraph*{Annotated Odd Cycle Transversal with terminal connectivity.}
An odd cycle transversal (OCT) is a vertex set whose deletion makes a
graph bipartite, meaning that its vertices admit a two-colouring in which
adjacent vertices receive different colours. Such a colouring is proper. \e{Annotated} OCT also prescribes colors on \e{protected} vertices (i.e., vertices that cannot be deleted as part of the solution)
\cite[Appendix~B]{JaffkeMorelleSauThilikos2023}.
We additionally require those vertices to remain connected.
Given a nonempty set $A$ and $\phi:A\to\{0,1\}$, seek
$X\subseteq\nodes(G){\setminus} A$, $|X|\leq k$, such that $G\sminus X$
has a proper two-colouring agreeing with $\phi$ on $A$ and all of $A$
lies in one connected component of $G\sminus X$. The set $A$ is unbounded and need not induce a
connected subgraph.

This requirement changes the problem: on a path $a-x-b$ with $r\geq2$
triangles, each containing $x$ and two distinct new vertices, and protected terminals
$a,b$ both coloured zero, ordinary \textsc{Annotated OCT} deletes $x$.
Preserving the connection between $a,b$ forces $x$ to survive and requires
one deletion from each of the $r$ triangles. Thus the optimum increases from one to $r$.
The $2.3146^k n^{O(1)}$ ordinary OCT algorithm~\cite{OCTLP2014} does not
by itself enforce the required connectivity constraint.

\begin{theorem}
\label{thm:annotated-oct-connectivity}
\textsc{Annotated Odd Cycle Transversal with terminal connectivity} can be
solved deterministically in $2^{O(k^2\log k)}(n+m)^{O(1)}$,
where $n=|\nodes(G)|$ and $m=|\edges(G)|$.
The parameter dependence is independent of $|A|$.
\end{theorem}

To our knowledge, this is the first algorithm for this connectivity-constrained
problem with the stated $2^{O(k^2\log k)}$ parameter dependence,
independent of $|A|$.
For comparison, adapting the treewidth-reduction machinery of
Marx, O'Sullivan, and Razgon~\cite{DBLP:journals/talg/MarxOR13} gives a
$2^{2^{O(k^2)}}(n+m)^{O(1)}$ algorithm; we derive this bound in
Proposition~\ref{prop:annotated-oct-treewidth}.
Our algorithm replaces this doubly exponential dependence by
$2^{O(k^2\log k)}$ and shows that CP enumeration applies beyond
cut demands, to destroying odd cycles while preserving actual terminal
connectivity. Section~\ref{sec:annotated-oct-connectivity} gives the
overview; Appendix~\ref{sec:annotated-oct-appendix} proves the theorem and
the comparison bound.
\end{applicationrevision}

\paragraph*{Related Work.}
Our framework combines inclusionwise-minimal source components with
connectivity preservation: among ordinary minimal $s,t$-separators
of no greater size, a CP-important separator admits no strictly
smaller source component that still connects the required terminals.
Whenever a solution exists, our reductions guarantee at least one
enumerated choice of connected sets contained in the terminal components
of an optimal solution. Keeping that choice intact leaves the same solution
available to the remaining deletion algorithm.
This containment guarantee is absent from the statements of the following
secluded-connectivity results.
Chechik et al.~\cite{DBLP:conf/esa/ChechikJPP13} minimize exposure,
which in the unweighted case counts the vertices in the connected subgraph and its
neighborhood; our budget counts only deleted vertices.
Donkers et al.~\cite{DBLP:conf/ipec/DonkersJK22} consider induced
trees, whereas our preserved components may contain cycles.
Jansen et al.~\cite{DBLP:journals/jcss/JansenKW25} enumerate a bounded family
containing all connected induced subgraphs with bounded neighborhoods that avoid a fixed finite family
of forbidden induced subgraphs and cannot be enlarged, subject to prescribed
inclusions and exclusions, without increasing their boundary size.
A maximal connected component, however, may include vertices that cannot all
survive in any optimal solution.
Keeping such an enlarged connected set intact can therefore exclude every optimum,
rather than merely select a different one.
For OCT for example, the obstruction is immediate if the vertices selected to
survive contain an odd cycle: every feasible solution must delete
at least one of them.
Bipartiteness excludes odd induced cycles of unbounded length, so it is not
one of the finitely characterized properties covered by that theorem.
Our contained CP-important representatives avoid it: they preserve
the required connections without forbidding any deletion of the
optimal solution represented by that choice.

\subsection{\apprev{Limits in directed graphs}}
In a directed graph, CP requires every vertex of $A$ to remain reachable
from $s$ after deletion; importance minimizes this reachable set, under
set inclusion, among CP separators of no greater size.
\rev{We show that in directed graphs, the number of CP-important separators of size
	at most $k$ cannot be bounded by any function of $k$ alone. This contrasts
	with ordinary important separators, whose number is bounded by $4^k$
	in both undirected and directed graphs~\cite{DBLP:books/sp/CyganFKLMPPS15}.
	Specifically, we show that in directed graphs there can be arbitrarily many CP-important 
	$s,t$-separators already
	when $k=2$. More generally, our construction gives
	$n^{\Omega(k)}$ such separators, ruling out an algorithm that enumerates
	them all in total time $f(k)n^{O(1)}$: the number of outputs alone
	exceeds any such bound. Thus, the bounded-family enumeration approach
	available in undirected graphs cannot extend to directed graphs.}

\begin{theorem}[Directed lower bound]
	\label{thm:directed-lower-bound}
	\rev{For every pair of integers $N\geq2$ and $r\geq1$, there is a
		directed graph on $r(3N-1)+2$ vertices, with a source $s$, target $t$,
		and prescribed set $A$ containing $s$, that has exactly $N^r$
		connectivity-preserving-important $s,t$-separators of size at most $2r$.
		In particular, for $r=1$, there are exactly $N$ such separators
		of size at most two.}
\end{theorem}

\paragraph*{Organization.}
\apprev{Section~\ref{sec:Preliminaries} gives notation and background.
Sections~\ref{sec:cp_min_seps} and~\ref{sec:algo} develop the structural
characterization and enumeration algorithm.
Section~\ref{sec:applications} proves the source-component optimization
corollary and presents the two applications of
Section~\ref{sec:intro-applications}.
The appendix contains the full application proofs and the directed lower
bound in Appendix~\ref{sec:cp-important-lower-bound}.
AI assistance is disclosed in Appendix~\ref{sec:ai-assistance}.}

\section{Preliminaries and Notation}
\label{sec:Preliminaries}
Let $G$ be a finite simple undirected graph with vertex set $\nodes(G)$,
edge set $\edges(G)$, $n=|\nodes(G)|$, and $m=|\edges(G)|$.
For vertex sets, write $AB=A\cup B$ and $vA=\{v\}\cup A$.
The open neighborhood $N_G(v)$ contains the neighbors of $v$, and
$N_G[v]=N_G(v)\cup\{v\}$. For a set $U$,
$N_G(U)=(\bigcup_{v\in U}N_G(v)){\setminus} U$ is its external boundary,
and $N_G[U]=U\cup N_G(U)$.
The induced subgraph $G[U]$ retains exactly the vertices of $U$ and the
edges with both endpoints in $U$; write $G\sminus U=G[\nodes(G){\setminus} U]$
and $G\sminus v=G\sminus\{v\}$.
A subgraph may omit vertices and edges.

A path is a sequence of distinct vertices with adjacent consecutive
vertices. A set $U$ is connected if every two of its vertices are joined
by a path in $G[U]$; a connected component is an inclusion-maximal
connected set. Write $\cc(G)$ for the components, $C_v(G)$ for the one
containing $v$, and $C_A(G)=\bigcup_{a\in A}C_a(G)$.
In $G\sminus S$, these expressions are used only for surviving vertices $\nodes(G){\setminus}S$
and sets $A$ disjoint from $S$.
Contracting a connected set $U$ to $u\in U$ replaces $U$ by $u$ with
neighborhood $N_G(U)$. All graph operations discard loops and repeated
edges. Two disjoint sets are adjacent if an edge joins them; a path
between sets has one endpoint in each.

For a family $\mathcal E$ of subsets of a universe $\mathcal U$, a
\emph{hitting set} of $\mathcal E$ is a subset $X\subseteq \mathcal U$ that intersects every member of $\mathcal E$.
It is minimal if no proper subset still does so; $\MHS(\mathcal E)$
denotes all minimal hitting sets.
A problem with parameter $k$ is \emph{fixed-parameter tractable (FPT)}
if it is solvable in $f(k)|I|^{O(1)}$ time on input $I$, for some
computable function $f$. The notation $O^*(\cdot)$ suppresses polynomial
factors in the input size.

\subsection{Minimal Separators}
\label{sec:minimalSeparators}
For disjoint vertex sets $A,B$, an $A,B$-separator is a set
$S\subseteq\nodes(G){\setminus}(A\cup B)$ intersecting every path between them.
It is \emph{minimal} if no proper subset of $S$ is an $A,B$-separator.
Write $\minlsep{A,B}{G}$ for all minimal $A,B$-separators, and
$\minlsepst{G}$ when $A=\{s\}$ and $B=\{t\}$, with $s\ne t$.
A component $C$ of $G\sminus S$ is \emph{full} for $S$ if $N_G(C)=S$.
\begin{citedlemma}{\cite{DBLP:journals/ijfcs/BerryBC00}}
\label{lem:fullComponents}
An $s,t$-separator $S$ is minimal if and only if
$N_G(C_s(G\sminus S))=N_G(C_t(G\sminus S))=S$.
\end{citedlemma}

We use the following submodularity inequality to compare the
boundaries of intersections and unions of vertex sets.
\begin{citedlemma}{Submodularity~\cite{DBLP:books/sp/CyganFKLMPPS15}}
	\label{lem:submodularity}
	For any $X,Y\subseteq \nodes(G)$:
	\begin{align*}
		|N_G(X)|+|N_G(Y)|\geq |N_G(X\cup Y)|+|N_G(X\cap Y)|.
	\end{align*}
\end{citedlemma}
For minimal $s,t$-separators, Lemma~\ref{lem:fullComponents}
identifies each separator with the boundary of its source
component, giving the following consequence.
\begin{corollary}
	\label{corr:submodularity}
	Let $S,T\in \minlsepst{G}$. Then:
	\begin{align*}
		|S|+|T|\geq |N_G(C_s(G\sminus S)\cap C_s(G\sminus T))|+|N_G(C_s(G\sminus S)\cup C_s(G\sminus T))|.
	\end{align*}
\end{corollary}

\def\inclusionCsCt{
	Let $s,t\in \nodes(G)$, and let $S,T\in \minlsepst{G}$. The following holds:
	\begin{align*}
		C_s(G\sminus S)\subseteq C_s(G\sminus T) &&\Longleftrightarrow&& S\subseteq T\cup C_s(G\sminus T) &&\Longleftrightarrow&&   T\subseteq S\cup C_t(G\sminus S). 
	\end{align*}
}

\begin{citeddefinition}{Important separators, \rev{reversed orientation}; cf.~\cite{DBLP:books/sp/CyganFKLMPPS15}}
	\label{def:importantABSeps}
	Let $A,B\subseteq\nodes(G)$ be disjoint.
	We say that $S\in \minlsep{A,B}{G}$ is important if, for every $S'\in \minlsep{A,B}{G}$: $C_A(G\sminus S')\subsetneq C_A(G\sminus S) \Longrightarrow |S'|>|S|.$
\end{citeddefinition}
Our convention minimizes the reachable source side among separators of
no larger size. It agrees with the usual source-maximizing convention
when source and target are exchanged, as the next lemma explains.
We denote by $\minlsepImp{A,B}{G}$ the set of important $A,B$-separators, and by $\minlsepImp{A,B,k}{G}$ the set of important $A,B$-separators of $G$ whose cardinality is at most $k$. 
\begin{lemma}
	\label{lem:inclusionCsCt}
	\inclusionCsCt
\end{lemma}
\eat{
For singleton endpoints, apply Lemma~\ref{lem:inclusionCsCt} in both
orientations: shrinking $C_s$ is equivalent to enlarging $C_t$.
For set endpoints, Appendix~\ref{sec:minsepsvertexsets} proves the
corresponding reduction to singleton endpoints.
}

\def\simpABlemma{
	Let $A$ and $B$ be two disjoint, non-adjacent subsets of $\nodes(G)$. Then $S\in \minlsep{A,B}{G}$ if and only if $S$ is an $A,B$-separator, and for every $w\in S$, there exist two connected components $C_A,C_B\in \cc(G\sminus S)$ such that $C_A\cap A\neq \emptyset$, $C_B\cap B\neq \emptyset$, and $w\in N_G(C_A)\cap N_G(C_B)$.
}

\def\lemMinlsASep{
	Let $A,B\subseteq \nodes(G){\setminus}\set{s,t}$ where $sA$ and $tB$ are disjoint and non-adjacent. Let $H$ be the graph that results from $G$ by adding all edges between $s$ and $N_G[A]$. That is, $\edges(H)=\edges(G)\cup \set{(s,v):v\in N_G[A]}$. Then $\minlsep{sA,tB}{G}=\minlsep{s,tB}{H}$, $C_s(H\sminus S)=C_{sA}(G\sminus S)$, and $C_{tB}(H\sminus S)=C_{tB}(G\sminus S)$ for every $S\in \minlsep{sA,tB}{G}$.
}

\def\contractEdgesMinlSep{
	Let $u\in N_G(s)$, and let $H$ be the graph that results from $G$ by contracting the edge $(s,u)$ to vertex $s$. Then $\minlsepst{H}=\set{S\in \minlsepst{G}: u\in C_s(G\sminus S)}$.
}

\paragraph*{Minimum Separators.}
\label{sec:minseps}

A \emph{minimum} $s,t$-separator is an $s,t$-separator of minimum cardinality; its size is the $s,t$-\emph{connectivity} of $G$, denoted  $\kappa_{s,t}(G)$, and the set of all minimum $s,t$-separators is $\minsepst{G}$. Similarly, for $A,B\subseteq V(G)$, $\kappa_{A,B}(G)$ denotes the minimum size of a set $X\subseteq \nodes(G){\setminus} AB$ separating $A$ from $B$, and $\minsep_{A,B}(G)$ denotes the family of all minimum $A,B$-separators.

We denote by $\minstVertices{s,t}(G)$ the union of all minimum $s,t$-separators:
\begin{equation}
	\label{eq:minstVertices}
	\minstVertices{s,t}(G) \eqdef \set{v\in \nodes(G) : \exists S\in \minsepst{G} \text{ s.t. }v\in S}
\end{equation}
For set endpoints, likewise define $\minstVertices{A,B}(G)=\bigcup_{S\in\minsep_{A,B}(G)}S$.
An infeasible minimum has value $+\infty$ and an empty separator family.

\begin{citedlemma}{Unique minimum important separator~\cite{DBLP:books/sp/CyganFKLMPPS15}}
\label{lem:uniqueMinImportantSep}
Let $\emptyset \subsetneq A,B \subseteq \nodes(G)$. If $A$ and $B$ are disjoint and non-adjacent, then there is a unique minimum $A,B$-separator denoted $\closestMinSep{A,B}(G)$ that is important, which can be found in time $O((m+n)(\kappa_{A,B}(G)+1))$.
\end{citedlemma}
Let $A,B\subseteq\nodes(G)$ be nonempty, disjoint and non-adjacent. We say that a minimum $A,B$-separator $S\in \minsep_{A,B}(G)$ is \emph{closest to $B$} if, for every $T\in \minsep_{A,B}(G)$, it holds that $C_B(G\sminus S)\subseteq C_B(G\sminus T)$. By the standard formulation of important separators and Lemma~\ref{lem:inclusionCsCt}, Lemma~\ref{lem:uniqueMinImportantSep} implies that such a separator is unique. We denote it by $L^B_{A,B}(G)$. In particular, when $B=\{t\}$, we write $L^t_{A,t}(G)$, and when $A=\{s\}$ we write $L^t_{s,t}(G)$.

\subsection{Connectivity-Preserving (CP) Important Separators}
\label{sec:safeImportantCloseSeps}
Let $s,t\in \nodes(G)$, $A\subseteq \nodes(G)\rev{\setminus\{t\}}$. \rev{Including $s$ in $A$ has no effect; when the algorithm takes $A\subseteq\nodes(G){\setminus}\set{s,t}$, we pass $A{\setminus}\set{s}$.}\eat{We consider minimal $s,t$-separators in $G$ that meet the connectivity constraint that $sA$ remains connected after applying the separation.}
We say that $S\in \minlsepst{G}$ is connectivity-preserving (CP) with respect to $A$ if $A\subseteq C_s(G\sminus S)$.
We denote by $\safeSeps{s}{t}(G,A)$ the minimal CP $s,t$-separators with respect to $A$, or just the minimal CP $s,t$-separators, when $A$ is clear from the context:
\begin{equation}
	\label{eq:FsAt}
	\safeSeps{s}{t}(G,A)\eqdef \set{S\in \minlsepst{G}:A\subseteq C_s(G\sminus S)}.
\end{equation}
We denote by $\safeSepsk{s}{t}{k}(G,A)$ the minimal CP $s,t$-separators with respect to $A$, whose cardinality is at most $k$.
We denote by $f_{s,t}(G,A)$ the size of minimum CP $s,t$-separators:
\begin{equation}
	\label{eq:fsAt}
	f_{s,t}(G,A)\eqdef \min\set{\lvert S \rvert: S \in \safeSeps{s}{t}(G,A)}.
\end{equation}
\rev{We set $f_{s,t}(G,A)=+\infty$ if this family is empty.} We denote by $\safeSepsMin{s}{t}(G,A)$ the minimum cardinality CP $s,t$-separators (i.e., whose cardinality is $f_{s,t}(G,A)$). The following is immediate from~\eqref{eq:FsAt} and~\eqref{eq:fsAt}.
\def\simpleProp{
		 (i) $\safeSeps{s}{t}(G,A)\subseteq \minlsep{sA,t}{G}$, and (ii) $\kappa_{sA,t}(G) \leq f_{s,t}(G,A)$.
}
\begin{proposition}
	\label{prop:simpleProp}
	\simpleProp
\end{proposition}\eat{
\begin{proof}
	By Definition (see~\eqref{eq:FsAt}), $\safeSeps{s}{t}(G,A)\subseteq \minlsep{sA,t}{G}$. Item 2 is immediately follows.
\end{proof}}
\begin{definition}
	\label{def:safeImp}
We say that  $S\in \safeSeps{s}{t}(G,A)$ is \e{important} if, for every $S'\in \safeSeps{s}{t}(G,A)$:
	\begin{align*} 
		C_{s}(G\sminus S') \subsetneq C_{s}(G\sminus S)&&\Longrightarrow&& |S'|>|S|.
	\end{align*}
\end{definition}
We denote by $\safeSepsImp{s}{t}(G,A)$ the important CP $s,t$-separators, and by $\safeSepskImp{s}{t}{k}(G,A)$ the important CP $s,t$-separators of cardinality at most $k$.
We write $\safeSepskImp{s}{t}{\min}(G,A)=\safeSepsImp{s}{t}(G,A)\cap\safeSepsMin{s}{t}(G,A)$ for those of minimum CP cardinality.

\begin{revision}
\textbf{A small example.}
Take the two paths $s,u_1,v_1,t$ and $s,u_2,v_2,t$, and add a vertex $a$ adjacent to $u_1$ and $u_2$. Require $a$ to remain connected to $s$. The ordinary important minimum separator $\{u_1,u_2\}$ isolates $a$ from $s$. There are instead two CP-important minimum separators, $\{v_1,u_2\}$ and $\{u_1,v_2\}$, with incomparable source components $\{s,a,u_1\}$ and $\{s,a,u_2\}$. The violating component $\{a\}$ has boundary $\{u_1,u_2\}$: keeping either boundary vertex on the source side repairs its connectivity. The next section proves when such boundary choices suffice in general.
\end{revision}

\section{Connectivity-Preserving Important Minimum Separators}
\label{sec:cp_min_seps}

The example above shows why the unique unconstrained important minimum
separator of Lemma~\ref{lem:uniqueMinImportantSep} need not preserve connectivity. We repair its disconnected
terminal components by choosing vertices from their boundaries.
Theorem~\ref{thm:minsepsAt} supplies the key guarantee: vertices chosen
from minimum separators can be forced into the actual source component.
The resulting hitting-set characterization appears in
Theorem~\ref{thm:mainThmBoundedCardinality}; full proofs are in
Appendix~\ref{sec:mainThmBoundedCardinalityAppendix}.

\subsection{Structural Properties of Minimum Separators}

\clarityadd{A separator between a source set and $t$ leaves every source vertex
undeleted and separated from $t$, but these vertices may lie in different
components. The next theorem gives actual connectivity for a particular
choice of source vertices: select vertices that each belong to some
minimum $s,t$-separator, and add them to the source set containing $s$.
Every minimum separator for this enlarged source set keeps all selected
vertices connected to $s$. The selected vertices may belong to different minimum
$s,t$-separators, and the guarantee remains valid after adding arbitrary
further vertices to the source set. This is the property that will
justify our boundary choices for repairing disconnected terminal
components.}

\def\minsepsAtThm{
	Let $D \subseteq\minstVertices{s,t}(G)$, and let $X\subseteq \nodes(G){\setminus}\set{s,t}$. For every $T \in \minsep_{sDX,t}(G)$ it holds that $D \subseteq C_s(G\sminus T)$.
}
\begin{theorem}
	\label{thm:minsepsAt}
\minsepsAtThm
\end{theorem}

\begin{proof}[Proof Idea]
		The proof proceeds by contradiction. Let $T \in \minsep_{sDX,t}(G)$. Define $T' \eqdef N_G(C_s(G\sminus T))$ (note that $T'\subseteq T$). We define the induced graph $G' \eqdef G[C_s(G\sminus T) \cup T' \cup C_t(G\sminus T)]$. The main part of the proof establishes that $G'$ has the same $s,t$-connectivity as $G$ (i.e., $\kappa_{s,t}(G')=\kappa_{s,t}(G)$).
\rev{To that end, we first show that every minimal $s,t$-separator of $G'$ contained in
$T'\cup C_s(G'\sminus T')$ is also a minimal $s,t$-separator of $G$.
Second, we show that every minimal $s,t$-separator of $G'$ contained in $T'\cup C_t(G'\sminus T')$ has size at least $|T'|$:
a smaller one could replace $T'$ in $T$, contradicting the minimum size of
$T$. The remaining case is that a minimal $s,t$-separator $S$ of $G'$
may have vertices on both sides of $T'$ (i.e., $C_s(G'\sminus T')$ and $C_t(G'\sminus T')$), so neither observation applies
to $S$ directly. We overcome this by comparing the vertices reachable
from $s$ in $G'\sminus S$ (i.e., $C_s(G'\sminus S)$) with those reachable after deleting $T'$ (i.e., $C_s(G'\sminus T')$).
Consider the boundaries of their intersection and union:
\[
B_1\eqdef N_{G'}(C_s(G'\sminus S)\cap C_s(G'\sminus T')),
\qquad
B_2\eqdef N_{G'}(C_s(G'\sminus S)\cup C_s(G'\sminus T')).
\]
Both boundaries separate $s$ from $t$:
every $s,t$-path must leave the corresponding set through its boundary.
Neither boundary contains $s$ or $t$.

The first boundary $B_1$ lies in $ C_s(G'\sminus T')\cup T'$, whereas the second lies in
$C_t(G'\sminus T')\cup T'$. We can therefore apply the two observations
above, taking a minimal separator contained in each boundary if
necessary. The first boundary $B_1$ has at least $\kappa_{s,t}(G)$ vertices.
The second boundary $B_2$ has at least $|T'|$ vertices; otherwise, a smaller separator
contained in it could replace $T'$ in $T$, contradicting the minimum
size of $T$.

Applying Corollary~\ref{corr:submodularity} to the minimal $s,t$-separators $S$ and $T'$ in $G'$, together with the lower bounds previously established $|B_1|\geq \kappa_{s,t}(G)$, and $|B_2|\geq |T'|$ gives
\[
\kappa_{s,t}(G)+|T'|
\leq |B_1|+|B_2|
\underbrace{\leq}_{\text{Submodularity (Cor.~\ref{corr:submodularity})}} |S|+|T'|.
\]
Cancelling $|T'|$ gives $|S|\geq\kappa_{s,t}(G)$.
Since $S$ is an arbitrary minimal $s,t$-separator of $G'$, it means that discarding the vertices outside $G'$ has not made $s$ and $t$
easier to separate; that is, $\kappa_{s,t}(G')\geq \kappa_{s,t}(G)$. Conversely, intersecting any minimum $s,t$-separator
of $G$ with $V(G')$ gives an $s,t$-separator of $G'$ of no greater size.
Hence $\kappa_{s,t}(G')=\kappa_{s,t}(G)$.

We now use the defining property of $D$: every vertex of $D$ belongs
to some minimum $s,t$-separator of $G$. Suppose that a vertex
$v\in D$ does not belong to $C_s(G\sminus T)$. Since $T$ separates
$sDX$ from $t$ without deleting any vertex of $D$, the vertex $v$
belongs neither to $T$ nor to $C_t(G\sminus T)$. Therefore $v$ was
discarded when forming $G'$.

Choose a minimum $s,t$-separator $W$ of $G$ containing $v$.
The vertices of $W$ that remain in $G'$ still separate $s$ from $t$ in $G'$:
every $s,t$-path in $G'$ is also such a path in $G$ and must meet $W$.
But $v$ is missing, so
\[
|W\cap V(G')|
\leq |W|-1
=\kappa_{s,t}(G)-1.
\]
This contradicts the fact that $G'$ still requires
$\kappa_{s,t}(G)$ vertices to separate $s$ from $t$.
Thus no vertex of $D$ can lie outside $C_s(G\sminus T)$,
as required.}
	\end{proof}

\def\minsepsAtSetCorr{
	Let $A\subseteq \nodes(G){\setminus}\set{s,t}$, $D \subseteq\minstVertices{sA,t}(G)$, and $X\subseteq \nodes(G){\setminus}\set{s,t}$. For every $T \in \minsep_{sADX,t}(G)$ it holds that $D \subseteq C_{sA}(G\sminus T)$.
}

\begin{corollary}
	\label{corr:minsepsAtSets}
	\minsepsAtSetCorr
\end{corollary}
Corollary~\ref{corr:minsepsAtSets} is the extension of Theorem~\ref{thm:minsepsAt} when the source side is a set (i.e., $sA$).

The following is an iterated form of Corollary~\ref{corr:minsepsAtSets} that is applied in the enumeration algorithm.\eat{; its role will become clear when we describe the enumeration algorithm.}
\def\XInvariant{
		Let $X_1,X_2,\dots,X_\ell \subseteq \nodes(G)$ be nonempty and pairwise disjoint. Let $\calX_0 \eqdef \emptyset$, and for every $i\in \set{1,\dots,\ell}$ let $\calX_i\eqdef X_1\cup X_2\cup \cdots \cup X_i$. If, for all $i\in \set{1,2,\dots,\ell}$, it holds that $X_i\subseteq \minstVertices{s\calX_{i-1},t}(G)$, then for every $A\subseteq \nodes(G)$, it holds that: if $T\in \minsep_{sA\calX_\ell,t}(G)$ then $\calX_\ell\subseteq C_s(G\sminus T)$.
}
\begin{corollary}
	\label{corr:XInvariant}
	\XInvariant
\end{corollary}
\eat{
\begin{proof}
	Fix $A_0\subseteq V(G)$, and let $T\in \minsep_{sA_0\calX_\ell,t}(G)$. We show that $X_i \subseteq C_s(G\sminus T)$ by induction on $i$.
	
	Base: $i=1$. Since $X_1 \subseteq \minstVertices{s,t}(G)$, then by Theorem~\ref{thm:minsepsAt}, we have that $X_1\subseteq C_s(G\sminus T)$ (i.e., take $D\defeq X_1$ and
	$X\defeq A_0\cup (\calX_\ell{\setminus} X_1)$ in Theorem~\ref{thm:minsepsAt}).
	
	Step: Let $1\leq j\leq \ell-1$, and assume the claim holds for all indices $i\leq j$, we prove for $i=j+1$.
	By the induction hypothesis, we have that $X_k\subseteq C_s(G\sminus T)$ for every $k\leq j$. In particular,
	\begin{equation}
		\label{eq:hypothesis}
		\calX_{i-1}=\bigcup_{k=1}^{j=i-1}X_k \subseteq C_s(G\sminus T).
	\end{equation}
	Since $X_i \subseteq \minstVertices{s\calX_{i-1},t}(G)$, we apply
	Corollary~\ref{corr:minsepsAtSets} with
	\[
	A_{\mathrm{cor}} \defeq \calX_{i-1},\qquad
	D_{\mathrm{cor}} \defeq X_i,\qquad
	X_{\mathrm{cor}} \defeq A_0\cup (\calX_\ell{\setminus} \calX_i).
	\]
	Indeed,
	\[
	sA_{\mathrm{cor}}D_{\mathrm{cor}}X_{\mathrm{cor}}
	=
	s\calX_{i-1}X_iA_0(\calX_\ell{\setminus} \calX_i)
	=
	sA_0\calX_\ell,
	\]
	and therefore the hypothesis $T\in \minsep_{sA_0\calX_\ell,t}(G)$ is exactly the
	hypothesis needed to apply Corollary~~\ref{corr:minsepsAtSets}. Hence $X_i\subseteq C_{s\calX_{i-1}}(G\sminus T)$.
Since $\calX_{i-1}\subseteq C_s(G\sminus T)$ by the induction hypothesis, we have $C_{s\calX_{i-1}}(G\sminus T)=C_s(G\sminus T)$.
Therefore, $X_i\subseteq C_s(G\sminus T)$.
\end{proof}}

\subsection{Characterization of Minimum CP Important Separators}
\eat{
Assume that a CP separator exists at the unconstrained minimum size $\kappa_{s,t}(G)$.
In other words, suppose that $f_{s,t}(G,A)=\kappa_{s,t}(G)$.
The next theorem establishes that under this assumption, all minimum CP-important separators (i.e., of size $f_{s,t}(G,A)=\kappa_{s,t}(G)$)
can be enumerated in total time $O(2^{\kappa_{s,t}(G)}\cdot n\cdot T(n,m))$. 
Since we are interested in CP $s,t$-separators of size at most $k$, then we may assume that $\kappa_{s,t}(G)\leq k$; otherwise, no such CP separator exists. The enumeration procedure proceeds by
``repairing'' the components that the canonical separator $L^*_{sA,t}(G)$ (Lemma~\ref{lem:uniqueMinImportantSep}) leaves disconnected
from $s$. Each repair chooses a minimal hitting set of their boundaries.
}

When $f_{s,t}(G,A)=\kappa_{s,t}(G)$, the next theorem enumerates all
minimum CP-important separators by repairing the components that
$L^*_{sA,t}(G)$ (Lemma~\ref{lem:uniqueMinImportantSep}) leaves disconnected from $s$, using minimal hitting
sets of their boundaries.

\def\mainThmBoundedCardinality{
	Let $A\subseteq\nodes(G)\setminus\{t\}$ and suppose that $\kappa_{s,t}(G)=\kappa_{sA,t}(G)=f_{s,t}(G,A)<\infty$. Let $L^*\eqdef\closestMinSep{sA,t}(G)$, and define the set of violating components in $G\sminus L^*$:
	\begin{align*}
			\D\eqdef \set{C\in \cc(G\sminus L^*) : C\cap A \neq \emptyset, s\notin C} && \text{ and }&&\varepsilon_{\D}\eqdef \set{N_G(C) : C\in \D}.
	\end{align*}
	If $\D\neq\emptyset$, then:
	\[
	\emptyset \subset \safeSepskImp{s}{t}{\min}(G,A)\subseteq \mediumbigcup_{\substack{Y\in\MHS(\varepsilon_{\D}):\\\kappa_{sAY,t}(G)=\kappa_{s,t}(G)}}\set{\closestMinSep{sAY,t}(G)},
	\]
	and $\safeSepskImp{s}{t}{\min}(G,A)$ can be computed in time $O(2^{\kappa_{s,t}(G)}\cdot n\cdot T(n,m))$.
}
\begin{theorem}
	\label{thm:mainThmBoundedCardinality}
	\mainThmBoundedCardinality
\end{theorem}
\eat{
The theorem starts from the canonical unconstrained separator $L^*\eqdef \closestMinSep{sA,t}(G)$ which, by Lemma~\ref{lem:uniqueMinImportantSep}, exits, is unique and can be computed in time $O((m+n)\kappa_{s,t}(G))$. Then it identifies the connected components of $G\sminus L^*$ that witness violations of the constraints, and shows that every important minimum connectivity-preserving separator is obtained by repairing these violations through a minimal hitting set of their neighborhoods (i.e, $\MHS(\varepsilon_{\D})$). 
}
\begin{proof}[Proof Sketch]
	Let $L^* \defeq \closestMinSep{sA,t}(G)$ (Lemma~\ref{lem:uniqueMinImportantSep}),
	and let $\D$ be the set of violating components of $G\sminus L^*$. For every $C\in\D$ it holds, by definition, that $N_G(C)\subseteq L^*$. Let $\varepsilon_{\D}\defeq \set{N_G(C)\mid C\in\D}$.
	See Figure~\ref{fig:proofMainIllustration} for illustration that depicts the connected components of $\D$ that are exactly the obstructions to $L^*$ being connectivity-preserving.
	
	\begin{enumerate}
		\item \textbf{Characterizing feasibility among minimum separators.}
		We first show that if $S\in \minsep_{sA,t}(G)$, then $S\in \safeSepsMin{s}{t}(G,A)$ if and only if $C\subseteq C_s(G\sminus S)$ for every $C\in\D$.
		The proof of this equivalence is deferred to section~\ref{sec:thm_mainThmBoundedCardinalityProof} of the Appendix; it uses the minimality of the separators and the fact that the boundary of every violating component is contained in the canonical minimum separator $L^*$.

\item \textbf{Encoding the constraints by a hitting set of $\varepsilon_{\D}$.}
Let $S\in \rev{\safeSepskImp{s}{t}{\min}(G,A)}$, and let $X \defeq C_s(G\sminus S)\cap L^*$.
By the first step, for every $C\in\D$ we have $X\cap N_G(C)\neq\emptyset$.
Thus $X$ is a hitting set of $\varepsilon_{\D}$ (see Figure~\ref{fig:proofMainIllustration}). \rev{Choose an inclusionwise minimal hitting set $Y\subseteq X$; then $Y\in\MHS(\varepsilon_{\D})$.
Since $A\cup Y\subseteq C_s(G\sminus S)$, the separator $S$ is an $sAY,t$-separator of size $\kappa_{s,t}(G)$. Every $sAY,t$-separator also separates $s$ from $t$, so
$\kappa_{sAY,t}(G)=\kappa_{s,t}(G)$.}
		\item \textbf{Using Theorem~\ref{thm:minsepsAt}.}
		This is the key step. Every set in $\varepsilon_{\D}$ is a subset of $L^*$, and since $\kappa_{s,t}(G)=\kappa_{sA,t}(G)$, $L^*$ is a minimum $s,t$-separator. Therefore, for every $\rev{Y}\in \MHS(\varepsilon_{\D})$, it holds that $\rev{Y}\subseteq L^* \subseteq \minstVertices{s,t}(G)$.
		Theorem~\ref{thm:minsepsAt} then implies that $\rev{Y\subseteq C_s(G\sminus S_Y)}$ for every $\rev{S_Y\in \minsep_{sAY,t}(G)}$.
		\rev{For the set $Y$ chosen in the previous step, let $L_Y\eqdef\closestMinSep{sAY,t}(G)$. Then $|L_Y|=\kappa_{s,t}(G)$ and $Y\subseteq C_s(G\sminus L_Y)$.}
		Since $\rev{Y}$ is a hitting set of $\varepsilon_{\D}$, for every $C\in\D$ we have $\rev{N_G(C)\cap C_s(G\sminus L_Y)\neq\emptyset}$.
		\rev{Moreover, every violating component $C$ is disjoint from every minimum $s,t$-separator (Lemma~\ref{lem:minVerticesInMainGraph}), so $L_Y$ does not delete any vertex of $C$.} The same argument as in the proof of the first step now shows that $\rev{C\subseteq C_s(G\sminus L_Y)}$ for every $C\in\D$. Applying the converse direction of the first step, we conclude that $\rev{L_Y\in \safeSepsMin{s}{t}(G,A)}$.
		\rev{Since $L_Y$ is the unique important minimum $sAY,t$-separator and $S$ is also a minimum $sAY,t$-separator, $C_s(G\sminus L_Y)\subseteq C_s(G\sminus S)$. The importance of $S$ implies equality of these components, and minimality then gives $L_Y=S$.} Thus every important minimum connectivity-preserving separator is obtained in this way, and therefore
		\[
		\safeSepskImp{s}{t}{\min}(G,A)
		\subseteq
		\mediumbigcup_{\substack{\rev{Y}\in \MHS(\varepsilon_{\D}):\\
		\rev{\kappa_{sAY,t}(G)=\kappa_{s,t}(G)}}}
		\{\,\rev{\closestMinSep{sAY,t}(G)}\,\}.
		\]
		
		\item \textbf{Enumeration.}
Since every set in $\varepsilon_{\D}$ is a subset of $L^*$ and $|L^*|=\kappa_{s,t}(G)$, every minimal hitting set of $\varepsilon_{\D}$ is a subset of $L^*$, and therefore there are at most $2^{\kappa_{s,t}(G)}$ such sets.
		Enumerating all minimal hitting sets $\rev{Y}\in \MHS(\varepsilon_{\D})$
		and computing the corresponding closest minimum separator $\rev{\closestMinSep{sAY,t}(G)}$ (Lemma~\ref{lem:uniqueMinImportantSep}) \rev{whenever $\kappa_{sAY,t}(G)=\kappa_{s,t}(G)$ yields a family containing all important minimum connectivity-preserving separators. Removing repetitions and candidates whose source component properly contains that of another candidate gives exactly the important minimum connectivity-preserving separators. The filtering procedure and its running-time analysis are given in Appendix~\ref{sec:thm_mainThmBoundedCardinalityProof}. The total time remains}
	$O(2^{\kappa_{s,t}(G)}\cdot n\cdot T(n,m))$.
	\end{enumerate}
\end{proof}

\eat{
Corollary~\ref{corr:mainThmBoundedCardinalityExt} is the form of
Theorem~\ref{thm:mainThmBoundedCardinality} that is needed by the recursive
enumeration algorithm.  At that point in the recursion, the algorithm has
already forced a set $X$ to belong to the $s$-component. Therefore the relevant
minimum connectivity is not $\kappa_{s,t}(G)$, but rather
$\kappa_{sX,t}(G)$, and the connectivity constraint is imposed on $AX$ rather than $A$.
The corollary says that, as long as every minimum $sX,t$-separator maintains the connectivity of every vertex in $X$ to $s$, Theorem~\ref{thm:mainThmBoundedCardinality} can be
applied with $X$ absorbed into the source side. Consequently, the minimum
CP-important separators with respect to $AX$ are still obtained by repairing
the violating components of the unique minimum important separator
$L^*_{sAX,t}(G)$ through minimal hitting sets of their neighborhoods.
This is the exact structural step used by the algorithm when it reaches a recursive call in
which the premise $\kappa_{sX,t}(G)=\kappa_{sAX,t}(G)=f_{s,t}(G,AX)$
has been established and the unique minimum important separator $L^*_{sAX,t}(G)$ violates the connectivity constraint for $A$.}

\def\mainThmBoundedCardinalityExt{
	Let $X\subseteq \nodes(G){\setminus}\set{s,t}$ such that for every $T\in \minsep_{sX,t}(G)$ it holds that $X\subseteq C_s(G\sminus T)$, and let $A\subseteq\nodes(G){\setminus}\{t\}$. Suppose that $\kappa_{sX,t}(G)=\kappa_{sAX,t}(G)=f_{s,t}(G,AX)<\infty$. Let $L^*\eqdef\closestMinSep{sAX,t}(G)$, and define the set of violating components in $G\sminus L^*$:
\begin{align*}
	\D\eqdef \set{C\in \cc(G\sminus L^*) : C\cap A \neq \emptyset, s\notin C} && \text{ and }&&\varepsilon_{\D}\eqdef \set{N_G(C) : C\in \D}.
	\end{align*}
	If $\D\neq\emptyset$, then:
	\[
	\emptyset \subset \safeSepskImp{s}{t}{\min}(G,AX) \subseteq
	\mediumbigcup_{\substack{Y\in\MHS(\varepsilon_{\D}):\\\kappa_{sAXY,t}(G)=\kappa_{sX,t}(G)}}\set{L^*_{sAXY,t}(G)}
	\]
	and $\safeSepskImp{s}{t}{\min}(G,AX)$ can be computed in time $O(2^{\kappa_{sX,t}(G)}\cdot n\cdot T(n,m))$.
}

\begin{corollary}
	\label{corr:mainThmBoundedCardinalityExt}
	\mainThmBoundedCardinalityExt
\end{corollary}
The corollary applies the same repair to an additional set $X$ already
guaranteed to stay connected to $s$ under every minimum $sX,t$-separator.
Its proof is in Appendix~\ref{sec:mainThmBoundedCardinalityAppendix}.

\begin{figure}[ht]
	\centering
	
	\begin{tikzpicture}[scale=0.75, transform shape]
		\foreach \y/\name in {3/v1, 2/v2, 1/v3, 0/v4} {
			\node[inner sep=0pt, minimum size=2.2mm, circle, fill=black] (\name) at (0,\y) {};
		}
		
		
	\node[above=0.5pt of v1] {\scriptsize $v_1$};
	
	\node[above=0.5pt of v2] (v2lab) {\scriptsize $v_2$};
	
	\node[above=0.5pt of v3] (v3lab) {\scriptsize $v_3$};
	
	\node[above=0.5pt of v4] {\scriptsize $v_4$};
	
	
\begin{scope}[on background layer]
	\node[
	ellipse,
	fit={(-0.18,2.18) (0.18,1.00)},
	fill=yellow,
	fill opacity=0.35,
	draw=yellow!70!black
	] (Xellipse) {};
\end{scope}

\node[right=-1.7pt of Xellipse] {\small $\rev{Y}$};

		\node[draw, circle, minimum size=2cm] (Cs) at (-4,1.5) {};
		\node[draw, circle, minimum size=2cm] (Ct) at (4,1.5) {};
		\node[draw, circle, minimum size=0.8cm] (C1) at (1.3,5) {};
		\node[draw, circle, minimum size=0.8cm] (C2) at (-0.5,-1.5) {};
		\node[draw, circle, minimum size=0.8cm] (C3) at (-1.3,4.5) {};
		
		\node at (Cs) {\scriptsize $C_s(G\sminus L^*)$};
		\node at (Ct) {\scriptsize $C_t(G\sminus L^*)$};
		\node at (C1) {\scriptsize $C_1$};
		\node at (C2) {\scriptsize $C_2$};
		\node at (C3) {\scriptsize $C_3$};

		\foreach \i/\angle in {1/60, 2/30, 3/-30, 4/-60} {
			\path (Cs.center) ++(\angle:1cm) coordinate (p\i);
			\draw (p\i) -- (v\i);
		}
		
		\foreach \i/\angle in {1/120, 2/150, 3/210, 4/240} {
			\path (Ct.center) ++(\angle:1cm) coordinate (q\i);
			\draw (q\i) -- (v\i);
		}
		
		\path (C1.center) ++(220:0.4cm) coordinate (c1left);
		\path (C1.center) ++(-90:0.4cm) coordinate (c1bottom);
		\path (C1.center) ++(0:0.4cm) coordinate (c1right);
		\draw (c1left) -- (v1);
		\draw[bend left=25] (c1bottom) to (v2);
		
		\path (C2.center) ++(90:0.4cm) coordinate (c2top);
		\path (C2.center) ++(60:0.4cm) coordinate (c2topright);
		\draw (c2top) -- (v4);
		\draw[bend right=36] (c2topright) to (v2);
		
			\path (C3.center) ++(180:0.4cm) coordinate (c3left);
	\path (C3.center) ++(-90:0.4cm) coordinate (c3bottom);
	\draw (c3bottom) -- (v1);
	\draw[bend right=25] (c3left) to (v3);

	\end{tikzpicture}
	\caption{Illustration for the proof of Theorem~\ref{thm:mainThmBoundedCardinality}. The vertices of $L^*$ are represented by black solid vertices. The connected components $C_1,C_2,C_3\in \cc(G\sminus L^*)$ are those that witness the connectivity constraints violation (i.e., $\D=\set{C_1,C_2,C_3}$). Note that $N_G(C_1)=\set{v_1,v_2}$, $N_G(C_2)=\set{v_2,v_4}$, and $N_G(C_3)=\set{v_1,v_3}$. Therefore, $\rev{Y}=\set{v_2,v_3}$ is a minimal hitting set of $\varepsilon_{\D}$.	\label{fig:proofMainIllustration}}
\end{figure}

\section{Enumerating Connectivity-Preserving Important Separators}
\label{sec:algo}
We now describe the recursive enumeration algorithm
$\algname{Gen{-}Seps}$. The complete pseudocode appears as
Algorithm~\ref{alg2e:GenSeps} in
Appendix~\ref{sec:AlgForImpSafeSepsAppendix}
(page~\pageref{alg2e:GenSeps}), which also contains the full
correctness and running-time analysis. A recursive call has the form
$\algname{Gen{-}Seps}(G,s,t,A,X,Z,k)$.
The graph $G$, the terminals $s,t$, and the integer $k$ describe the current
residual instance. The set $Z$ contains vertices that have already been
committed to the separator along the current recursion path; these vertices are
appended to the separators returned by the current call. The set $X$ records
vertices that have been added to the connectivity requirement. Thus the current
call works with the terminal set $AX$, and the separators considered in the
call are members of $\safeSepskImp{s}{t}{k}(G,AX)$. In particular, if
$S\in\safeSepskImp{s}{t}{k}(G,AX)$, then $AX\subseteq C_s(G\sminus S)$.

For every recursive call $\algname{Gen{-}Seps}(G,s,t,A,X,Z,k)$, the algorithm
maintains the invariant
\[
\tag{Inv}
\text{for every }Q\subseteq\nodes(G)
\text{ and every }T\in\minsep_{sXQ,t}(G),\qquad
X\subseteq C_s(G\sminus T).
\]
This invariant is what allows Corollary~\ref{corr:mainThmBoundedCardinalityExt}
to be applied after vertices have been added to $X$. The algorithm first
establishes the premise
$\kappa_{sX,t}(G)=\kappa_{sAX,t}(G)=f_{s,t}(G,AX)<\infty$. Once this premise holds, it
computes the unique important minimum $sAX,t$-separator
$L^*\eqdef\closestMinSep{sAX,t}(G)$ and distinguishes two cases: whether or not $L^*$ is
CP with respect to $AX$ (i.e., $AX\subseteq C_s(G\sminus L^*)$). If not, Corollary~\ref{corr:mainThmBoundedCardinalityExt} is used to find
a set $Y$ that is added to $X$ before the branching step is applied.
\eat{
 If $L^*$ is
CP with respect to $AX$ (i.e., $AX\subseteq C_s(G\sminus L^*)$), the algorithm handles the call by 
a branching step over $L^*$ that is similar to the standard branching step in the algorithm for enumerating (standard) important separators. If $L^*$ violates the connectivity
requirement, Corollary~\ref{corr:mainThmBoundedCardinalityExt} is used to find
a set $Y$ that is added to $X$ before the branching step is applied.}

\rev{We allow disconnected graphs. For a fixed source $s$, vertices outside $C_s(G)$ can be discarded after checking that $A\subseteq C_s(G)$; otherwise no CP separator exists. If $t\notin C_s(G)$, the empty set is the only minimal $s,t$-separator. Thus the nontrivial initial enumeration instance can be assumed connected; recursive instances need not be. Unless stated otherwise, $s\neq t$ and $k\geq0$. An infeasible minimum has value $+\infty$ and no canonical separator; such a branch is rejected before using one.} 
Progress is measured by the potential
\[
\lambda(G,A,X,k)\eqdef
(k+1)\bigl(2k-\kappa_{sAX,t}(G)\bigr)+
\bigl(k-\kappa_{sX,t}(G)\bigr).
\]
The potential bounds the depth of the recursion tree~\cite{DBLP:journals/jacm/FominGK09}. 
\begin{revision}
We show that the potential never increases from a parent call to a
child call when both pass the initial feasibility tests
($\kappa_{sX,t}(G)\le k$ and $\kappa_{sAX,t}(G)\le k$).
Along any recursion path, a transition preserving $\lambda$ must be
followed
by a transition that strictly decreases $\lambda$.
Moreover, $\lambda$ never increases between such calls.
The tests ensure $0\le\lambda(G,A,X,k)\le 2k^2+3k$.
Since $\lambda$ is integer-valued, a root-to-leaf path can contain
only $O(k^2)$ strict decreases. Hence the recursion height is
$O(k^2)$, including a possible final call that fails a test.
\end{revision}

\paragraph*{Phase 1: Establishing the Premise of
	Corollary~\ref{corr:mainThmBoundedCardinalityExt}.}
The first phase ensures that the current recursive call reaches a state in
which Corollary~\ref{corr:mainThmBoundedCardinalityExt} can be applied. In a
call $\algname{Gen{-}Seps}(G,s,t,A,X,Z,k)$, the relevant connectivity
requirement is $AX$, and the premise needed is
$\kappa_{sX,t}(G)=\kappa_{sAX,t}(G)=f_{s,t}(G,AX)$.
To test this premise, the algorithm computes
$L^t\eqdef L^t_{sX,t}(G)$, the unique minimum $sX,t$-separator closest to $t$.
By (Inv), every minimum $sX,t$-separator $T$ satisfies
$X\subseteq C_s(G\sminus T)$; in particular,
$X\subseteq C_s(G\sminus L^t)$. Thus the only possible obstruction to $L^t$
being CP with respect to $AX$ is the existence of a vertex of $A$ that is not
contained in $C_s(G\sminus L^t)$.

If $A\subseteq C_s(G\sminus L^t)$, then $L^t$ is CP with respect to $AX$.
Since $L^t$ is a minimum $sX,t$-separator, we have
$f_{s,t}(G,AX)\le |L^t|=\kappa_{sX,t}(G)$. The reverse inequality follows from
Proposition~\ref{prop:simpleProp}, and therefore
$\kappa_{sX,t}(G)=\kappa_{sAX,t}(G)=f_{s,t}(G,AX)$. The algorithm then proceeds
to Phase~2.

If $A\not\subseteq C_s(G\sminus L^t)$, $L^t$ violates the connectivity requirement for $AX$. Then no
minimum $sX,t$-separator is CP with respect to $AX$, and hence
$\kappa_{sX,t}(G)<f_{s,t}(G,AX)$. Lemma~\ref{lem:partZeroOfAlg} shows that
finding $\safeSepskImp{s}{t}{k}(G,AX)$ reduces to finding
$\safeSepskImp{s}{t}{k}(G,AXy)$ for every $y\in L^t$. The algorithm therefore
branches over the vertices $y\in L^t$. In the branch indexed by $y$, subsequent
separators are considered with the additional requirement
$y\in C_s(\cdot)$, i.e., with $X$ replaced by $Xy$.
This update preserves (Inv). Indeed, since
$y\in L^t\subseteq U^{\min}_{sX,t}(G)$, Corollary~\ref{corr:XInvariant}
implies that after adding $y$ to $X$, the invariant (Inv) continues to hold. It also
gives progress: Lemma~\ref{lem:potentialReducedPreFirstPart} shows that
$\kappa_{sXy,t}(G)>\kappa_{sX,t}(G)$, and therefore the potential $\lambda$ strictly
decreases. Thus each iteration of Phase~1 either establishes the premise of
Corollary~\ref{corr:mainThmBoundedCardinalityExt}, or creates at most
$|L^t|\le k$ recursive branches, each with smaller potential and with (Inv)
preserved.

\paragraph*{Phase 2: The Premise Verified.}
At this point $L^t\eqdef L^t_{sX,t}(G)$ is CP with respect to $AX$, and the
premise of Corollary~\ref{corr:mainThmBoundedCardinalityExt} holds (i.e., $\kappa_{sX,t}(G)=\kappa_{sAX,t}(G)=f_{s,t}(G,AX)$). The
algorithm computes the unique important minimum $sAX,t$-separator
$L^*\eqdef\closestMinSep{sAX,t}(G)$; that is, $|L^*|=\kappa_{sAX,t}(G)$. If $L^*$ is CP with respect to $AX$, then $L^*$ is the canonical minimum
CP-important separator for the current call.
\rev{One recursive call commits all of $L^*$ to the separator and
immediately outputs $Z\cup L^*$. For each $v\in L^*$, a further call
adds edges from $t$ to $N_G[v]\sminus\{t\}$, excluding $v$ from the
$s$-component. Thus this phase creates $1+|L^*|\le k+1$ children.}

\clarityadd{In this case, the calls cover every CP-important separator of size at most $k$
with respect to $AX$.
To see this, let $S$ be such a separator with $S\neq L^*$.
If every vertex of $L^*$ belonged to $S\cup C_s(G\sminus S)$,
Lemma~\ref{lem:inclusionCsCt} would give
$C_s(G\sminus L^*)\subseteq C_s(G\sminus S)$.
Since $L^*$ is CP and $|L^*|=f_{s,t}(G,AX)=\kappa_{sAX,t}(G) \leq |S|$, importance of $S$ would
force equality of these components. Their boundaries determine the
minimal separators (Lemma~\ref{lem:fullComponents}), so this would
give $L^*=S$, a contradiction.
Hence some $v\in L^*$ survives outside $C_s(G\sminus S)$.
No neighbor of this surviving vertex can lie in the $s$-component,
so adding edges from $t$ to $N_G[v]\sminus\{t\}$ leaves that
component unchanged. Lemma~\ref{corr:MinlsBtSepSafe} proves that
$S$ also remains minimal and important in the modified graph with the added edges.
Thus $S$ is covered by the branch indexed by $v$.
The complete argument is given in Lemma~\ref{lem:partTwoOfAlg}
and Corollary~\ref{corr:partTwoOfAlg}.}

If $L^*$ is not CP with respect to $AX$, then the algorithm proceeds to
Phase~3.

\paragraph*{Phase 3: Handling Constraint Violations
	(Corollary~\ref{corr:mainThmBoundedCardinalityExt}).}
Assume that $L^*$ violates the connectivity requirement, i.e.,
$AX\not\subseteq C_s(G\sminus L^*)$. Since Phase~1 has already established
$\kappa_{sX,t}(G)=\kappa_{sAX,t}(G)=f_{s,t}(G,AX)$, the hypotheses of
Corollary~\ref{corr:mainThmBoundedCardinalityExt} are satisfied in the current
instance.

Let $\D$ be the connected components of $G\sminus L^*$ that contain a vertex of
$A$ but do not contain $s$, and let
$\varepsilon_{\D}\eqdef\set{N_G(C):C\in\D}$. Corollary~\ref{corr:mainThmBoundedCardinalityExt}
guarantees that there exists a minimal hitting set
$Y\in\MHS(\varepsilon_{\D})$ such that
$\closestMinSep{sAXY,t}(G)\in\safeSepskImp{s}{t}{\min}(G,AX)$. Since
$Y\subseteq L^*$ and $|L^*|\le k$, the algorithm can find such a set $Y$ by
enumerating the minimal hitting sets of $\varepsilon_{\D}$ in time $O^*(2^k)$.
This enumeration finds a suitable set $Y$; the recursion does
not branch over all minimal hitting sets.
\rev{After such a set $Y$ is found, the algorithm covers the remaining separators by one all-$Y$ branch enforcing the inclusion of $Y$ in the $s$-component and two branches for each $y\in Y$ as follows.} One recursive call replaces $X$ by $XY$, thereby
strengthening the connectivity requirement from $AX$ to $AXY$. Also, for each
$y\in Y$, the algorithm creates the following vertex-indexed calls: one call commits $y$ to the separator, and the other forces $y$ to be separated from $s$. Since $|Y|\le |L^*| \leq k$, this gives at most $2k+1$ recursive
calls. The call with $X$ replaced by $XY$ preserves (Inv) by
Corollary~\ref{corr:XInvariant}.

\clarityadd{In this case as well, the calls cover every CP-important separator of size at most $k$
with respect to $AX$. One chosen hitting set $Y$ suffices for this because every CP-important
separator $S$ of size at most $k$ with respect to $AX$ falls into one of the following
cases. If $Y\subseteq C_s(G\sminus S)$, then $S$ also preserves
$AXY$ and remains important: strengthening the connectivity requirement
only removes possible competitors. The call with $X$ replaced by $XY$
therefore covers $S$.
Otherwise, some $y\in Y$ lies outside $C_s(G\sminus S)$.
If $y\in S$, the deletion call covers $S\sminus\{y\}$ with
budget $k-1$, with $y$ recorded in $Z$.
If $y\notin S$, it survives outside the $s$-component, so the
edge-insertion call covers $S$, as in Phase~2.
Lemma~\ref{lem:thirdPart} proves that the residual separators remain
CP-important in both cases, including their ordinary minimality;
the edge-insertion case uses Lemma~\ref{corr:MinlsBtSepSafe}.
Thus the all-$Y$ call and the two calls for each $y\in Y$ cover
every required separator using at most $2|Y|+1\leq 2k+1$ calls.}

\paragraph*{Why the minimum separator size increases and the potential decreases.}
In Phase~1, adding a vertex to $X$ does not reduce the remaining
budget $k$. Progress instead comes from increasing
$\kappa_{sX,t}(G)$. The choice of the minimum separator closest to
$t$ guarantees this increase, as follows.

Recall that $L^t=L^t_{sX,t}(G)$ and $\kappa_{sX,t}(G)\leq k$.
For every minimal $sX,t$-separator
$T$, Lemma~\ref{lem:sAtMinlSepLargerThanMinsAtSep_t} states that
\[
C_t(G\sminus L^t)\not\subseteq C_t(G\sminus T)
\quad\Longrightarrow\quad |T|>\kappa_{sX,t}(G).
\]

\begin{proof}[Proof sketch]
	Put $R=C_t(G\sminus L^t)$ and $Q=C_t(G\sminus T)$.
	Suppose that $R\not\subseteq Q$ and $|T|\leq\kappa_{sX,t}(G)$.
	Then $|T|=|L^t|=\kappa_{sX,t}(G)$.
	Minimality gives $N_G(R)=L^t$ and $N_G(Q)=T$.
	The boundaries of $R\cap Q$ and $R\cup Q$ are $sX,t$-separators:
	both sets contain $t$ and have no vertices or neighbours in $sX$.
	Hence both boundaries have size at least $\kappa_{sX,t}(G)$.
	By Lemma~\ref{lem:submodularity},
	\[
	|N_G(R\cap Q)|+|N_G(R\cup Q)|
	\le |L^t|+|T|=2\kappa_{sX,t}(G).
	\]
	Therefore $N_G(R\cap Q)$ is a minimum separator. Its $t$-component
	is contained in $R\cap Q\subsetneq R$, contradicting the choice
	of $L^t$ as closest to $t$.
\end{proof}

Now fix $y\in L^t$. If no $sXy,t$-separator exists, its minimum
size is $+\infty$ and the strict increase is immediate. Otherwise,
suppose that a minimum $sXy,t$-separator $T$
had size $\kappa_{sX,t}(G)$; it would also be a minimum
$sX,t$-separator. By minimality of $L^t$, the vertex $y$ has a
neighbour in $C_t(G\sminus L^t)$. This neighbour cannot belong to
$C_t(G\sminus T)$, since $y$ survives and is separated from $t$.
By Lemma~\ref{lem:sAtMinlSepLargerThanMinsAtSep_t} described above, $|T|> \kappa_{sX,t}(G)$ contradicting the size of $T$.
As enlarging the source set cannot decrease the minimum size, we obtain
\[
\kappa_{sXy,t}(G)>\kappa_{sX,t}(G).
\]
This proves the strict increase used in Phase~1.
The argument above summarizes the proofs of
Lemmas~\ref{lem:sAtMinlSepLargerThanMinsAtSep_t}
and~\ref{lem:potentialReducedPreFirstPart}
in Appendix~\ref{sec:AlgForImpSafeSepsAppendix}.
For the edge-insertion branches of Phases~2 and~3,
the appendix uses Lemma~\ref{lem:sAtMinlSepLargerThanMinsAtSep_s}
to prove that $\kappa_{sAX,t}(G)$ increases.
We next show how these increases, together with reductions
in the deletion budget, control the potential function.

\begin{revision}
\paragraph*{Progress of the recursion.}
Let $a\eqdef\kappa_{sAX,t}(G)$ and $b\eqdef\kappa_{sX,t}(G)$.
We compare calls that pass the initial feasibility tests $(a,b\leq k)$, and write
$\lambda'$ for the child's potential. As previously shown, in Phase~1, $b$ strictly increases
and $a$ does not decrease. Phase 2 creates a branch that excludes a vertex from the $s$-component. In this case, $a$ increases by at least one
(Lemma~\ref{lem:sAtMinlSepLargerThanMinsAtSep_s}). Such a branch also replaces
$A$ by $AX$ and resets $X$ to $\emptyset$, so $b$ may decrease.
This explains the coefficient $k+1$ in $\lambda(G,A,X,k)$: the first term of $\lambda$
decreases by at least $k+1$ (i.e., because $a$ increased by at least $1$), whereas the second can increase by at most $k$.

For a deletion branch in Phase~3, we have $a=b$ and the new budget is
$k-1$. Deleting one vertex reduces $a$ by at most one, giving
\[
\lambda'\le k(2k-a-1)+(k-1),\qquad
\lambda-\lambda'\ge 3k+1-2a\ge k+1>0.
\]
The branch deleting all of $L^*$ in Phase~2 immediately outputs a separator.

Only the all-$Y$ branch enforcing $Y$ to belong to the $s$-component may preserve $\lambda$. Enlarging $X$ without
changing $G,A,k$ cannot increase the potential. Let
$L_Y\eqdef\closestMinSep{sAXY,t}(G)$ be the selected separator.
It has size $b$ and is CP with respect to $AX$. The preserved invariant
(Inv) also places $Y$ in the $s$-component, so
$AXY\subseteq C_s(G\sminus L_Y)$. Monotonicity now gives
$\kappa_{sXY,t}(G)=b$, and $L_Y$ is a minimum $sXY,t$-separator.
The minimum $sXY,t$-separator closest to $t$ has an $s$-component
containing $C_s(G\sminus L_Y)$, and is therefore CP with respect to $AXY$.
Thus the next call passes Phase~1 and reaches Phase~2 with canonical
separator $L_Y$, which is CP. If the recursion continues, its next step
strictly decreases the potential. Consequently, two steps preserving
$\lambda$ cannot be consecutive. The $O(k^2)$ height and at most
$2k+1$ children per call give $2^{O(k^2\log k)}$ recursive calls.

\paragraph*{Exact enumeration.}
The recursion generates a family $F$ of $s,t$-separators containing every CP-important separator, and where $|F|\in 2^{O(k^2\log k)}$.
All final tests use the original graph, terminals, and budget. Discard from $F$
each candidate $S$ that exceeds the budget, intersects $sA$, or fails $A\subseteq C_s(G\sminus S)$.
Also require ordinary inclusionwise minimality, tested by
$N_G(C_s(G\sminus S))=N_G(C_t(G\sminus S))=S$
(Lemma~\ref{lem:fullComponents}). Remove duplicates and store
$R(S)\eqdef C_s(G\sminus S)$ for each remaining candidate. Discard
$S$ if another remaining candidate $T$ satisfies $|T|\le |S|$ and
$R(T)\subsetneq R(S)$.
Every important separator survives these tests. For $F$ generated candidates, counting repetitions, feasibility and
minimality take $O(F(n+m))$ time, and duplicate removal and dominance
testing take $O(F^2n)$ time. Since $F=2^{O(k^2\log k)}$, the
filtering fits the bound of Theorem~\ref{thm:singleExponentialSafeImportantOfSizek}.
\end{revision}

\section{Applications}
\label{sec:applications}
\def\F{\mathcal{F}}

\apprev{We first prove the source-component optimization guarantee of
Corollary~\ref{cor:targetedCPIsolation}. We then apply CP enumeration to
solve \textsc{Mixed Multiway Cut-Uncut}
and \textsc{Annotated OCT with terminal connectivity}. In each application,
enumeration identifies connected sets to preserve while an existing
algorithm solves the remaining deletion problem.}

\subsection{Proof of Corollary~\ref{cor:targetedCPIsolation}}

Enumerate $\safeSepskImp{s}{t}{k}(G,A)$, keep the separators $S$ for
which $B\cap C_s(G\sminus S)=\emptyset$, and return one of minimum
cardinality. If none remains, report infeasibility.

For correctness, we prove that an optimum feasible separator may be
assumed to be CP-important.
Start with an optimum feasible separator $S$.
If $S$ is not CP-important, replace it by a CP separator of no greater
size whose source component is properly contained in $C_s(G\sminus S)$.
This preserves the inclusion of $A$ and exclusion of $B$.
Repeated replacement must stop because the source component strictly
shrinks. The final separator is CP-important and still optimum, so it
appears in the enumeration and survives the test.
The running time is the bound of
Theorem~\ref{thm:singleExponentialSafeImportantOfSizek}.
Appendix~\ref{sec:proofsFromNMWCUApplication} gives the argument.
The vertices of $B$ may be deleted or lie in any components outside
$C_s$; their mutual connectivity is unrestricted.
\begin{applicationrevision}
\subsection{Mixed Multiway Cut-Uncut}
\label{sec:mixed-two-sets}

We explain Theorem~\ref{thm:mixed-two-sets} first for two classes.
Let $A,B$ be the two terminal
sets, let $k,\ell$ be the vertex and edge budgets, and $p=k+\ell$.
Subdivide each edge $e$ of $G$ by a new vertex $z_e$, obtaining a graph $H$.
Deleting an original vertex uses the vertex budget; deleting $z_e$ uses the
edge budget. Connectivity between original vertices is preserved
by this representation.

Choose $s\in A$ and $t\in B$. For completeness, consider the vertex set in $H$
representing a feasible mixed solution, and shrink it to an ordinary inclusion-minimal
$s,t$-separator $S$. Restoring vertices preserves the connectivity of each
terminal group, and taking a subset respects both budgets. Thus $A$ lies in
$C_s(H\sminus S)$ and $B$ lies in $C_t(H\sminus S)$.
CP domination in the two orientations supplies CP-important separators
$Q_A,Q_B$, each of size at most $p$, whose source components satisfy
\[
A\subseteq R_A\subseteq C_s(H\sminus S),
\qquad
B\subseteq R_B\subseteq C_t(H\sminus S).
\]
In particular, $R_A,R_B$ are disjoint connected sets avoiding $S$.

The algorithm enumerates both CP-important families and tries every pair
of disjoint source components. For each pair, it forbids deleting the
vertices of $R_A\cup R_B$ and seeks an $s,t$-separator respecting the two
budgets and any additional deletion prohibitions. Vertex splitting reduces this last task to
\textsc{Multi-budgeted Directed Cut}, solvable in
$2^{O(p^2)}(n+m+1)^{O(1)}$ time~\cite[Theorem~1]{MultiBudgetedCuts2020}.
The protected sets connect their terminal groups in every returned solution;
for the pair above, $S$ witnesses feasibility. The boundaries $Q_A,Q_B$
are not deleted or charged to either budget, so domination need only control
their total cardinality. The two CP families have
$2^{O(p^2\log p)}$ members each, giving the bound for two classes.

For $r$ classes, the same idea requires one connected set $R_i$ containing
each $A_i$. To find it, add a temporary target adjacent to all terminals
outside $A_i$ and enumerate CP-important separators from a representative
of $A_i$ to that target. Whenever a solution exists, the appendix supplies
a pairwise disjoint tuple whose sets all avoid the same feasible deletion set. Protecting all
$R_i$ then reduces the problem to separating their representatives with
the two budgets. A recursion using multi-budgeted important cuts
\cite[Theorem~7 and Section~4.1, Lemma~2]{MultiBudgetedCuts2020} handles this
remaining task in $2^{O(rp^2\log p)}(n+m+1)^{O(1)}$ time.
Trying the $2^{O(rp^2\log p)}$ tuples gives the same total bound.
Appendix~\ref{sec:mixed-two-sets-appendix} gives one proof for all $r$,
including the transformations, protection, and zero budgets.

\subsection{Annotated Odd Cycle Transversal with terminal connectivity}
\label{sec:annotated-oct-connectivity}

To prove Theorem~\ref{thm:annotated-oct-connectivity}, we seek a connected
set $U$ containing $A$ that some feasible solution leaves undeleted.
We then solve ordinary OCT while forbidding deletion of every vertex
of $U$, ensuring that all vertices of $A$ remain connected.
The difficulty is to find such a set without knowing the solution.

The colours $0,1$ name the two sides of a bipartition: every edge joins
vertices on different sides. Construct a graph $H$ with copies $v^0,v^1$
of each vertex $v$, where $v^i$ represents placing $v$ on side $i$.
For each edge $uv$ of $G$, add edges $u^0v^1,u^1v^0$ to $H$.
Choose $a_0\in A$, set $s=a_0^{\phi(a_0)}$ and
$t=a_0^{1-\phi(a_0)}$, and put
$A_\phi=\{a^{\phi(a)}:a\in A\}$.

Suppose a solution exists and choose an inclusion-minimal feasible
solution $X$ in $G$. Fix a proper colouring of $G\sminus X$ agreeing
with $\phi$. For each vertex $v$ in the component containing $A$,
choose the copy $v^i$ corresponding to its colour $i$.
Edges of that component join the chosen copies in $H$, so they form
a connected set $D$ containing $A_\phi$ and excluding $t$.
Its boundary $N_H(D)$ consists of copies of vertices of $X$ and hence
has size at most $2k$. Appendix~\ref{sec:annotated-oct-appendix} proves
the crucial further fact that $N_H(D)$ is an inclusion-minimal
$s,t$-separator in $H$.

Starting from this CP separator, repeatedly shrink its source component
while keeping $A_\phi$ connected to $s$ and without increasing the
separator size, until the separator is CP-important.
Thus our enumeration in $H$ with budget $2k$ includes a CP-important
$s,t$-separator $Q$ with respect to $A_\phi$ satisfying $A_\phi\subseteq C_s(H\sminus Q)\subseteq D$.
Shrinking the source component thus preserves the terminal connections
without introducing copies of vertices deleted by $X$.

For every enumerated CP important $s,t$-separator in $H$, replace the copies in its source
component by their original vertices, obtaining $U$.
Retain $U$ only if $G[U]$ is bipartite and has a proper colouring agreeing
with $\phi$ on $A$. Then solve ordinary OCT in $G$ with budget $k$,
forbidding deletion of every vertex of $U$.
For this separator $Q$, $U$ lies in the component of $G\sminus X$
containing $A$. Thus $U$ passes the colouring test, and $X$ remains an
allowed solution to this OCT call, so the call succeeds.
Conversely, any returned solution leaves $G[U]$ intact, keeping $A$
connected. Since connected $G[U]$ has a unique bipartition up to
exchanging its two sides, a proper colouring of the remaining graph
can be made to agree with $\phi$ by swapping the colours on the
component containing $U$, if necessary.

There are $2^{O(k^2\log k)}$ candidates, and each OCT call takes
$2^{O(k)}(n+m+1)^{O(1)}$ time~\cite{OCTLP2014}, giving the claimed bound.
Appendix~\ref{sec:annotated-oct-appendix} gives the full proof and explains
how to forbid vertex deletions when using the algorithm of~\cite{OCTLP2014}.
\end{applicationrevision}

\label{end:main-text}
\clearpage
\bibliography{main,motivation}
\clearpage
\appendix
\section{Proofs from Section~\ref{sec:minimalSeparators}}
\label{sec:minsepsvertexsets}
\eat{
Let $A$ and $B$ be two disjoint, non-adjacent subsets of $\nodes(G)$. A vertex-set $X\subseteq \nodes(G)$ is called an \e{$A,B$-separator} if, in the graph $G\sminus X$, there is no path between $A$ and $B$. We say that $X$ is a minimal $A,B$-separator if no proper subset of $X$ has this property. We say that a subset $X\subseteq \nodes(G)\setminus AB$ is a minimum $A,B$-separator if, for every $A,B$-separator $S$, it holds that $|X|\leq |S|$. We denote by $\minsep_{A,B}(G)$ the set of minimum $A,B$-separators, and by $\kappa_{A,B}(G)$ the size of a minimum $A,B$-separator.
We denote by $\minlsep{A,B}{G}$  the set of minimal $A,B$-separators in $G$. 
\eat{
We say that an $A,B$-separator $X$ is 
\e{safe} if there are two distinct, connected components $C_A,C_B\in \cc_G(X)$, where $A\subseteq C_A$ and $B\subseteq C_B$.}
In this short section, we establish that finding a minimal or minimum $A,B$-separator, where $A,B\subseteq \nodes(G)$ are disjoint and non-adjacent, can be reduced to finding a minimal or minimum $a,b$-separator in $G'$, where $a\in A,b\in B$, and $G'$ is the graph that results from merging $A$ to $a$ and $B$ to $b$ (see~\eqref{eq:mergeDef}).
}
\begin{replemma}{\ref{lem:inclusionCsCt}}
	\inclusionCsCt
\end{replemma}
\begin{proof}
	If $C_s(G\sminus S)\subseteq C_s(G\sminus T)$ then $N_G(C_s(G\sminus S))\subseteq C_s(G\sminus T)\cup N_G(C_s(G\sminus T))$. Since $S,T\in \minlsepst{G}$, then by Lemma~\ref{lem:fullComponents}, it holds that $S=N_G(C_s(G\sminus S))$ and $T=N_G(C_s(G\sminus T))$. Therefore, $S \subseteq C_s(G\sminus T)\cup T$. Hence, $C_s(G\sminus S)\subseteq C_s(G\sminus T) \Longrightarrow S \subseteq C_s(G\sminus T)\cup T$. 
	If $S \subseteq C_s(G\sminus T)\cup T$, then by definition, $S\cap C_t(G\sminus T)=\emptyset$. Therefore, $C_t(G\sminus T)$ is connected in $G\sminus S$. By definition, this means that $C_t(G\sminus T)\subseteq C_t(G\sminus S)$. Therefore, it holds that $N_G(C_t(G\sminus T))\subseteq C_t(G\sminus S) \cup N_G(C_t(G\sminus S))$. Since $S,T\in \minlsepst{G}$, then by Lemma~\ref{lem:fullComponents}, it holds that $S=N_G(C_t(G\sminus S))$ and $T=N_G(C_t(G\sminus T))$. Consequently, $T\subseteq S\cup C_t(G\sminus S)$. So, we have shown that
	$C_s(G\sminus S)\subseteq C_s(G\sminus T)  \Longrightarrow S\subseteq T\cup C_s(G\sminus T) \Longrightarrow T\subseteq S\cup C_t(G\sminus S)$.
	If $T\subseteq S\cup C_t(G\sminus S)$, then by definition, $T\cap C_s(G\sminus S)=\emptyset$. Therefore, $C_s(G\sminus S)$ is connected in $G\sminus T$. Consequently, $C_s(G\sminus S)\subseteq C_s(G\sminus T)$. This completes the proof.
\end{proof}

\subsection{Separators Between Vertex-Sets}
\begin{lemma}
	\label{lem:simpAB}
	\simpABlemma
\end{lemma}
\begin{proof}
	If $S\in \minlsep{A,B}{G}$, then for every $w\in S$ it holds that $S{\setminus} \set{w}$ no longer separates $A$ from $B$. Hence, there is a path from some $a\in A$ to some $b\in B$ in $G\sminus (S{\setminus} \set{w})$. 
	Let $C_a$ and $C_b$ denote the connected components of $\cc(G\sminus S)$ containing $a\in A$ and $b\in B$, respectively. Since $C_a$ and $C_b$ are connected in $G\sminus (S{\setminus} \set{w})$, then $w\in N_G(C_a)\cap N_G(C_b)$.
	
	Suppose that for every $w\in S$, there exist two connected components $C_A,C_B\in \cc(G\sminus S)$ such that $C_A\cap A\neq \emptyset$, $C_B\cap B\neq \emptyset$, and $w\in N_G(C_A)\cap N_G(C_B)$. If $S\notin \minlsep{A,B}{G}$, then $S{\setminus} \set{w}$ separates $A$ from $B$ for some $w\in S$. Since $w$ connects $C_A$ to $C_B$ in $G\sminus (S{\setminus} \set{w})$, no such $w\in S$ exists, and thus $S\in \minlsep{A,B}{G}$.
\end{proof}

Observe that Lemma~\ref{lem:simpAB} implies Lemma~\ref{lem:fullComponents}. By Lemma~\ref{lem:simpAB}, it holds that $S\in \minlsepst{G}$ if and only if $S$ is an $s,t$-separator and $S\subseteq N_G(C_s(G\sminus S))\cap N_G(C_t(G\sminus S))$. By definition, $N_G(C_s(G\sminus S))\subseteq S$ and $N_G(C_t(G\sminus S))\subseteq S$, and hence $S=N_G(C_s(G\sminus S))\cap N_G(C_t(G\sminus S))$, and $S=N_G(C_s(G\sminus S))=N_G(C_t(G\sminus S))$.

\begin{lemma}
	\label{lem:minlsepsupergraph}
	Let $G$ and $H$ be graphs where $\nodes(G)=\nodes(H)$ and $\edges(G)\subseteq \edges(H)$. Let $S\in \minlsep{A,B}{G}$. If $S$ is an $A,B$-separator in $H$, then $S\in \minlsep{A,B}{H}$.
\end{lemma}
\begin{proof}
	Since $S\in \minlsep{A,B}{G}$, then by Lemma~\ref{lem:simpAB}, for every $w\in S$ there exist $C_A^w(G\sminus S)\in \cc(G\sminus S)$ and $C_B^w(G\sminus S) \in \cc(G\sminus S)$ where $A\cap C_A^w(G\sminus S) \neq \emptyset$, $B\cap C_B^w(G\sminus S)\neq \emptyset$, and $w\in N_G(C_A^w(G\sminus S))\cap N_G(C_B^w(G\sminus S))$. Since $\edges(H)\supseteq \edges(G)$, and since $S$ is an $A,B$-separator in $H$, then $C_A^w(H\sminus S)\supseteq C_A^w(G\sminus S)$, and $C_B^w(H\sminus S)\supseteq C_B^w(G\sminus S)$. Therefore, $w\in N_H(C_A^w(H\sminus S))\cap N_H(C_B^w(H\sminus S)$) for every $w\in S$.
By Lemma~\ref{lem:simpAB}, it holds that $S\in \minlsep{A,B}{H}$.
\end{proof}

\def\lemMinlsASep{
	Let $A,B\subseteq \nodes(G){\setminus}\set{s,t}$ where $sA$ and $tB$ are disjoint and non-adjacent. Let $H$ be the graph that results from $G$ by adding all edges between $s$ and $N_G[A]$. That is, $\edges(H)=\edges(G)\cup \set{(s,v):v\in N_G[A]}$. Then $\minlsep{sA,tB}{G}=\minlsep{s,tB}{H}$, $C_s(H\sminus S)=C_{sA}(G\sminus S)$, and $C_{tB}(H\sminus S)=C_{tB}(G\sminus S)$ for every $S\in \minlsep{sA,tB}{G}$.
}

\begin{lemma}
	\label{lem:MinlsASep}
	\lemMinlsASep
\end{lemma}
\begin{proof}
	Let $T \in \minlsep{sA,tB}{G}$, and let $C_1,\dots,C_k$ denote the connected components of $\cc(G\sminus T)$ containing vertices from $sA$, and $D_1,\dots, D_{\ell}$ denote the connected components of $\cc(G\sminus T)$ that contain vertices from $Bt$. We let $C_{sA}\eqdef \mediumbigcup_{i=1}^k C_i$ and $D_{tB}\eqdef \mediumbigcup_{i=1}^{\ell}D_i$. Since $T$ is an $sA,tB$-separator of $G$ then $C_{sA}$ and $D_{tB}$ are disjoint and nonadjacent.
	Assume wlog that $s\in C_1$. By definition, the edges added to $G$ to form $H$ are between $C_1$ and $C_{sA} \cup T$. Since, by definition, $D_{tB} \cap (C_{sA} \cup T)=\emptyset$, then, $T$ separates $sA$ from $tB$ in $H$, and in particular, $T$ separates $s$ from $tB$ in $H$. Since $\edges(H) \supseteq \edges(G)$, then by Lemma~\ref{lem:minlsepsupergraph}, if $T \in \minlsep{sA,tB}{G}$ and $T$ is an $sA,tB$-separator in $H$, then $T\in \minlsep{sA,tB}{H}$. 
	Since, by construction, $A\subseteq N_{H}[s]{\setminus} T$ then $H\sminus T$ has a connected component $D_{s}\eqdef C_{sA}$ that contains $sA$, and the set of connected components $D_1,\dots, D_{\ell}$ containing vertices from $tB$. By Lemma~\ref{lem:simpAB}, we have that $T = N_{H}(D_{s})=N_{H}(C_{sA})$, and $T= N_{H}(D_{tB})$. Hence, $T= N_H(D_{s})\cap N_H(D_{tB})$. By Lemma~\ref{lem:simpAB}, we have that $T\in \minlsep{s,tB}{H}$, where $D_s\eqdef C_s(H\sminus T)=C_{sA}(G\sminus T)$.
	
	Let $T\in \minlsep{s,tB}{H}$. \rev{First, $T\cap A=\emptyset$. Otherwise, for $a\in T\cap A$, minimality gives a neighbor $z$ of $a$ in a component of $H\sminus T$ containing a vertex of $tB$. The edge $az$ belongs to $G$, since every added edge is incident with $s$ and $z$ is outside the $s$-component. But then $z\in N_G[A]$, so $sz$ is an edge of $H$, contradicting separation.} We next show that $T$ separates $sA$ from $tB$ in $G$; if not, there is a path from $x\in sA$ to $tB$ in $G\sminus T$. By definition of $H$, $x\in N_{H}[s]{\setminus}T$. This means that there is a path from $s$ to $tB$ (via $x$) in $H \sminus T$, which contradicts the assumption that $T$ is an $s,tB$-separator of $H$. If $T \notin \minlsep{sA,tB}{G}$, then there is a $T' \in \minlsep{sA,tB}{G}$ where $T' \subset T$. By the previous direction,  $T'\in \minlsep{sA,tB}{G} \subseteq \minlsep{s,tB}{H}$, and hence $T'\in \minlsep{s,tB}{H}$, contradicting the minimality of $T\in \minlsep{s,tB}{H}$.
	
	Now, let $S\in \minlsep{s,tB}{H}$. Then $S \in \minlsep{sA,tB}{G}$, where every edge in $\edges(H){\setminus}\edges(G)$ is between $s$ and a vertex in $C_{sA}(G\sminus S)\cup S$. Therefore, $C_s(H\sminus S)=C_{sA}(G\sminus S)$. For the same reason, no added edge has an endpoint in a component of $G\sminus S$ that meets $tB$. Each such component is therefore unchanged in $H\sminus S$, and hence $C_{tB}(H\sminus S)=C_{tB}(G\sminus S)$.
\end{proof}

\begin{revision}
\begin{lemma}
\label{lem:reversedImportantSets}
Let $A,B\subseteq\nodes(G)$ be nonempty, disjoint, and nonadjacent. For minimal $A,B$-separators $S,T$,
\[
C_A(G\sminus T)\subsetneq C_A(G\sminus S)
\quad\Longleftrightarrow\quad
C_B(G\sminus T)\supsetneq C_B(G\sminus S).
\]
Consequently, Definition~\ref{def:importantABSeps} is the classical important-separator definition with the source and target interchanged, also when classical domination is tested against all separators.
\end{lemma}
\begin{proof}
Choose $s\in A$ and $t\in B$. Apply Lemma~\ref{lem:MinlsASep} first to join the source set to $s$, and then, with the roles reversed, to join the target set to $t$. The resulting graph $H$ has precisely the same minimal separators, with
$C_s(H\sminus S)=C_A(G\sminus S)$ and
$C_t(H\sminus S)=C_B(G\sminus S)$ for every minimal $A,B$-separator $S$.
Lemma~\ref{lem:inclusionCsCt}, applied in $H$ with the two orientations, gives the equivalence for non-strict inclusion. Equality of either component determines its minimal separator as its neighborhood, so strict inclusion is equivalent as well.

Finally, suppose an arbitrary $A,B$-separator $U$ of size at most $|S|$ has $C_B(G\sminus U)\supsetneq C_B(G\sminus S)$. Choose an inclusion-minimal $A,B$-separator $T\subseteq U$. Deleting fewer vertices only enlarges the set reachable from $B$, so
$C_B(G\sminus T)\supseteq C_B(G\sminus U)\supsetneq C_B(G\sminus S)$ and $|T|\leq|S|$. Thus allowing nonminimal competitors does not change importance in this target-maximizing orientation.
\end{proof}
\end{revision}

\eat{\subsection{Proof from Section~\ref{sec:safeImportantCloseSeps}}
\label{sec:safeImportantCloseSepsProofs}}
\def\lemMinlsBtSep{
	Let $A,B\subseteq \nodes(G)$ be disjoint, and let $t\in B$. Let $H$ be the graph that results from $G$ by adding all edges between $t$ and $B$. That is, $\edges(H)=\edges(G)\cup \set{(t,b):b\in B}$. Then 
	\begin{align*}
		\set{S\in \minlsep{A,t}{G}: B\cap C_A(G\sminus S)=\emptyset}\subseteq \minlsep{A,t}{H}.
	\end{align*}
	where $C_A(G\sminus S)\eqdef \bigcup_{\set{C\in \cc(G\sminus S): C\cap A\neq \emptyset}}C$.
}
\begin{lemma}
	\label{lem:MinlsBtSepNew}
	\lemMinlsBtSep
\end{lemma}
\begin{proof}
	Let $T \in \minlsep{A,t}{G}$ where $B\cap C_A(G\sminus T)=\emptyset$, and let $C_1,\dots,C_k$ denote the connected components of $\cc(G \sminus T)$ containing vertices from $Bt$. Assume wlog that $t\in C_1$. By definition, the edges added to $G$ to form $H$ are between $C_1$ and $C_1\cup \cdots \cup C_k\cup T$. Since $C_A(G\sminus T)\cap B=\emptyset$, then $T$ separates $A$ from $t$ in $H$. Since $\edges(G)\subseteq \edges(H)$, then by Lemma~\ref{lem:minlsepsupergraph}, if $T\in \minlsep{A,t}{G}$ and $T$ is an $A,t$-separator in $H$, then $T\in \minlsep{A,t}{H}$. 
\end{proof}

\def\lemMinlsBtSepSafe{
	Let $s,t\in \nodes(G)$, $A,B,\P,Q$ be the parameters defining the connectivity constraint $\mc{R}_{A,B,\P,Q}$, where $A,B,Q\subseteq \nodes(G){\setminus}\set{s,t}$, and let $D\subseteq \nodes(G){\setminus}(A\cup \set{s,t})$.
	Let $H$ be the graph that results from $G$ by adding all edges between $t$ and $D$. That is, $\edges(H)=\edges(G)\cup \set{(t,d):d\in D}$. 
	\begin{align*}
		\set{S\in \safeSeps{sA}{t}(G,\mc{R}): D\cap C_{sA}(G\sminus S)=\emptyset} &\subseteq \safeSeps{sA}{t}(H,\mc{R}), \text{ and }\\
		\set{S\in \safeSepsImp{sA}{t}(G,\mc{R}): D\cap C_{sA}(G\sminus S)=\emptyset} &\subseteq \safeSepsImp{sA}{t}(H,\mc{R}). 
	\end{align*}
}
\def\lemMinlsBtSepSafe{
	Let $s,t\in \nodes(G)$, and let
$A,D\subseteq \nodes(G){\setminus}\set{s,t}$. Let $H$ be the graph
obtained from $G$ by adding all edges between $t$ and $D$; that is,
\[
\nodes(H)=\nodes(G)
\qquad\text{and}\qquad
\edges(H)=\edges(G)\cup\set{(t,d):d\in D}.
\]
Then
\[
\set{S\in \safeSeps{s}{t}(G,A):D\cap C_{s}(G\sminus S)=\emptyset}
\subseteq
\safeSeps{s}{t}(H,A),
\]
and
\[
\set{S\in \safeSepsImp{s}{t}(G,A):D\cap C_{s}(G\sminus S)=\emptyset}
\subseteq
\safeSepsImp{s}{t}(H,A).
\]
}
\begin{lemma}
	\label{corr:MinlsBtSepSafe}
	\lemMinlsBtSepSafe
\end{lemma}
\begin{proof}
	We first prove the inclusion for CP separators. Let $S\in \safeSeps{s}{t}(G,A)$ where $	D\cap C_s(G\sminus S)=\emptyset$.
	Set $C\eqdef C_s(G\sminus S)$.
	Since $S\in \safeSeps{s}{t}(G,A)$, we have $S\in \minlsepst{G}$ and $A\subseteq C$.
	We claim first that
	\[
	C_s(H\sminus S)=C_s(G\sminus S)=C.
	\]
	Indeed, the only edges in $\edges(H){\setminus}\edges(G)$ are edges
	between $t$ and vertices of $D$. Since $t\notin C$ and
	$D\cap C=\emptyset$, no added edge has an endpoint in $C$. Hence the
	$s$-component is unchanged after deleting $S$. That is, $C_s(H\sminus S)=C_s(G\sminus S)=C$.
	In particular, $S$ is an $s,t$-separator of $H$ and $A\subseteq C_s(H\sminus S)$.

	It remains to verify minimality in $H$. Since $S\in\minlsepst{G}$,
	Lemma~\ref{lem:simpAB} implies that every vertex $v\in S$ has a
	neighbor in $C_s(G\sminus S)$ and a neighbor in $C_t(G\sminus S)$.
	The first component is unchanged in $H\sminus S$, and the
	$t$-component of $H\sminus S$ contains $C_t(G\sminus S)$. Therefore
	every $v\in S$ has a neighbor in both $C_s(H\sminus S)$ and
	$C_t(H\sminus S)$. By Lemma~\ref{lem:fullComponents}, applied to the
	pair $s,t$ in $H$, we get $S\in\minlsepst{H}$.
	Together with $A\subseteq C_s(H\sminus S)$, this gives $S\in\safeSeps{s}{t}(H,A)$, proving the first inclusion.

	We now prove the inclusion for important CP separators. Let $S\in \safeSepsImp{s}{t}(G,A)$ where $D\cap C_s(G\sminus S)=\emptyset$.
	By the first part, $S\in \safeSeps{s}{t}(H,A)$,
	and, as above, $C_s(H\sminus S)=C_s(G\sminus S)$.
	Set $C\eqdef C_s(G\sminus S)$.
	Suppose, toward a contradiction, that $	S\notin \safeSepsImp{s}{t}(H,A)$.
	Then there exists $S'\in\safeSeps{s}{t}(H,A)$ such that
	\[
	C_s(H\sminus S')\subsetneq C_s(H\sminus S)
	\qquad\text{and}\qquad
	|S'|\le |S|.
	\]
	Let $C'\eqdef C_s(H\sminus S')$.
	Then $C'\subsetneq C$.

	We first observe that $	C'=C_s(G\sminus S')$.
	Indeed, since $S'$ is an $s,t$-separator in $H$, the vertex $t$ is not
	in $C'$. If some $d\in D$ belonged to $C'$, then the added edge
	$(t,d)$ would put $t$ in the same component as $s$ in $H\sminus S'$,
	a contradiction. Hence $D\cap C'=\emptyset$.
	Since all edges of $\edges(H){\setminus}\edges(G)$ are incident with
	$t$ and a vertex of $D$, no added edge is incident with $C'$. Therefore
	the $s$-component is the same in $G\sminus S'$ and in $H\sminus S'$,
	as claimed. That is, $C_s(G\sminus S')=C_s(H\sminus S')$.
	
	Now define $W\eqdef N_G(C')$.
	Since $C'=C_s(G\sminus S')$, every neighbor of $C'$ in $G$ belongs to
	$S'$, and hence $W\subseteq S'$.
	Therefore
	\[
	|W|\le |S'|\le |S|.
	\]
	Moreover, $W$ separates $s$ from $t$ in $G$, and $C_s(G\sminus W)=C'$.
	Since $S'\in\safeSeps{s}{t}(H,A)$, we have
	\[
	A\subseteq C_s(H\sminus S')=C'=C_s(G\sminus W).
	\]
	
	It remains to prove that $W$ is a minimal $s,t$-separator in $G$; $W\in \minlsepst{G}$.
	By definition of $W=N_G(C')$, every vertex $w\in W$ has a neighbor in
	$C'$. We show that every $w\in W$ also has a path to $t$ in
	$G\sminus (W{\setminus}\set{w})$.
	
	Fix $w\in W$. Since $W\subseteq S'$ and
	$S'\in\minlsepst{H}$, Lemma~\ref{lem:simpAB} gives a $w,t$-path
	$P$ in $H\sminus (S'{\setminus}\set{w})$. If $P$ uses no edge from
	$\edges(H){\setminus}\edges(G)$, then $P$ is already a $w,t$-path in
	$G\sminus (W{\setminus}\set{w})$, because $W\subseteq S'$.
	
	Assume therefore that $P$ uses an added edge. Since all added edges
	are incident with $t$, we may take $P$ so that its last edge is
	$(d,t)$ for some $d\in D$, and the subpath $P[w,d]$ is contained in
	$G\sminus (S'{\setminus}\set{w})$.
	
	Since $d\in D$ and $D\cap C=\emptyset$, we have $d\notin C$. On the
	other hand, $w$ has a neighbor in $C'\subseteq C$. Thus, along the
	path obtained by starting at that neighbor of $w$, going to $w$, and
	then following $P[w,d]$, one goes from $C$ to a vertex outside $C$.
	Because $S=N_G(C)$, this path must meet $S$. Let $z$ be the first
	vertex of $S$ encountered on this path when moving from $w$ toward
	$d$, allowing the possibility $z=w$.
	The subpath from $w$ to $z$ uses only edges of $G$ and avoids
	$S'{\setminus}\set{w}$. Since $W\subseteq S'$, this subpath also avoids
	$W{\setminus}\set{w}$.
	
	Because $S\in\minlsepst{G}$, Lemma~\ref{lem:simpAB} gives a
	$z,t$-path in $G\sminus (S{\setminus}\set{z})$. This path is contained
	in $C_t(G\sminus S)\cup\set{z}$. Since $C'\subseteq C=C_s(G\sminus S)$,
	every vertex of $W=N_G(C')$ lies in $C\cup S$. Thus the internal
	vertices of this $z,t$-path avoid $W$. Moreover, if $z\neq w$, then
	$z$ was reached along a path avoiding $S'{\setminus}\set{w}$, so
	$z\notin S'$; as $W\subseteq S'$, this gives $z\notin W$. If
	$z=w$, then the path is allowed to start at $w$. Hence the $z,t$-path
	avoids $W{\setminus}\set{w}$.
	Combining the $w,z$ subpath with the $z,t$ path gives a $w,t$-path in
	$G\sminus (W{\setminus}\set{w})$.
	
	We have shown that every $w\in W$ has a neighbor in
	$C_s(G\sminus W)=C'$ and a path to $t$ in
	$G\sminus (W{\setminus}\set{w})$. Equivalently, every $w\in W$ has a
	neighbor in the $s$-component and a neighbor in the $t$-component of
	$G\sminus W$. By Lemma~\ref{lem:fullComponents}, $W\in\minlsepst{G}$.
	Since also $A\subseteq C_s(G\sminus W)$, we have $	W\in\safeSeps{s}{t}(G,A)$.
	
	Finally,
	\[
	C_s(G\sminus W)=C'\subsetneq C=C_s(G\sminus S)\qquad\text{and}\qquad	|W|\le |S|.
	\]
	This contradicts the assumption that
	$S\in\safeSepsImp{s}{t}(G,A)$.
	Therefore no such $S'$ exists, and hence $	S\in\safeSepsImp{s}{t}(H,A)$.
	This proves the second inclusion and completes the proof.
\end{proof}

\section{The Foundations of Connectivity-Preserving Important Minimum Separators}
\label{sec:mainThmBoundedCardinalityAppendix}
Let $A\subseteq \nodes(G){\setminus}\set{s,t}$.
When $\kappa_{sA,t}(G)<\infty$, recall that $\rev{\closestMinSep{sA,t}(G)}$ is the unique important, minimum $sA,t$-separator (see Lemma~\ref{lem:uniqueMinImportantSep}). In this section, we present our main technical result, which forms the foundation for the enumeration algorithm.
\begin{reptheorem}{\ref{thm:mainThmBoundedCardinality}}
	\mainThmBoundedCardinality
\end{reptheorem}	
Theorem~\ref{thm:mainThmBoundedCardinality} shows that if $\rev{\kappa_{s,t}(G)=\kappa_{sA,t}(G)=f_{s,t}(G,A)<\infty}$, then there are at most $2^{\kappa_{s,t}(G)}$ CP-important {\bf minimum} $s,t$-separators with respect to $A$, and they can be computed in time $O(2^{\kappa_{s,t}(G)} \cdot n \cdot T(n,m))$. This result extends the classical framework for important separators to the setting where the separator must preserve the connectivity constraints; namely that the terminals $A$ remain connected after the separation. While uniqueness no longer holds in this setting, Theorem~\ref{thm:mainThmBoundedCardinality} shows that the number of relevant separators is singly-exponential in $\kappa_{s,t}(G)$, and that this set can be computed efficiently.

Theorem~\ref{thm:minsepsAt} below plays a central role in the proof of Theorem~\ref{thm:mainThmBoundedCardinality}. It shows that if $D\subseteq \nodes(G)$ is contained in the union of all minimum $s,t$-separators (i.e., $D \subseteq \minstVertices{s,t}(G)$, see~\eqref{eq:minstVertices}), then every minimum $sD,t$-separator necessarily preserves the connectivity of $sD$; that is, $D \subseteq C_s(G \sminus T)$ for all $T \in \minsep_{sD,t}(G)$. This guarantee is essential for the construction of connectivity-preserving separators. Beyond its role in our analysis, Theorem~\ref{thm:minsepsAt} may be of independent interest, as it uncovers a structural property of $sD,t$-separators when $D \subseteq \minstVertices{s,t}(G)$. 
Section~\ref{sec:thm_minsepsAt} of this Appendix is devoted to proving Theorem~\ref{thm:minsepsAt}, and Section~\ref{sec:thm_mainThmBoundedCardinalityProof} applies it to prove Theorem~\ref{thm:mainThmBoundedCardinality}.

\subsection{Proof of Theorem~\ref{thm:minsepsAt} and Corollary~\ref{corr:minsepsAtSets}}
\label{sec:thm_minsepsAt}
\begin{reptheorem}{\ref{thm:minsepsAt}}
	\minsepsAtThm
\end{reptheorem}	
The rest of this section is devoted to proving Theorem~\ref{thm:minsepsAt} that is used for proving Theorem~\ref{thm:mainThmBoundedCardinality}.
To that end, we require the following.
\begin{citedtheorem}{Menger~\cite{DBLP:books/daglib/0030488}}
	\label{thm:Menger}
	Let $G$ be an undirected graph and $s,t \in \nodes(G)$ \rev{be distinct and nonadjacent}. The minimum number of vertices separating $s$ from $t$ in $G$ (i.e., $\kappa_{s,t}(G)$) is equal to the maximum number of internally vertex-disjoint $s,t$-paths in $G$.
\end{citedtheorem}

\def\vertexIncludeLem{
	\rev{Assume $\kappa_{s,t}(G)<\infty$.} Let $v\in \nodes(G){\setminus}\set{s,t}$. There exists a minimum $s,t$-separator $S\in \minsepst{G}$ that contains $v$ if and only if $\kappa_{s,t}(G\sminus v)=\kappa_{s,t}(G)-1$.
}
\begin{lemma}
	\label{lem:vertexInclude}
	\vertexIncludeLem
\end{lemma}
\begin{proof}
	Let $S\in \minsep_{s,t}(G)$ where $v\in S$. Then $S{\setminus}\set{v}$ is an $s,t$-separator of $G\sminus v$. Therefore, $\kappa_{s,t}(G\sminus v)\leq |S{\setminus}\set{v}|=\kappa_{s,t}(G)-1$.
	Suppose, by way of contradiction, that  $\kappa_{s,t}(G\sminus v)<\kappa_{s,t}(G)-1$, and let $T\in \minsep_{s,t}(G\sminus v)$. Then $T\cup \set{v}$ is an $s,t$-separator of $G$ where $|T\cup \set{v}|=\kappa_{s,t}(G\sminus v)+1<\kappa_{s,t}(G)$; a contradiction.
	
	For the other direction, suppose that $\kappa_{s,t}(G\sminus v)=\kappa_{s,t}(G)-1$, and let $T\in \minsep_{s,t}(G\sminus v)$. $T\cup \set{v}$ is an $s,t$-separator of $G$ where $|T\cup \set{v}|=\kappa_{s,t}(G\sminus v)+1=\kappa_{s,t}(G)$. By definition, $T\cup \set{v}\in \minsep_{s,t}(G)$.
\end{proof}

\def\technicalfortworesults{
	Let $D\subseteq \nodes(G)$, and let $T\in \minsep_{sD,t}(G)$. Let $T'\eqdef N_G(C_s(G\sminus T))$, and $G'\eqdef G[C_s(G\sminus T)\cup T' \cup C_t(G\sminus T)]$. Then the following holds:
	\begin{enumerate}
		\item For all $S\in \minlsepst{G'}$ such that $S\subseteq T'\cup C_s(G'\sminus T')$ it holds that $S\in \minlsepst{G}$.
		\item For all $S\in \minlsepst{G'}$ such that $S\subseteq T'\cup C_t(G'\sminus T')$, it holds that $|S|\geq  |T'|$.
		\item  $\kappa_{s,t}(G')=\kappa_{s,t}(G)$.
	\end{enumerate}
}
\begin{lemma}
	\label{lem:technicalfortworesults}
	\technicalfortworesults
\end{lemma}
\begin{proof}
	We first prove item 1.
	Since $T\in \minlsep{sD,t}{G}$, then by Lemma~\ref{lem:simpAB}, it holds that $T=N_G(C_t(G\sminus T))$. We claim that $T'\in \minlsepst{G}$. Since $T'\eqdef N_G(C_s(G\sminus T))$, then clearly it is an $s,t$-separator in $G$. Since $T'\subseteq T=N_G(C_t(G\sminus T))$, then $T'{\setminus}\set{w}$ no longer separates $s$ from $t$ in $G$, for any $w\in T'$. Therefore, by definition, we have that $T'\in \minlsepst{G}$. By definition of $G'$ it holds that $C_s(G\sminus T)=C_s(G\sminus T')=C_s(G'\sminus T')$ and $C_t(G\sminus T)=C_t(G'\sminus T')$. Therefore, $T'=N_{G'}(C_s(G'\sminus T'))=N_{G'}(C_t(G'\sminus T'))$, and by lemma~\ref{lem:fullComponents}, we have that $T'\in \minlsepst{G'}$. Overall, we have that $T'\in \minlsepst{G}\cap \minlsepst{G'}$.
	Let $T''\eqdef T{\setminus} T'$. Observe that $T'$ is an $sD,t$-separator in $G\sminus T''$.
	
	Let $W\in \minlsepst{G'}$, such that $W\subseteq T'\cup C_s(G\sminus T)$. To prove item 1 of the lemma it is enough to show that $W$ is an $s,t$-separator in $G$ because $C_s(G'\sminus W)\subseteq C_s(G\sminus W)$, and $C_t(G'\sminus W)\subseteq C_t(G\sminus W)$. By lemma~\ref{lem:fullComponents}, $C_s(G'\sminus W)$ and $C_t(G'\sminus W)$ are full components of $W$ in $G'$. Therefore, if $W$ is an $s,t$-separator in $G$ then $C_s(G\sminus W)$ and $C_t(G\sminus W)$ are full components of $W$ in $G$, and by lemma~\ref{lem:fullComponents}, it holds that $W\in \minlsepst{G}$. So, we prove that $W$ is an $s,t$-separator in $G$. Suppose it is not, and let $P$ be an $s,t$-path in $G\sminus W$. Since $T'\in \minlsepst{G}$, then $\nodes(P)\cap T'\neq \emptyset$. Let $a\in \nodes(P)\cap T'$ be the vertex in $ \nodes(P)\cap T'$ that is closest to $s$ on $P$. The subpath of $P$ from $s$ to $a$, except for $a$ itself, lies in $C_s(G\sminus T)$; the entire subpath avoids $W$. On the other hand, $T'\subseteq T = N_G(C_t(G\sminus T))$, and hence $a\in N_G(C_t(G\sminus T))$. Since, by definition, $C_t(G\sminus T)=C_t(G'\sminus T')$, and $W \subseteq T'\cup C_s(G\sminus T)$ where $C_s(G\sminus T)=C_s(G'\sminus T')$, then $W$ does not separate $s$ from $t$ in $G'$, contradicting the fact that $W\in \minlsepst{G'}$. This proves that $\set{W\in \minlsepst{G'}: W\subseteq T'\cup C_s(G\sminus T)}\subseteq \minlsepst{G}$. In particular, this means that for every $W\in \minlsepst{G'}$, such that $W\subseteq T'\cup C_s(G\sminus T)$, it holds that $|W|\geq \kappa_{s,t}(G)$.
	
	\eat{
	We now prove item 2.
	For every $W\in \minlsepst{G'}$, such that $W\subseteq C_t(G\sminus T)\cup T'$, it holds that $|T'|\leq |W|$. To prove the claim suppose, by way of contradiction, that there exists a $W\in \minlsepst{G'}$ where $W\subseteq C_t(G\sminus T)\cup T'$ and $|W|<|T'|$. Let $C\in \cc(G\sminus T)$ be any connected component of $G\sminus T$ such that $C\cap sD\neq \emptyset$ (see Figure~\ref{fig:proofIllustration} for illustration).
	
	Consider any path from $C$ to $t$ in $G$. Since $T=T'\cup T''$ where, by definition $T'\cap T''=\emptyset$, then any path from $C$ to $t$ in $G$ that avoids $T''$ must pass through a vertex in $T'$. In other words, every path from $C$ to $t$ in $G\sminus T''$ passes through a vertex in $T'$, or that $T'$ is an $s,tD$-separator in $G\sminus T''$.
	Since $T'\in \minlsepst{G'}$ and $W\in \minlsepst{G'}$ where $W\subseteq T'\cup C_t(G\sminus T)$, then every path from $T'$ to $t$ in $G\sminus  T''$ passes through a vertex in $W$. In particular, every path from $C$ to $t$ in $G\sminus T''$ passes through a vertex in $W$. Hence, we have that every path from $C\in \cc(G\sminus T)$ (where $C\cap sD\neq \emptyset$) to $t$, passes through a vertex in $W\cup T''$. In other words, $W\cup T''$ separates $C$ from $t$ in $G$ for every $C\in \cc(G\sminus T)$ where $C\cap sD\neq \emptyset$. That is, $W\cup T''$ is an $sD,t$-separator in $G$. On the other hand, $|W|<|T'|$, and therefore, $|W\cup T''|\leq |W|+|T''|<|T'|+|T''|=|T|$. But then, $T\notin \minsep_{sD,t}(G)$, and we arrive at a contradiction. }
	
	We now prove item 2. Let $T' \eqdef N_G(C_s(G\sminus T))$, and $T'' \eqdef T{\setminus} T'$.
	Let $W\in \minlsepst{G'}$ such that $W\subseteq T'\cup C_t(G'\sminus T')$.
	Since, by the definition of $G'$, we have $C_t(G'\sminus T')=C_t(G\sminus T)$,
	we may write $W\subseteq T'\cup C_t(G\sminus T)$.
	We prove that $|W|\ge |T'|$.
	
	Suppose, towards a contradiction, that $|W|<|T'|$. We first show that
	$W\cup T''$ is an $sD,t$-separator in $G$.
		
	Let $C$ be a connected component of $G\sminus T$ such that $C\cap sD\neq \emptyset$.
	We claim that every path from $C$ to $t$ in $G$ meets $W\cup T''$.
	Assume otherwise, and let $P$ be a path from a vertex of $C$ to $t$ in
	$G\sminus (W\cup T'')$.
		
	Since $T$ is an $sD,t$-separator, every path from $C$ to $t$ must meet
	$T$. Since $P$ avoids $T''$, it follows that $P$ meets $T'$. Let $a$ be
	the last vertex of $P$ that belongs to $T'$, when traversing $P$ from
	$C$ towards $t$. Since $P$ avoids $W$, we have $a\notin W$.
		
	Consider the suffix $P[a,t]$ of $P$ from $a$ to $t$. By the choice of
	$a$, this suffix contains no vertex of $T'$ other than $a$, and since
	$P$ avoids $T''$, it contains no vertex of $T''$. Hence the internal
	vertices of $P[a,t]$ avoid $T=T'\cup T''$.
		
	We claim that all internal vertices of $P[a,t]$ lie in $C_t(G\sminus T)$.
	Indeed, after leaving $a$, the path enters some connected component of
	$G\sminus T$. If this component were not $C_t(G\sminus T)$, then in order to reach
	$t$, the path would have to leave it through a vertex of $T$. But the
	suffix $P[a,t]$ contains no vertex of $T$ after $a$, a contradiction.
		Therefore, $	\nodes(P[a,t])\subseteq \set{a}\cup C_t(G\sminus T)$.
	Since $a\in T'$ and $C_t(G\sminus T)=C_t(G'\sminus T')$, this suffix is an $a,t$-path
	in $G'$. Moreover, it avoids $W$.
		
		On the other hand, $a\in T'=N_G(C_s(G\sminus T))$, so $a$ has a neighbor
		$b\in C_s(G\sminus T)$. Since $	W\subseteq T'\cup C_t(G\sminus T)$,
		we have $W\cap C_s(G\sminus T)=\emptyset$.
		Thus $b$ is connected to $s$ inside $C_s(G\sminus T)$ in $G'\sminus W$, and together
		with the edge $ba$ this gives an $s,a$-path in $G'\sminus W$. Combining this
		$s,a$-path with the $a,t$-path $P[a,t]$ yields an $s,t$-path in
		$G'\sminus W$, contradicting the assumption that $W$ is an $s,t$-separator in
		$G'$.
		
		Therefore every path from every component $C$ of $G\sminus T$ with
		$C\cap sD\neq\emptyset$ to $t$ meets $W\cup T''$. Equivalently,
		$W\cup T''$ is an $sD,t$-separator in $G$.
		
		Finally, $W\subseteq V(G')$, while $T''\cap V(G')=\emptyset$ by the
		definition of $G'$. Hence $W\cap T''=\emptyset$, and therefore
		$$
		|W\cup T''|
		=
		|W|+|T''|
		<
		|T'|+|T''|
		=
		|T|.
		$$
		This contradicts the assumption that $T\in \minsep_{sD,t}(G)$.		
		Thus no such $W$ with $|W|<|T'|$ exists, and consequently $	|W|\ge |T'|$.

	We now prove item 3. Since every $s,t$-path in $G'$ is contained in $G$, then clearly $\kappa_{s,t}(G')\leq\kappa_{s,t}(G)$. We prove that $\kappa_{s,t}(G')\geq\kappa_{s,t}(G)$.  
	From item 2 of the lemma,  we have that for every $W\in \minlsepst{G'}$ such that $W\subseteq T'\cup C_t(G\sminus T)$, it holds that $|W|\geq |T'|$. From item 1 of the lemma, we have that for every $W\in \minlsepst{G'}$ such that $W\subseteq T'\cup C_s(G\sminus T)$, it holds that $|W|\geq \kappa_{s,t}(G)$. By submodularity (Corollary~\ref{corr:submodularity}), and since $T'\in \minlsepst{G'}$, then for any $S\in \minlsepst{G'}$ it holds that:
	\begin{equation}
		\label{eq:submodFoAddingEdge1}
		|T'|+ |S|\geq |N_{G'}\left(C_s(G'\sminus S)\cap C_s(G'\sminus T')\right)| + |N_{G'}\left(C_s(G'\sminus S)\cup C_s(G'\sminus T')\right)|
	\end{equation}
\rev{Both neighborhoods on the right-hand side of~\eqref{eq:submodFoAddingEdge1}
separate $s$ from $t$ in $G'$: their respective interiors contain $s$,
and neither interior contains $t$ or a neighbor of $t$. Choose
inclusionwise minimal $s,t$-separators $S_1$ and $S_2$ inside these
neighborhoods, respectively. No uniqueness is needed.} Therefore, from~\eqref{eq:submodFoAddingEdge1}, we have that:
	\begin{align}
		|S|+|T'|&\geq |N_{G'}\left(C_s(G'\sminus S)\cap C_s(G'\sminus T')\right)| + |N_{G'}\left(C_s(G'\sminus S)\cup C_s(G'\sminus T')\right)| \nonumber \\
		&\geq |S_1|+|S_2| \underbrace{\geq}_{\substack{S_2\subseteq T'\cup C_t(G\sminus T)\\S_1 \subseteq T'\cup C_s(G\sminus T)}} \kappa_{s,t}(G)+|T'|	\label{eq:submodFoAddingEdge21}
	\end{align}
Put $R=C_s(G'\sminus T')$ and $Q=C_s(G'\sminus S)$.
Since $N_{G'}(R)=T'$, every neighbor of $R\cap Q$ belongs
to $R\cup T'$. Hence $S_1\subseteq R\cup T'$.
The external neighborhood of $R\cup Q$ is disjoint from $R$.
By the definition of $G'$,
$\nodes(G')=R\cup T'\cup C_t(G'\sminus T')$,
so $S_2\subseteq T'\cup C_t(G'\sminus T')$.
	
	So, we have that $S_1 \subseteq T'\cup C_s(G\sminus T)$, and that $S_2 \subseteq T'\cup C_t(G\sminus T)$. 
	By claim 2, for every $W\in \minlsepst{G'}$ such that $W\subseteq T'\cup C_t(G\sminus T)$, it holds that $|W|\geq |T'|$. In particular, $|S_2|\geq |T'|$. By claim 1, for every $W\in \minlsepst{G'}$ such that $W\subseteq T'\cup C_s(G\sminus T)$, it holds that $|W|\geq \kappa_{s,t}(G)$. From~\eqref{eq:submodFoAddingEdge21}, we get that $|S|\geq \kappa_{s,t}(G)$. Since $S\in \minlsepst{G'}$ can be any minimal $s,t$-separator of $G'$, then $\kappa_{s,t}(G')\geq \kappa_{s,t}(G)$. Since $\nodes(G')\subseteq \nodes(G)$ then $\kappa_{s,t}(G')\leq \kappa_{s,t}(G)$. Therefore,  $\kappa_{s,t}(G')= \kappa_{s,t}(G)$, and this completes the proof.
\end{proof}
\eat{
	\begin{figure}[h]
		\centering
		\includegraphics[width=0.35\textwidth]{proofillustrationnew.pdf}
		\caption{Illustration for the proof of Lemma~\ref{lem:technicalfortworesults}, case 2. The vertices of $T$ are represented by the black solid vertices. The graph $G'$ is enclosed in the dotted ellipse, and the vertices of $T''$ (that are outside of $G'$) are encircled.	The connected components $C_1,C_2,C_3\in \cc(G\sminus T)$ have the property that $C_i\cap D\neq \emptyset$. \label{fig:proofIllustration}}
	\end{figure}
}
\begin{figure}[ht]
	\centering
	
	\begin{tikzpicture}
		\foreach \y/\name in {3/v1, 2/v2, 1/v3, 0/v4} {
			\node[inner sep=0pt, minimum size=2.2mm, circle, fill=black] (\name) at (0,\y) {};
		}
		
		\node[inner sep=0pt, minimum size=2.2mm, circle, fill=black] (v5) at (3,4) {};
		\node[inner sep=0pt, minimum size=2.2mm, circle, fill=black] (v6) at (2,-1) {};
		\node[inner sep=0pt, minimum size=2.2mm, circle, fill=black] (v7) at (4,-1) {};
		
		\draw (v5) circle[radius=3mm];
		\draw (v6) circle[radius=3mm];
		\draw (v7) circle[radius=3mm];
		
		\node[draw, circle, minimum size=2cm] (Cs) at (-4,1.5) {};
		\node[draw, circle, minimum size=2cm] (Ct) at (4,1.5) {};
		\node[draw, circle, minimum size=0.8cm] (C1) at (1.3,5) {};
		\node[draw, circle, minimum size=0.8cm] (C2) at (-0.5,-1.5) {};
		\node[draw, circle, minimum size=0.8cm] (C3) at (4.2,-2.2) {};
		
		\node at (Cs) {\scriptsize $C_s(G{-}T)$};
		\node at (Ct) {\scriptsize $C_t(G{-}T)$};
		\node at (C1) {\scriptsize $C_1$};
		\node at (C2) {\scriptsize $C_2$};
		\node at (C3) {\scriptsize $C_3$};
		
		\node at (-2,3.7) {\scriptsize $G'$};
		
		\foreach \i/\angle in {1/60, 2/30, 3/-30, 4/-60} {
			\path (Cs.center) ++(\angle:1cm) coordinate (p\i);
			\draw (p\i) -- (v\i);
		}
		
		\foreach \i/\angle in {1/120, 2/150, 3/210, 4/240} {
			\path (Ct.center) ++(\angle:1cm) coordinate (q\i);
			\draw (q\i) -- (v\i);
		}
		
		\path (C1.center) ++(180:0.4cm) coordinate (c1left);
		\path (C1.center) ++(-90:0.4cm) coordinate (c1bottom);
		\path (C1.center) ++(0:0.4cm) coordinate (c1right);
		\draw (c1left) -- (v1);
		\draw[bend left=25] (c1bottom) to (v2);
		\draw (c1right) -- (v5);
		
		\path (C2.center) ++(90:0.4cm) coordinate (c2top);
		\path (C2.center) ++(60:0.4cm) coordinate (c2topright);
		\draw (c2top) -- (v4);
		\draw[bend right=36] (c2topright) to (v2);
		
		\path (C2.center) ++(0:0.4cm) coordinate (c2right);
		\draw (c2right) -- (v6);
		
		\path (C2.center) ++(315:0.4cm) coordinate (c2bottomright);
		\draw (c2bottomright) -- (v7);
		
		\path (Ct.center) ++(270:1cm) coordinate (ctbottom);
		\draw (ctbottom) -- (v7);
		
		\path (Ct.center) ++(90:1cm) coordinate (cttop);
		\draw (cttop) -- (v5);
		
		\path (C3.center) ++(90:0.4cm) coordinate (c3top);
		\draw (c3top) -- (v7);
		
		\path (Ct.center) ++(265:1cm) coordinate (ctlower);
		\draw (ctlower) -- (v6);
		
		\draw[dotted, thick] (0,1.5) ellipse [x radius=5.6cm, y radius=2.1cm];
		
	\end{tikzpicture}
	\caption{Illustration for the proof of Lemma~\ref{lem:technicalfortworesults}, case 2. The vertices of $T$ are represented by the black solid vertices. The graph $G'$ is enclosed in the dotted ellipse, and the vertices of $T''$ (that are outside of $G'$) are encircled.	The connected components $C_1,C_2,C_3\in \cc(G\sminus T)$ have the property that $C_i\cap D\neq \emptyset$. \label{fig:proofIllustrationAppendix}}
\end{figure}
\begin{reptheorem}{\ref{thm:minsepsAt}}
	\minsepsAtThm
\end{reptheorem}	
\begin{proof}
	Let $T\in \minsep_{sDX,t}(G)$, and let $T'\eqdef N_G(C_s(G\sminus T))$. 
	Apply Lemma~\ref{lem:technicalfortworesults} with $D' \eqdef D\cup X$. Since
	$T\in \minsep_{sD',t}(G)$, for	$T' \eqdef N_G(C_s(G\sminus T))$ and
	$G' \eqdef G[C_s(G\sminus T)\cup T'\cup C_t(G\sminus T)]$, we have by Lemma~\ref{lem:technicalfortworesults} (item 3) that
	$\kappa_{s,t}(G')=\kappa_{s,t}(G)$.

	Now, suppose by way of contradiction that $D\not\subseteq C_s(G\sminus T)$. Since $T\in \minsep_{sDX,t}(G)$, then $D\cap (T\cup C_t(G\sminus T))=\emptyset$, and hence $D\not\subseteq \nodes(G')$. Let $v\in D{\setminus} \nodes(G')$. Since $D\subseteq \minstVertices{s,t}(G)$, then $v\in \minstVertices{s,t}(G){\setminus} \nodes(G')$. 
	
	Since $v\in U^{\min}_{s,t}(G)$, Lemma~\ref{lem:vertexInclude} gives
	\[
	\kappa_{s,t}(G\sminus v)=\kappa_{s,t}(G)-1.
	\]
	Moreover, since $v\notin \nodes(G')$, the graph $G'$ is a subgraph of
	$G\sminus v$. Therefore,
	\[
	\kappa_{s,t}(G')
	\le \kappa_{s,t}(G\sminus v)
	=
	\kappa_{s,t}(G)-1.
	\]
	Together with $\kappa_{s,t}(G')=\kappa_{s,t}(G)$ from Lemma~\ref{lem:technicalfortworesults},
	this gives the contradiction
	\[
	\kappa_{s,t}(G)
	=
	\kappa_{s,t}(G')
	\le
	\kappa_{s,t}(G)-1.
	\]
Therefore, $D\subseteq C_s(G\sminus T)$ as required.
\end{proof}

\begin{repcorollary}{\ref{corr:minsepsAtSets}}
	\minsepsAtSetCorr
\end{repcorollary}
\begin{proof}
	Let $H$ be the graph obtained from $G$ by adding all edges between
	$s$ and $N_G[A]$; that is,
	$$
	V(H)=V(G)
	\qquad\text{and}\qquad
	E(H)=E(G)\cup \{(s,v):v\in N_G[A]\}.
	$$
	By Lemma~\ref{lem:MinlsASep}, applied with $B=\emptyset$, the minimal
	$sA,t$-separators of $G$ are exactly the minimal $s,t$-separators
	of $H$ ($\minlsepst{H}=\minlsep{sA,t}{G}$), and for every $S\in \minlsepst{H}$, $	C_s(H\sminus S)=C_{sA}(G\sminus S)$.
	In particular, the minimum separators are the same: $	\minsep_{sA,t}(G)=\minsep_{s,t}(H)$.
	Consequently, $\minstVertices{s,t}(H)=\minstVertices{sA,t}(G)$.
	Since $D\subseteq \minstVertices{sA,t}(G)$, we therefore have $D\subseteq \minstVertices{s,t}(H)$.
	
	Now let $T\in \minsep_{sADX,t}(G)$. We claim that $T\in \minsep_{sDX,t}(H)$.
	
	Let $K_G$ be the graph obtained from $G$ by adding all edges between
	$s$ and $N_G[ADX]$. By Lemma~\ref{lem:MinlsASep}, applied in $G$ with the set $ADX$,
	the minimum $sADX,t$-separators of $G$ are exactly the minimum
	$s,t$-separators of $K_G$; that is $\minsep_{sADX,t}(G)=\minsep_{s,t}(K_G)$.
	
	Let $K_H$ be the graph obtained from $H$ by adding all edges between
	$s$ and $N_H[DX]$. By Lemma~\ref{lem:MinlsASep}, applied in $H$ with the set $DX$,
	the minimum $sDX,t$-separators of $H$ are exactly the minimum
	$s,t$-separators of $K_H$; that is $\minsep_{sDX,t}(H)=\minsep_{s,t}(K_H)$.
	
	We now observe that $K_G=K_H$. Indeed, $H$ is obtained from $G$ by
	adding all edges from $s$ to $N_G[A]$. Thus the additional edges used
	to construct $K_H$ from $H$ are exactly the edges from $s$ to vertices
	in $N_G[DX]\cup N_G[A]$, together with edges already incident with $s$
	created in the first step. Hence both constructions add precisely all
	edges from $s$ to $N_G[ADX]$. Therefore $K_G=K_H$.
	
	It follows that
	\[
	\minsep_{s,t}(K_G)=\minsep_{sADX,t}(G)=\minsep_{sDX,t}(H)=\minsep_{s,t}(K_H).
	\]
	Thus $T\in \minsep_{sDX,t}(H)$.
	
	We may now apply Theorem~\ref{thm:minsepsAt} in the graph $H$, with the same set $D$
	and with the auxiliary set $X$. Since $D\subseteq \minstVertices{s,t}(H)$
	and $T\in \minsep_{sDX,t}(H)$, Theorem~\ref{thm:minsepsAt} implies that $D\subseteq C_s(H\sminus T)$.

	It remains only to identify $C_s(H\sminus T)$ with $C_{sA}(G\sminus T)$. This does
	not require minimality of $T$. Since $H$ is obtained from $G$ by adding
	edges from $s$ to every vertex of $N_G[A]$, every component of $G\sminus T$
	that contains a vertex of $sA$ is connected to $s$ in $H\sminus T$. Hence $C_{sA}(G\sminus T)\subseteq C_s(H\sminus T)$.
	Conversely, every edge of $\edges(H){\setminus}\edges(G)$ is incident with $s$ and its
	other endpoint lies in $N_G[A]$. Such an endpoint belongs either to
	$T$ or to a component of $G\sminus T$ meeting $A$. Therefore these added edges
	cannot connect $s$ in $H\sminus T$ to any component of $G\sminus T$ disjoint from
	$sA$. Hence $C_s(H\sminus T)\subseteq C_{sA}(G\sminus T)$.
	Thus $C_s(H\sminus T)=C_{sA}(G\sminus T)$.
	Consequently, $D\subseteq C_{sA}(G\sminus T)$,
	as required.
\end{proof}

\begin{repcorollary}{\ref{corr:XInvariant}}
	\XInvariant
\end{repcorollary}
\begin{proof}
	Fix $A_0\subseteq V(G)$, and let $T\in \minsep_{sA_0\calX_\ell,t}(G)$. We show that $X_i \subseteq C_s(G\sminus T)$ by induction on $i$.
	
	Base: $i=1$. Since $X_1 \subseteq \minstVertices{s,t}(G)$, then by Theorem~\ref{thm:minsepsAt}, we have that $X_1\subseteq C_s(G\sminus T)$ (i.e., take $D\defeq X_1$ and
	$X\defeq A_0\cup (\calX_\ell{\setminus} X_1)$ in Theorem~\ref{thm:minsepsAt}).
	
	Step: Let $1\leq j\leq \ell-1$, and assume the claim holds for all indices $i\leq j$, we prove for $i=j+1$.
	By the induction hypothesis, we have that $X_k\subseteq C_s(G\sminus T)$ for every $k\leq j$. In particular,
	\begin{equation}
		\label{eq:hypothesis}
		\calX_{i-1}=\bigcup_{k=1}^{j=i-1}X_k \subseteq C_s(G\sminus T).
	\end{equation}
	Since $X_i \subseteq \minstVertices{s\calX_{i-1},t}(G)$, we apply
	Corollary~\ref{corr:minsepsAtSets} with
	\[
	A_{\mathrm{cor}} \defeq \calX_{i-1},\qquad
	D_{\mathrm{cor}} \defeq X_i,\qquad
	X_{\mathrm{cor}} \defeq A_0\cup (\calX_\ell{\setminus} \calX_i).
	\]
	Indeed,
	\[
	sA_{\mathrm{cor}}D_{\mathrm{cor}}X_{\mathrm{cor}}
	=
	s\calX_{i-1}X_iA_0(\calX_\ell{\setminus} \calX_i)
	=
	sA_0\calX_\ell,
	\]
	and therefore the hypothesis $T\in \minsep_{sA_0\calX_\ell,t}(G)$ is exactly the
	hypothesis needed to apply Corollary~~\ref{corr:minsepsAtSets}. Hence $X_i\subseteq C_{s\calX_{i-1}}(G\sminus T)$.
	Since $\calX_{i-1}\subseteq C_s(G\sminus T)$ by the induction hypothesis, we have $C_{s\calX_{i-1}}(G\sminus T)=C_s(G\sminus T)$.
	Therefore, $X_i\subseteq C_s(G\sminus T)$.
\end{proof}

\subsection{Proof of Theorem~\ref{thm:mainThmBoundedCardinality}}
\label{sec:thm_mainThmBoundedCardinalityProof}
We prove Theorem~\ref{thm:mainThmBoundedCardinality} by a series of Lemmas.
\def\minVerticesInMainGraph{
	Let $T\in \minsepst{G}$. Then $\minstVertices{s,t}(G)\subseteq C_s(G\sminus T)\cup T \cup C_t(G\sminus T)$.
}
\begin{lemma}
	\label{lem:minVerticesInMainGraph}
	\minVerticesInMainGraph
\end{lemma}
\begin{proof}
	By Menger's Theorem (Theorem~\ref{thm:Menger}), there exist $k\eqdef \kappa_{s,t}(G)$ pairwise internally vertex-disjoint paths from $s$ to $t$ in $G$ denoted $P_1,\dots,P_k$. Each one of these $k$ paths meets exactly one vertex in $T$. To see why, observe that every $s,t$-path in $G$ must meet at least one vertex in $T$. On the other hand, if some path $P_i$ meets two or more vertices in $T$, then the remaining $k-1$ pairwise internally vertex-disjoint paths  must meet at most $k-2$ vertices of $T$. By the pigeon-hole principle, some pair of paths meet the same vertex of $T$, contradicting the fact that they are pairwise-vertex-disjoint. 
	Consequently, for every path $P_i$, we have that $\nodes(P_i)\subseteq C_s(G\sminus T)\cup T \cup C_t(G\sminus T)$. This, in turn, means that the induced graph $G[C_s(G\sminus T)\cup T \cup C_t(G\sminus T)]$ has $k$ pairwise internally vertex-disjoint $s,t$ paths. By Menger's Theorem, we have that $\kappa_{s,t}(G[C_s(G\sminus T)\cup T \cup C_t(G\sminus T)])\geq k$. 
	
	Now, suppose by way of contradiction, that there is a vertex $v \in \minstVertices{s,t}(G) {\setminus} (C_s(G\sminus T)\cup T \cup C_t(G\sminus T))$. By Lemma~\ref{lem:vertexInclude}, we have that:
	\begin{align}
		k= \kappa_{s,t}(G) > \kappa_{s,t}(G\sminus v) \geq  \kappa_{s,t}(G[C_s(G\sminus T)\cup T \cup C_t(G\sminus T)])\geq k
	\end{align}
	but this is clearly a contradiction.
\end{proof}

\eat{
\def\lemmainThmBoundedCardinality{
	Let $\emptyset \subset A\subseteq \nodes(G){\setminus}\set{s,t}$, $Q\subseteq \nodes(G){\setminus}\set{s,t}$, and $\P\eqdef \set{A_1,\dots,A_M}$ pairwise disjoint subsets of $sA$, and $\mc{R}_{sA,\P,Q}$ a set of monotonic connectivity constraints (see~\eqref{eq:mcR}). Letting $L^*\eqdef \closestMinSep{sAQ,t}(G)$ for brevity, define the set of violating components $\D$ in $G\sminus L^*$:
	\begin{align*}
		\D \eqdef &\set{C\in \cc(G\sminus  L^*): s\notin C, \exists i\in [1,M] \text{ s.t. }\emptyset\subset C \cap A_i\subset A_i}\\
		&~~~~~\cup \set{C\in \cc(G\sminus  L^*): C\cap Q\neq \emptyset, sA\cap C=\emptyset}.
	\end{align*}
	If $\kappa_{s,t}(G)=\kappa_{sA,t}(G)=f_{sA,t}(G,\mc{R})$ then
	\begin{align*}
		\emptyset \subset \safeSepsMin{sA}{t}(G,\mc{R})=\set{S\in \minsep_{sAQ,t}(G): \mediumbigcup_{C\in \D}C \subseteq C_s(G\sminus S)}.
	\end{align*}
}
}
\def\lemmainThmBoundedCardinality{
	Let $\emptyset\subset A\subseteq\nodes(G){\setminus}\set{s,t}$. Suppose that $\kappa_{s,t}(G)=\kappa_{sA,t}(G)=f_{s,t}(G,A)<\infty$. Letting $L^*\eqdef\closestMinSep{sA,t}(G)$ for brevity, define the set of violating components $\D$ in $G\sminus L^*$:
	\begin{align*}
		\D \eqdef &\set{C\in \cc(G\sminus  L^*): A\cap C \neq \emptyset, s\notin C}.
	\end{align*}
	Then
	\begin{align*}
		\emptyset \subset \safeSepsMin{s}{t}(G,A)=\set{S\in \minsep_{sA,t}(G): \mediumbigcup_{C\in \D}C \subseteq C_s(G\sminus S)}.
	\end{align*}
}
\begin{lemma}
	\label{lem:mainThmBoundedCardinality}
	\lemmainThmBoundedCardinality
\end{lemma}	
\begin{proof}
	Since $f_{s,t}(G,A)<\infty$, the family
	$\safeSepsMin{s}{t}(G,A)$ is nonempty. Moreover, every separator in
	$\safeSepsMin{s}{t}(G,A)$ has cardinality $\kappa_{sA,t}(G)$ and is an
	$sA,t$-separator, and hence belongs to $\minsep_{sA,t}(G)$. In other words, since $\kappa_{sA,t}(G)=f_{s,t}(G,A)$, then $\safeSepsMin{s}{t}(G,A)\subseteq \minsep_{sA,t}(G)$.
	
	Let $L^*=L^*_{sA,t}(G)$. By Lemma~\ref{lem:uniqueMinImportantSep}, $L^*$ is the unique minimum
	important $sA,t$-separator. Equivalently, for every
	$S\in \minsep_{sA,t}(G)$, it holds that $C_{sA}(G\sminus L^*)\subseteq C_{sA}(G\sminus S)$.
	
	We first prove the inclusion from left to right. Let
	$S\in \safeSepsMin{s}{t}(G,A)$. Then $A\subseteq C_s(G\sminus S)$ and
	\[
	|S|=f_{s,t}(G,A)=\kappa_{sA,t}(G).
	\]
	Thus $S\in \minsep_{sA,t}(G)$. Since also
	$f_{s,t}(G,A)=\kappa_{s,t}(G)$, we have $S\in \minsep_{s,t}(G)$.
	
	Let $C\in\mathcal D$. Since $C$ is a component of $G\sminus L^*$ that intersects
	$A$, and $L^*$ is an $sA,t$-separator, the component $C$ is distinct
	from $C_t(G\sminus L^*)$. It is also distinct from $C_s(G\sminus L^*)$ by the
	definition of $\mathcal D$. Since $S\in \minsep_{s,t}(G)$, we have
	$S\subseteq \minstVertices{s,t}(G)$. By Lemma~\ref{lem:minVerticesInMainGraph}, applied with
	$T=L^*$,
	\[
	S\subseteq \minstVertices{s,t}(G)\underset{\text{Lem.}~\ref{lem:minVerticesInMainGraph}}{\subseteq} C_s(G\sminus L^*)\cup L^*\cup C_t(G\sminus L^*).
	\]
	Therefore $S\cap C\subseteq 	\minstVertices{s,t}(G)\cap C=\emptyset$.
	Because $C$ is connected, $C\cap A\neq\emptyset$, and
	$A\subseteq C_s(G\sminus S)$, it follows that $	C\subseteq C_s(G\sminus S)$.
	Thus $	\bigcup_{C\in\mathcal D} C\subseteq C_s(G\sminus S)$.
	This proves
	\[
	\safeSepsMin{s}{t}(G,A)
	\subseteq
	\set{S\in \minsep_{sA,t}(G): \bigcup_{C\in\mathcal D} C\subseteq C_s(G-S)}.
	\]
	
	For the converse inclusion, let $S\in \minsep_{sA,t}(G)$ where $\bigcup_{C\in \D} C\subseteq C_s(G\sminus S)$.
	Since $\kappa_{sA,t}(G)=\kappa_{s,t}(G)$, the separator $S$ has
	cardinality $\kappa_{s,t}(G)$ and separates $s$ from $t$. Hence
	$S\in \minsep_{s,t}(G)$, and in particular $S\in \minlsepst{G}$.
	
	It remains to show that $A\subseteq C_s(G\sminus S)$. Let $a\in A$. If
	$a\notin C_s(G\sminus L^*)$, then the connected component of $G\sminus L^*$ that
	contains $a$ belongs to $\D$, and therefore
	$a\in C_s(G\sminus S)$ by assumption.
	
	Otherwise, $a\in C_s(G\sminus L^*)$. By the source-side minimality of $L^*$, $C_{sA}(G\sminus L^*)\subseteq C_{sA}(G\sminus S)$.
	In particular, $C_s(G\sminus L^*)\subseteq C_{sA}(G\sminus S)$. Since
	$C_s(G\sminus L^*)$ contains $s$ and is connected, and since $S$ is disjoint
	from $C_s(G\sminus L^*)$, this inclusion actually places
	$C_s(G\sminus L^*)$ inside the $s$-component of $G\sminus S$; that is $C_s(G\sminus L^*)\subseteq C_s(G\sminus S)$. Hence	$a\in C_s(G\sminus S)$.
	
	Thus every $a\in A$ belongs to $C_s(G\sminus S)$, so
	$S\in \safeSeps{s}{t}(G,A)$. Since $|S|=\kappa_{sA,t}(G)=f_{s,t}(G,A)$,
	we get that $S\in \safeSepsMin{s}{t}(G,A)$.
	This proves the reverse inclusion and completes the proof.
\end{proof}

\begin{lemma}
		\label{lem:minSafeSepWithUstMin}
	Let $R\in \minsep_{s,t}(G)$, and let $D\in \cc(G\sminus R)$ be a connected
	component such that $s\notin D$ and $t\notin D$. Define
	$T_D\defeq N_G(D)$. Then, for every nonempty $A\subseteq D$, it holds that
	\[
	\set{S\in \minsep_{sA,t}(G): A\subseteq C_s(G-S)}
	=
	\set{S\in \minsep_{s,t}(G): T_D\cap C_s(G\sminus S)\neq\emptyset}.
	\]
\end{lemma}
\begin{proof}
	First observe that $R$ is an $sA,t$-separator, because $A\subseteq D$,
	$s\notin D$, $t\notin D$, and $D$ is a connected component of $G\sminus R$.
	Therefore $	\kappa_{sA,t}(G)\le |R|=\kappa_{s,t}(G)$.
	The reverse inequality $\kappa_{s,t}(G)\le \kappa_{sA,t}(G)$ is immediate,
	since every $sA,t$-separator separates $s$ from $t$. Hence
	\[
	\kappa_{sA,t}(G)=\kappa_{s,t}(G).
	\]
	Consequently, $\minsep_{sA,t}(G)\subseteq \minsep_{s,t}(G)$.
	
	We next record two simple consequences of Lemma~\ref{lem:minVerticesInMainGraph}. Since
	$T_D=N_G(D)\subseteq R\in \minsep_{s,t}(G)$, we have $T_D\subseteq \minstVertices{s,t}(G)$.
	Thus, for every $S\in \minsep_{s,t}(G)$, Lemma~\ref{lem:minVerticesInMainGraph} applied to the minimum
	separator $R$ gives
	\[
	S\subseteq \minstVertices{s,t}(G)
	\subseteq C_s(G\sminus R)\cup R\cup C_t(G\sminus R).
	\]
	Since $D\in \cc(G\sminus R)$ is a connected component of $G\sminus R$ distinct from both
	$C_s(G\sminus R)$ and $C_t(G\sminus R)$, it follows that $	S\cap D=\emptyset$.
	In particular, for every $S\in \minsep_{s,t}(G)$, the component $D$ remains
	connected in $G\sminus S$.
	We now prove the two inclusions.
	
	Let $S\in \minsep_{sA,t}(G)$ where $A\subseteq C_s(G\sminus S)$.
	As observed above, $S\in \minsep_{s,t}(G)$ and $S\cap D=\emptyset$. Since
	$A\subseteq D$ and $D$ is connected in $G\sminus S$, the whole component $D$
	is contained in $C_s(G\sminus S)$.
\rev{Since $s\notin D$, a path from $D$ to $s$ in $G\sminus S$
leaves $D$ through some vertex $v\in N_G(D)\sminus S$.
Thus $v\in T_D\cap C_s(G\sminus S)$; no connectedness assumption on
$G$ is needed.}
	Therefore $T_D\cap C_s(G\sminus S)\neq\emptyset.$

	Conversely, let $S\in \minsep_{s,t}(G)$ where $T_D\cap C_s(G\sminus S)\neq \emptyset$.
	Choose $v\in T_D\cap C_s(G\sminus S)$. Since $v\in T_D=N_G(D)$, there is a
	neighbor $u\in D$ of $v$. As shown above, $S\cap D=\emptyset$, so $D$
	is connected in $G\sminus S$. Hence every vertex of $A\subseteq D$ is connected
	to $u$, and therefore to $v$, in $G\sminus S$. Since $v\in C_s(G\sminus S)$, it follows
	that $	A\subseteq C_s(G\sminus S)$.
	Thus $S$ is an $sA,t$-separator. Since $	|S|=\kappa_{s,t}(G)=\kappa_{sA,t}(G)$,
	we get $S\in \minsep_{sA,t}(G)$. Hence
	\[
	S\in \set{S\in \minsep_{sA,t}(G): A\subseteq C_s(G\sminus S)}.
	\]
	The two inclusions prove the equality.
\end{proof}

\begin{lemma}
	\label{lem:minSafeSepImportant}
	Let $R\in \minsep_{s,t}(G)$, and let $D\in \cc(G\sminus R)$ where $s\notin D$ and $t\notin D$. Define $T_D\eqdef N_G(D)$. For every $\emptyset \subset A\subseteq D$, it holds that:
	
	\begin{align}
		\label{eq:minSafeSepImportant}
		\set{S\in \minsep_{sA,t}(G): A\subseteq C_s(G\sminus S)}=\mediumbigcup_{v \in T_D}\set{S \in \minsep_{sv,t}(G): |S|=\kappa_{s,t}(G)}.
	\end{align}
\end{lemma}
\begin{proof}
	By Lemma~\ref{lem:minSafeSepWithUstMin}, it holds that:
	\begin{align}
		\set{S\in \minsep_{sA,t}(G): A\subseteq C_s(G\sminus S)}=\mediumbigcup_{v\in T_D} \set{S\in \minsep_{s,t}(G): v\in C_s(G\sminus S)}. \label{eq:minSafeSepImportant1}
	\end{align}
	
	Since $T_D \subseteq R \in \minsep_{s,t}(G)$, then $T_D\subseteq \minstVertices{s,t}(G)$. By Theorem~\ref{thm:minsepsAt}, for every $v\in T_D\subseteq \minstVertices{s,t}(G)$:
	\begin{align*}
		\minsep_{sv,t}(G) &\subseteq \set{S\in \minlsepst{G}: v\in C_s(G\sminus S)},&& \text{and since $\minsep_{s,t}(G)\subseteq \minlsepst{G}$, then} \\
		\minsep_{s,t}(G) \cap \minsep_{sv,t}(G) &= \set{S\in \minsep_{s,t}(G): v\in C_s(G\sminus S)} && 
	\end{align*}
	Noting that $\minsep_{s,t}(G) \cap \minsep_{sv,t}(G)=\set{S\in \minsep_{sv,t}(G) : |S|=\kappa_{s,t}(G)}$, we have that:
	$$\mediumbigcup_{v\in T_D}\set{S\in \minsep_{sv,t}(G): |S|=\kappa_{s,t}(G)}=\mediumbigcup_{v\in T_D}\set{S\in \minsep_{s,t}(G): v\in C_s(G\sminus S)},$$ 
	and the claim follows from~\eqref{eq:minSafeSepImportant1}.
\end{proof}

\begin{reptheorem}{\ref{thm:mainThmBoundedCardinality}}
	\mainThmBoundedCardinality
\end{reptheorem}	
\begin{proof}
	Let $k\defeq \kappa_{s,t}(G)$.
	By assumption,
	\[
	k=\kappa_{s,t}(G)=\kappa_{sA,t}(G)=f_{s,t}(G,A).
	\]
	Let $L^*\eqdef L^*_{sA,t}(G)$.
	For each $C\in \D$, since $C$ is a connected component of
	$G\sminus L^*$, we have $N_G(C)\subseteq L^*$.
	Thus every set in $\varepsilon_{\D}$ is a subset of $L^*$.
	
	First, $\safeSepsMin{s}{t}(G,A)$ is nonempty by the assumption
	$f_{s,t}(G,A)=k<\infty$. Hence
	$\safeSepskImp{s}{t}{\min}(G,A)$ is also nonempty:
	indeed, choose $S\in \safeSepsMin{s}{t}(G,A)$ such that
	$C_s(G\sminus S)$ is inclusionwise minimal among all sets
	$C_s(G\sminus T)$ with $T\in \safeSepsMin{s}{t}(G,A)$.
	If $S$ were not important in $\safeSeps{s}{t}(G,A)$, then there
	would be $S'\in \safeSeps{s}{t}(G,A)$ such that
	\[
	C_s(G\sminus S')\subsetneq C_s(G\sminus S)
	\qquad\text{and}\qquad
	|S'|\le |S|.
	\]
	Since $S$ has minimum cardinality in $\safeSeps{s}{t}(G,A)$, this
	would imply $S'\in \safeSepsMin{s}{t}(G,A)$, contradicting the
	choice of $S$. Therefore
	$S\in \safeSepskImp{s}{t}{\min}(G,A)$, and so $	\safeSepskImp{s}{t}{\min}(G,A)\neq \emptyset$.

	We now prove the desired inclusion. Let $S\in \safeSepskImp{s}{t}{\min}(G,A)$.
	Since $S\in \safeSepsMin{s}{t}(G,A)$ and
	$f_{s,t}(G,A)=\kappa_{s,t}(G)=\kappa_{sA,t}(G)$,
	Lemma~\ref{lem:mainThmBoundedCardinality} gives
	\[
	S\in \minsep_{sA,t}(G)
	\qquad\text{and}\qquad
	\bigcup_{C\in \D} C\subseteq C_s(G\sminus S).
	\]
	In particular, for every $C\in \D$, $	C\subseteq C_s(G\sminus S)$.
	Fix a component $C\in \D$, and let $A_C\eqdef A\cap C$.
	Then $A_C$ is nonempty. Since $L^*\in \minsep_{sA,t}(G)$ and
	$\kappa_{sA,t}(G)=\kappa_{s,t}(G)$, we have
	$L^*\in \minsep_{s,t}(G)$.
	Moreover, $C$ is a connected component of $G\sminus L^*$ with
	$s,t\notin C$. Therefore Lemma~\ref{lem:minSafeSepImportant},
	applied with $R\eqdef L^*$, $D\eqdef C$, and $A\eqdef A_C$, gives
	\begin{equation}
		\label{eq:mainProofEq1}
		\set{T\in \minsep_{sA_C,t}(G): A_C\subseteq C_s(G\sminus T)}
		=
		\bigcup_{v\in N_G(C)}
		\set{T\in \minsep_{sv,t}(G): |T|=k}.
	\end{equation}
	Now $S\in \minsep_{sA,t}(G)$, hence $S$ is an $sA_C,t$-separator. Hence, $|S|\geq \kappa_{sA_C,t}(G)$.
	Since $|S|=\kappa_{s,t}(G) \leq \kappa_{sA_C,t}(G)$, then 
	$|S|=k=\kappa_{sA_C,t}(G)$; it follows that $S\in \minsep_{sA_C,t}(G)$.
	Also $A_C\subseteq C\subseteq C_s(G\sminus S)$. Thus $S$ belongs
	to the left-hand side of~\eqref{eq:mainProofEq1}. Consequently,
	there exists a vertex $v_C\in N_G(C)$ such that $S\in \minsep_{sv_C,t}(G)$ and 	$|S|=k$.

	Choose one such vertex $v_C$ for every $C\in\D$, and define
	\[
	Z\defeq \set{v_C:C\in \D}.
	\]
	Then $Z$ is a hitting set of $\varepsilon_{\D}$. Let
	$Y\subseteq Z$ be an inclusionwise minimal hitting set of
	$\varepsilon_{\D}$. Then $Y\in \MHS(\varepsilon_{\D})$.
	Since $Y\subseteq Z$, and since $S\in \minsep_{sv_C,t}(G)$ for
	every $v_C\in Z$, the set $S$ separates $sY$ from $t$. Since also
	$S\in \minsep_{sA,t}(G)$, it separates $sAY$ from $t$.
	Moreover,
	\[
	\kappa_{s,t}(G)\le \kappa_{sAY,t}(G)\le |S|=k=\kappa_{s,t}(G).
	\]
	Therefore $	S\in \minsep_{sAY,t}(G)$.

	Let
	\[
	L_Y\eqdef L^*_{sAY,t}(G).
	\]
	Since $S\in \minsep_{sAY,t}(G)$, we get $|L_Y|=\kappa_{sAY,t}(G)=k$.
	We claim that $L_Y$ is a minimum CP $s,t$-separator with respect
	to $A$.
	
	First observe that $	Y\subseteq \bigcup_{C\in\D} N_G(C)\subseteq L^*$.
	Since $L^*\in \minsep_{s,t}(G)$, this implies $	Y\subseteq L^* \subseteq \minstVertices{s,t}(G)$.
	By Theorem~\ref{thm:minsepsAt}, applied with $D\eqdef Y$ and
	$X\eqdef A$, every separator in $\minsep_{sAY,t}(G)$ places $Y$
	on the $s$-side. In particular, $	Y\subseteq C_s(G\sminus L_Y)$.

	Now let $C\in\D$. Since $Y$ is a hitting set of
	$\varepsilon_{\D}$, there exists $y_C\in Y\cap N_G(C)$.
	By the previous paragraph, 	$y_C \in Y \subseteq C_s(G\sminus L_Y)$; thus $y_C\in C_s(G\sminus L_Y)$.

	We also have $L_Y\cap C=\emptyset$. Indeed, since
	$L_Y\in \minsep_{sAY,t}(G)$ and $|L_Y|=k=\kappa_{s,t}(G)$, the
	set $L_Y$ is also a minimum $s,t$-separator; that is, $	L_Y\in \minsep_{s,t}(G)$.
	Therefore $	L_Y\subseteq \minstVertices{s,t}(G)$.
	Applying Lemma~\ref{lem:minVerticesInMainGraph} with
	$T\eqdef L^*$ gives
	\[
	L_Y \subseteq \minstVertices{s,t}(G)
	\subseteq
	C_s(G\sminus L^*)\cup L^*\cup C_t(G\sminus L^*).
	\]
	The component $C$ is a component of $G\sminus L^*$ distinct from
	both $C_s(G\sminus L^*)$ and $C_t(G\sminus L^*)$, and it is
	disjoint from $L^*$. Hence $	C\cap \minstVertices{s,t}(G)=\emptyset$,
	and so $L_Y\cap C=\emptyset$.

	Since $C$ is connected, $L_Y\cap C=\emptyset$, and
	$y_C\in N_G(C)$ belongs to $C_s(G\sminus L_Y)$, it follows that $C\subseteq C_s(G\sminus L_Y)$.
	As this holds for every $C\in \D$, and since
	$L_Y\in \minsep_{sA,t}(G)$ because $L_Y\in \minsep_{sAY,t}(G)$
	and $|L_Y|=k=\kappa_{sA,t}(G)$, Lemma~\ref{lem:mainThmBoundedCardinality}
	implies $	L_Y\in \safeSepsMin{s}{t}(G,A)$.

	We now compare $L_Y$ with the original separator $S$. Since both
	$S$ and $L_Y$ belong to $\minsep_{sAY,t}(G)$, and
	$L_Y=L^*_{sAY,t}(G)$ is the unique minimum important
	$sAY,t$-separator, by Lemma~\ref{lem:uniqueMinImportantSep} we have
	$C_{sAY}(G\sminus L_Y)\subseteq C_{sAY}(G\sminus S)$.
	Since $S\in \safeSepsMin{s}{t}(G,A)$ and $S\in\minsep_{sAY,t}(G)$,
	we have $A\cup Y\subseteq C_s(G\sminus S)$,
	and therefore $	C_{sAY}(G\sminus S)=C_s(G\sminus S)$.
	Similarly, since $L_Y\in \safeSepsMin{s}{t}(G,A)$ and
	$Y\subseteq C_s(G\sminus L_Y)$, we have $C_{sAY}(G\sminus L_Y)=C_s(G\sminus L_Y)$.
	Hence $	C_s(G\sminus L_Y)\subseteq C_s(G\sminus S)$.

	Since $L_Y\in \safeSepsMin{s}{t}(G,A)$ and
	$S\in \safeSepskImp{s}{t}{\min}(G,A)$, the last inclusion cannot
	be strict. Otherwise $L_Y$ would be a CP $s,t$-separator with
	\[
	C_s(G\sminus L_Y)\subsetneq C_s(G\sminus S)
	\qquad\text{and}\qquad
	|L_Y|=|S|=k,
	\]
	contradicting the importance of $S$ in $\safeSeps{s}{t}(G,A)$ (i.e., that $S\in \safeSepskImp{s}{t}{\min}(G,A)$).
	Therefore $	C_s(G\sminus L_Y)=C_s(G\sminus S)$.
	Since both $L_Y$ and $S$ are minimal $s,t$-separators,
	Lemma~\ref{lem:fullComponents} gives
	\[
	L_Y
	=
	N_G(C_s(G\sminus L_Y))
	=
	N_G(C_s(G\sminus S))
	=
	S.
	\]
	Thus $	S=L_Y=L^*_{sAY,t}(G)$,
	where $Y\in \MHS(\varepsilon_{\D})$ and $|L^*_{sAY,t}(G)|=k=\kappa_{s,t}(G)$.
	Since $S\in \safeSepskImp{s}{t}{\min}(G,A)$ was arbitrary, we
	conclude that
	\[
	\safeSepskImp{s}{t}{\min}(G,A)
	\subseteq
	\bigcup_{\substack{Y\in\MHS(\varepsilon_{\D})\\\kappa_{sAY,t}(G)=\kappa_{s,t}(G)}}\set{L^*_{sAY,t}(G)}.
	\]
	
	It remains to justify the enumeration bound. Since every member of
	$\varepsilon_{\D}$ is a subset of $L^*$, every minimal hitting set
	of $\varepsilon_{\D}$ is also a subset of
	$\bigcup_{C\in\D}N_G(C)\subseteq L^*$.
	Moreover,
	\[
	|L^*|=\kappa_{sA,t}(G)=\kappa_{s,t}(G)=k.
	\]
	Therefore $	|\MHS(\varepsilon_{\D})|\le 2^k$.
	
\begin{revision}
We enumerate all subsets of $K\eqdef L^*$ and keep the inclusionwise
minimal hitting sets of $\varepsilon_{\D}$. For each such $Y$, test
whether $\kappa_{sAY,t}(G)=k$; if so, compute
$L_Y\eqdef\closestMinSep{sAY,t}(G)$ and retain it.
The preceding proof that $L_Y$ is CP uses only that $Y$ hits
$\varepsilon_{\D}$ and $|L_Y|=k$, so it applies to every retained candidate.
The containment proved above shows that all important minimum CP
separators occur among them. It remains to remove repetitions and
nonimportant candidates without spending $4^k$ time.

For each retained candidate define
\[
R_Y\eqdef C_s(G\sminus L_Y),\qquad B_Y\eqdef K\cap R_Y.
\]
We claim that for any two retained candidates,
\begin{equation}
\label{eq:signature-inclusion}
R_Y\subseteq R_Z\quad\Longleftrightarrow\quad B_Y\subseteq B_Z.
\end{equation}
The forward implication is immediate. For the converse, suppose
$B_Y\subseteq B_Z$. Theorem~\ref{thm:minsepsAt} gives $Y\subseteq B_Y$,
so $Y\subseteq R_Z$. Since $L_Z$ is CP, $A\cup Y\subseteq R_Z$.
Thus $L_Z$ is an $sAY,t$-separator of size $k$, and hence a minimum one.
By canonicality of $L_Y$,
\[
C_{sAY}(G\sminus L_Y)\subseteq C_{sAY}(G\sminus L_Z).
\]
All vertices of $A\cup Y$ lie in the actual $s$-component on both sides,
so this is $R_Y\subseteq R_Z$, proving~\eqref{eq:signature-inclusion}.
Consequently equal signatures give equal source components and, by
minimality, equal separators. Strict inclusion of signatures is exactly
strict inclusion of source components.

Fix an ordering of $K$ and represent each subset by its $k$-bit incidence vector, called its mask; a submask represents a subset. Set $h(B)=1$ if signature $B$
occurs, and $h(B)=0$ otherwise. A mask $B$ is occupied when $h(B)=1$.
Starting with $d=h$, for each $v\in K$
and each mask $B$ containing $v$, update
$d(B)\gets d(B)\lor d(B\sminus\{v\})$.
After these $k2^{k-1}$ updates, $d(B)$ records whether $B$ contains any
occupied submask. An occupied $B$ has an occupied proper submask exactly
when $d(B\sminus\{v\})=1$ for some $v\in B$. Retain one candidate for
each occupied mask with no occupied proper submask.

This returns precisely the important candidates. Indeed, an important
candidate has no strict dominator. Conversely, if a candidate is not
important, repeatedly choose a no-larger CP separator with a strictly
smaller source component. This process terminates because the graph is
finite. All separators in this descent have size $k$, since $k$ is the
minimum separator size. Its final separator is important and thus occurs
among our candidates. By~\eqref{eq:signature-inclusion}, its signature
is an occupied proper submask of the original signature, which is
therefore discarded.

The violating boundaries can be represented in $O(n+m)$ time. Checking
whether a subset of $K$ is a minimal hitting set takes $O(kn)$ time:
it must hit every boundary, and each selected vertex must be the unique
hit of at least one boundary. By Lemma~\ref{lem:uniqueMinImportantSep},
computing $L_Y$ takes $O(n(n+m))$ time: every finite separator size
is at most $n$. Since $T(n,m)$ includes input reading, this is
$O(n\cdot T(n,m))$. Constructing its signature takes $O(n+m)$ time. The subset dynamic
program takes $O(k2^k)$ time. Since $k\le n$ and $T(n,m)$ includes
input reading, the total running time is
\[
O(2^k\cdot n\cdot T(n,m))
=O(2^{\kappa_{s,t}(G)}\cdot n\cdot T(n,m)).
\]
\end{revision}
	This completes the proof.
\end{proof}

\begin{repcorollary}{\ref{corr:mainThmBoundedCardinalityExt}}
	\mainThmBoundedCardinalityExt
\end{repcorollary}
\begin{proof}
	Let $H$ be the graph obtained from $G$ by adding all edges between
	$s$ and $N_G[X]$; that is,
	\[
	\nodes(H)=\nodes(G)
	\qquad\text{and}\qquad
	\edges(H)=\edges(G)\cup\set{(s,v):v\in N_G[X]}.
	\]
	
	We shall use the following consequence of Lemma~\ref{lem:MinlsASep}. For every
	$Z\subseteq \nodes(G){\setminus}\set{s,t}$, the minimal $sXZ,t$-separators of
	$G$ are exactly the minimal $sZ,t$-separators of $H$; that is $\minlsep{sXZ,t}{G}=\minlsep{sZ,t}{H}$. Moreover, for
	every such separator $S\in \minlsep{sZ,t}{H}$, $C_{sZ}(H\sminus S)=C_{sXZ}(G\sminus S)$.
	Indeed, applying Lemma~\ref{lem:MinlsASep} to $G$ with the set $XZ$ reduces
	$sXZ,t$-separators in $G$ to $s,t$-separators in the graph obtained
	from $G$ by adding all edges from $s$ to $N_G[XZ]$. Applying
	Lemma~~\ref{lem:MinlsASep} to $H$ with the set $Z$ reduces $sZ,t$-separators in $H$ to
	$s,t$-separators in the graph obtained from $H$ by adding all edges
	from $s$ to $N_H[Z]$. These two resulting graphs have the same
	edge set, namely $	\edges(G)\cup \set{(s,u):u\in N_G[XZ]}$.
	The component identity follows from Lemma~\ref{lem:MinlsASep}.
	Therefore:
	\begin{equation}
		\label{eq:transferClaim}
		\minlsep{sXZ,t}{G}=\minlsep{sZ,t}{H}
		\qquad\text{and}\qquad
		C_{sZ}(H\sminus S)=C_{sXZ}(G\sminus S).\tag{1}
	\end{equation}
	Taking $Z=\emptyset$ gives
	\[
	\minsep_{sX,t}(G)=\minsep_{s,t}(H)
	\qquad\text{and}\qquad
	\kappa_{sX,t}(G)=\kappa_{s,t}(H).
	\]
	Taking $Z=A$ gives
	\[
	\minsep_{sAX,t}(G)=\minsep_{sA,t}(H)
	\qquad\text{and}\qquad
	\kappa_{sAX,t}(G)=\kappa_{sA,t}(H).
	\]
	Therefore, by the assumptions of the corollary,
	\[
	\kappa_{s,t}(H)
	=
	\kappa_{sX,t}(G)
	=
	\kappa_{sAX,t}(G)
	=
	\kappa_{sA,t}(H).
	\]
	
	We next prove that $f_{s,t}(H,A)=f_{s,t}(G,AX)$.
	Let $S\in \safeSepsMin{s}{t}(G,AX)$. Then $	|S|=f_{s,t}(G,AX)=\kappa_{sX,t}(G)$.
	Since $AX\subseteq C_s(G\sminus S)$, the set $S$ is an $sX,t$-separator in
	$G$ where $|S|=\kappa_{sX,t}(G)$. Hence $S\in \minsep_{sX,t}(G)$. By eq.~\eqref{eq:transferClaim}, applied with
	$Z=\emptyset$, $S\in \minsepst{H}$, and
	\[
	C_s(H\sminus S)=C_{sX}(G\sminus S)\underset{X\subseteq C_s(G\sminus S)}{=}C_s(G\sminus S),
	\]
	where the last equality follows from $X\subseteq C_s(G\sminus S)$. Since
	$A\subseteq C_s(G\sminus S)$, we get $A\subseteq C_s(H\sminus S)$. Thus
	$S\in \safeSeps{s}{t}(H,A)$, and hence $f_{s,t}(H,A)\le |S|=f_{s,t}(G,AX)$.
	
	Conversely, let $S\in \safeSepsMin{s}{t}(H,A)$. Then $|S|=f_{s,t}(H,A)$.
	Since every CP $s,t$-separator in $H$ is an $s,t$-separator in $H$, $f_{s,t}(H,A)\ge \kappa_{s,t}(H)=\kappa_{sX,t}(G)$.
	On the other hand, the previous paragraph showed that
	$f_{s,t}(H,A)\le f_{s,t}(G,AX)=\kappa_{sX,t}(G)$. Therefore $	|S|=f_{s,t}(H,A)=\kappa_{sX,t}(G)$.
	By eq.~\eqref{eq:transferClaim} applied with $Z=\emptyset$, $S$ is an $sX,t$-separator
	in $G$ where $|S|=\kappa_{sX,t}(G)$, and hence $S\in \minsep_{sX,t}(G)$. By the hypothesis of the
	corollary, $X\subseteq C_s(G\sminus S)$.
	Again using eq.~\eqref{eq:transferClaim},
	\begin{equation}
		\label{eq:componentIdentity}
		C_s(H\sminus S)=C_{sX}(G\sminus S)=C_s(G\sminus S).\tag{2}
	\end{equation}
	Since $A\subseteq C_s(H\sminus S)$, we get $AX\subseteq C_s(G\sminus S)$.
	
	Moreover, $S\in \minlsepst{G}$. Indeed, since
	$S$ is a minimal $sX,t$-separator in $G$, every vertex of $S$ has a
	neighbor in $C_{sX}(G\sminus S)$ and a neighbor in $C_t(G\sminus S)$.
	Because $X\subseteq C_s(G\sminus S)$, we have
	$C_{sX}(G\sminus S)=C_s(G\sminus S)$. Hence every vertex of $S$ has a neighbor in
	$C_s(G\sminus S)$ and a neighbor in $C_t(G\sminus S)$, and therefore, by Lemma~\ref{lem:fullComponents},
	$S\in \minlsepst{G}$. Thus $S\in \safeSeps{s}{t}(G,AX)$, and since
	\[
	|S|=\kappa_{sX,t}(G)=f_{s,t}(G,AX),
	\]
	we get $S\in \safeSepsMin{s}{t}(G,AX)$.
	
	We have shown
	\[
	\safeSepsMin{s}{t}(H,A)=\safeSepsMin{s}{t}(G,AX).
	\]
	We now show that this equality restricts to the subfamilies of
	important minimum CP separators; namely, $\safeSepskImp{s}{t}{\min}(H,A)=\safeSepskImp{s}{t}{\min}(G,AX)$. Let $S\in \safeSepskImp{s}{t}{\min}(G,AX)$. By the equality $\safeSepsMin{s}{t}(H,A)=\safeSepsMin{s}{t}(G,AX)$, we also have that $S\in \safeSepsMin{s}{t}(H,A)$.
	For every
	separator $S'\in \safeSepsMin{s}{t}(H,A)$, the component identities
	above (eq.~\eqref{eq:componentIdentity}) give $C_s(H\sminus S')=C_s(G\sminus S')$.
	Indeed, every such $S'$ has size $\kappa_{sX,t}(G)$ and is an
	$sX,t$-separator in $G$, and hence, by the hypothesis of the corollary,
	$X\subseteq C_s(G\sminus S')$.
	
	Suppose first that $S\notin \safeSepskImp{s}{t}{\min}(H,A)$. Then there
	exists $S'\in \safeSeps{s}{t}(H,A)$ such that
	\[
	C_s(H\sminus S')\subsetneq C_s(H\sminus S)
	\qquad\text{and}\qquad
	|S'|\le |S|.
	\]
	Since $S\in \safeSepsMin{s}{t}(H,A)$, it has minimum cardinality in $\safeSeps{s}{t}(H,A)$. Hence the inequality
	$|S'|\le |S|$ implies $S'\in \safeSepsMin{s}{t}(H,A)$. By the equality $\safeSepsMin{s}{t}(H,A)=\safeSepsMin{s}{t}(G,AX)$, we have 
	$S'\in \safeSepsMin{s}{t}(G,AX)$, and the component identity gives $C_s(G\sminus S')\subsetneq C_s(G\sminus S)$.
	Thus $S\notin \safeSepskImp{s}{t}{\min}(G,AX)$; a contradiction.
	
	The converse implication is identical, exchanging $G$ and $H$.
	Therefore,
	\[
	\safeSepskImp{s}{t}{\min}(H,A)=\safeSepskImp{s}{t}{\min}(G,AX).
	\]
	
	We next identify the canonical separator used in Theorem~\ref{thm:mainThmBoundedCardinality}. Let
	$L_H^*\defeq L^*_{sA,t}(H)$.
	By eq.~\eqref{eq:transferClaim}, applied with $Z=A$, we have $\minlsep{sA,t}{H}=\minlsep{sAX,t}{G}$, and $C_{sA}(H\sminus S)=C_{sAX}(G\sminus S)$
	for every such separator $S\in \minlsep{sA,t}{H}$. Hence the source-side order is
	preserved. By uniqueness of the minimum important separator
	from Lemma~\ref{lem:uniqueMinImportantSep}, $L_H^*=L^*_{sAX,t}(G)=L^*$.	We now compare the violating components in $H\sminus L^*$ and in $G\sminus L^*$.
	Since $	|L^*|=\kappa_{sAX,t}(G)=\kappa_{sX,t}(G)$,
	and $L^*$ is an $sX,t$-separator in $G$, we have
	$L^*\in \minsep_{sX,t}(G)$. By the hypothesis of the corollary, $	X\subseteq C_s(G\sminus L^*)$.
	Consequently, every vertex of $N_G[X]{\setminus} L^*$ lies in
	$C_s(G\sminus L^*)$. Therefore the edges added in the construction of $H$
	only add edges from $s$ to vertices already lying in $C_s(G\sminus L^*)$,
	after deleting $L^*$. Hence they do not merge any component of
	$G\sminus L^*$ outside $C_s(G\sminus L^*)$ with another component, nor with the
	$s$-component. Thus the connected components of $H\sminus L^*$ not containing
	$s$ are exactly the connected components of $G\sminus L^*$ not containing
	$s$.
	
	It follows that the family of violating components defined by
	Theorem~\ref{thm:mainThmBoundedCardinality} in $H$ is precisely
	\[
	\mathcal D
	=
	\set{C\in \cc(G\sminus L^*) : C\cap A\neq\emptyset,\ s\notin C}.
	\]
	Moreover, for every $C\in\mathcal D$, $N_H(C)=N_G(C)$.
	Indeed, every edge of $\edges(H){\setminus} \edges(G)$ is incident with $s$, and its
	other endpoint lies in $N_G[X]\subseteq C_s(G\sminus L^*)\cup L^*$. Since
	$C$ is a component of $G\sminus L^*$ not containing $s$, no such added edge
	has an endpoint in $C$ and the other endpoint outside $C$ except
	possibly through a vertex of $L^*$, which is already counted in
	$N_G(C)$. Hence the neighborhood family used by Theorem~\ref{thm:mainThmBoundedCardinality} in $H$ is
	exactly $\varepsilon_{\mathcal D}=\set{N_G(C):C\in\mathcal D}$.

	We may now apply Theorem~\ref{thm:mainThmBoundedCardinality} in the graph $H$ with the set $A$.
	Since $\kappa_{s,t}(H)=\kappa_{sA,t}(H)=f_{s,t}(H,A)$,
	and $\mathcal D\neq\emptyset$, Theorem~\ref{thm:mainThmBoundedCardinality} gives
	\[
	\emptyset
	\subset
	\safeSepskImp{s}{t}{\min}(H,A)
	\subseteq
	\bigcup_{\substack{Y\in\MHS(\varepsilon_{\mathcal D})\\\kappa_{sAY,t}(H)=\kappa_{s,t}(H)}}\{L^*_{sAY,t}(H)\}.
	\]
	It remains only to translate the right-hand side back to $G$. Let $Y$ be an index in this union. By~\eqref{eq:transferClaim}, applied with $Z=A\cup Y$,
	it holds that $\minlsep{sAY,t}{H}=\minlsep{sAXY,t}{G}$, and $C_{sAY}(H\sminus S)=C_{sAXY}(G\sminus S)$
	for every such separator $S\in \minlsep{sAY,t}{H}$. Thus the source-side order is preserved,
	and by uniqueness of the minimum important separator from Lemma~\ref{lem:uniqueMinImportantSep}, $L^*_{sAY,t}(H)=L^*_{sAXY,t}(G)$.
	Also, $\kappa_{s,t}(H)=\kappa_{sX,t}(G)$.
	Using the equality $\safeSepskImp{s}{t}{\min}(H,A)=\safeSepskImp{s}{t}{\min}(G,AX)$, the inclusion obtained
	from Theorem~\ref{thm:mainThmBoundedCardinality} becomes
	\[
	\emptyset
	\subset
	\safeSepskImp{s}{t}{\min}(G,AX)
	\subseteq
	\bigcup_{\substack{Y\in\MHS(\varepsilon_{\mathcal D})\\\kappa_{sAXY,t}(G)=\kappa_{sX,t}(G)}}\set{L^*_{sAXY,t}(G)}.
	\]
	Finally, $H$ has at most $n$ vertices and $m+n$ edges.
	The implementation above uses $O(n(n+m))$ time per candidate also
	in $H$. As $\kappa_{s,t}(H)=\kappa_{sX,t}(G)$ and $T(n,m)$ includes
	input reading, its total time is
	\[
	O(2^{\kappa_{sX,t}(G)}\cdot n(n+m))
	\subseteq O(2^{\kappa_{sX,t}(G)}\cdot n\cdot T(n,m)).
	\]
	This completes the proof.
\end{proof}

\section{Formal Description and Analysis of Algorithm ~\hyperref[alg2e:GenSeps]{$\algname{Gen{-}Seps}$}}
\label{sec:AlgForImpSafeSepsAppendix}

We now give the formal description and analysis of
Algorithm~\ref{alg2e:GenSeps}. A recursive call has the form
\[
\algname{Gen{-}Seps}(G,s,t,A,X,Z,k).
\]
The set $A$ contains the vertices currently required to remain connected
to $s$; when a recursive call resets $X$ to $\emptyset$, it replaces
$A$ by $AX$.
The set $X$ contains vertices that the recursion has already forced to remain
in the $s$-component, and the set $Z$ contains vertices that have already been
committed to the separator. Thus, the current call is responsible for covering
the separators in $\safeSepskImp{s}{t}{k}(G,AX)$,
and whenever the call outputs a set $S$ in the current graph, the corresponding
candidate in the original recursion branch is $Z\cup S$.

The key invariant maintained by the algorithm is the following:
\[
\tag{Inv}
\text{for every } Q\subseteq \nodes(G)
\text{ and every } T\in \minsep_{sXQ,t}(G),
\quad
X\subseteq C_s(G\sminus T).
\]
In words, once a vertex is placed in $X$, every minimum separator that is later
computed while treating $X$ as part of the source side is guaranteed to include all vertices of
$X$ in the $s$-component (i.e., every vertex of $X$ remains connected to $s$). This is the property that allows the algorithm to replace statements about $sX,t$- and $sAX,t$-separators by statements about CP
$s,t$-separators with respect to $AX$. The invariant is not assumed: it is
proved in Lemma~\ref{lem:GenSepsInvariant}. The proof uses the fact that the
algorithm enlarges $X$ only in Line~\ref{line2e:increaseX} and in
Line~\ref{line:recursiveCommitToSep}, and in both cases the added vertices
belong to the union of minimum separators for the current source side, so
Corollary~\ref{corr:XInvariant} applies.

To guide the analysis, we associate each call to
Algorithm~\ref{alg2e:GenSeps} with the potential
\begin{equation}
	\label{eq:LambdaPotential}
	\lambda(G,A,X,k)
	\eqdef
	(k+1)(2k-\kappa_{sAX,t}(G))+(k-\kappa_{sX,t}(G)).
\end{equation}
This potential is used only for calls that can still generate a separator of
size at most $k$. In such calls, the tests in
Lines~\ref{line2e:L_sAtt>k} and~\ref{line2e:L_stt>k} have failed, and hence
\[
\kappa_{sAX,t}(G)\le k
\qquad\text{and}\qquad
\kappa_{sX,t}(G)\le k.
\]
Therefore $\lambda(G,A,X,k)$ is a nonnegative integer bounded by
$2k^2+3k$. Lemma~\ref{claim:recursion-height} proves that every recursive
step either strictly decreases this potential, or is the special recursive
step in Line~\ref{line:recursiveCommitToSep}, which cannot occur twice
consecutively. This gives the $O(k^2)$ bound on the recursion height.

The algorithm first checks whether the current instance can contain any
separator of size at most $k$. In Line~\ref{line2e:L_sAtt>k}, if
$\kappa_{sAX,t}(G)>k$, then every CP $s,t$-separator with respect to $AX$ has
size larger than $k$, because every such separator is an $sAX,t$-separator.
Hence the call returns.

The algorithm then computes, in Line~\ref{line2e:L_stt}, the unique minimum
$sX,t$-separator closest to $t$, denoted $L^t_{sX,t}(G)$. By
Lemma~\ref{lem:uniqueMinImportantSep}, this separator is well-defined and can
be computed in time \rev{$O(n\cdot T(n,m))$}. If
$|L^t_{sX,t}(G)|=\kappa_{sX,t}(G)>k$, then every separator in
$\safeSepskImp{s}{t}{k}(G,AX)$ has size larger than $k$, since every CP
$s,t$-separator with respect to $AX$ is also an $sX,t$-separator. Therefore
the call returns in Line~\ref{line2e:L_stt>k}.

Next suppose that the algorithm reaches Line~\ref{line2e:returnZEmptyIf}, so
that $L^t_{sX,t}(G)=\emptyset$. Then the empty set is a minimum
$sX,t$-separator. Applying (Inv) with $Q=\emptyset$ and
$T=\emptyset$, we get $X\subseteq C_s(G)$. If $A\subseteq C_s(G)$, then
$AX\subseteq C_s(G)$, so the empty set is a CP $s,t$-separator with respect to
$AX$ in the current graph. The algorithm therefore outputs $Z$ in
Line~\ref{line2e:returnZEmpty}. If $A\not\subseteq C_s(G)$, then no separator
in the current graph can make all vertices of $AX$ lie in the $s$-component,
because deleting vertices cannot connect a vertex of $A$ to $s$. In this case
the call returns in Line~\ref{line2e:returnZEmpty2}. From
Line~\ref{line2e:L_{s,t}^tNormal} onward, the algorithm therefore proceeds
under the assumption that $0<|L^t_{sX,t}(G)|\le k$.

The first branching step is the block
Lines~\ref{line2e:prefirstIf}--\ref{line2e:prefirstIfEnd}. This block handles
the case where
\[
L^t_{sX,t}(G)\notin \safeSeps{s}{t}(G,A).
\]
Let $L^t_{sX,t}(G)=\set{x_1,\ldots,x_\ell}$.
By Lemma~\ref{lem:partZeroOfAlg}, applied with the terminal set $AX$, every
separator in $\safeSepskImp{s}{t}{k}(G,AX)$ is contained in one of the families
\[
\safeSepskImp{s}{t}{k}(G,AXx_i),
\qquad i\in\{1,\ldots,\ell\}.
\]
Therefore the recursive calls in Line~\ref{line2e:increaseX} cover all
separators represented by the current call. Moreover, for every
$x_i\in L^t_{sX,t}(G)$, Lemma~\ref{lem:potentialReducedPreFirstPart} gives $\kappa_{sXx_i,t}(G)>\kappa_{sX,t}(G)$.
Thus the call in Line~\ref{line2e:increaseX} strictly decreases the potential
\eqref{eq:LambdaPotential}. Lemma~\ref{lem:GenSepsInvariant} proves that
(Inv) is preserved for these recursive calls.

We may now assume that the test in Line~\ref{line2e:prefirstIf} fails; that is, $L^t_{sX,t}(G)\in \safeSeps{s}{t}(G,A)$.
Apply (Inv) with $Q=\emptyset$ and
$T=L^t_{sX,t}(G)\in\minsep_{sX,t}(G)$. This gives $X\subseteq C_s(G\sminus L^t_{sX,t}(G))$.
Since the current branch also satisfies $A\subseteq C_s(G\sminus L^t_{sX,t}(G))$,
we get $AX\subseteq C_s(G\sminus L^t_{sX,t}(G))$.
Hence $L^t_{sX,t}(G)$ is a CP $s,t$-separator with respect to $AX$, and
\[
f_{s,t}(G,AX)\le |L^t_{sX,t}(G)|=\kappa_{sX,t}(G).
\]
On the other hand, every CP $s,t$-separator with respect to $AX$ is an
$sAX,t$-separator, and every $sAX,t$-separator is an $sX,t$-separator. Therefore
\[
\kappa_{sX,t}(G)
\le
\kappa_{sAX,t}(G)
\le
f_{s,t}(G,AX).
\]
Consequently,
\[
\kappa_{sX,t}(G)=\kappa_{sAX,t}(G)=f_{s,t}(G,AX).
\]
This equality is the premise needed in the final branch, where
Corollary~\ref{corr:mainThmBoundedCardinalityExt} is applied.

The next step is the branch beginning at Line~\ref{line2e:elseIfBegin}. In
this branch the algorithm tests whether the unique important minimum
$sAX,t$-separator $L^*\defeq \closestMinSep{sAX,t}(G)$
is CP with respect to $A$. If $L^*\in \safeSeps{s}{t}(G,A)$,
then we apply (Inv) with $Q=A$ and
$T=L^*\in\minsep_{sAX,t}(G)$. This gives $X\subseteq C_s(G\sminus L^*)$.
Since the branch assumes $A\subseteq C_s(G\sminus L^*)$, we get $AX\subseteq C_s(G\sminus L^*)$.
Thus $L^*$ is CP with respect to $AX$. Corollary~\ref{corr:partTwoOfAlg} gives
\[
\safeSepsImp{s}{t}(G,AX)
\subseteq
\{L^*\}
\cup
\bigcup_{x\in L^*}
\safeSepsImp{s}{t}(G_x,AX),
\]
where $G_x$ is obtained from $G$ by adding all edges between $t$ and $N_G[x]$.
The separator $L^*$ itself is covered by the recursive call in
Line~\ref{line2e:secondPart1}, where $L^*$ is deleted, added to $Z$, and the
budget is reduced by $|L^*|$. The remaining families are covered by the loop in
Lines~\ref{line2e:secondPart1LoopBegin}--\ref{line2e:secondPart1End}. For each
$x_i\in L^*$, the graph $G_i$ is obtained by adding all edges between $t$ and
$N_G[x_i]$, and the recursive call in Line~\ref{line2e:secondPart1End} covers
$\safeSepskImp{s}{t}{k}(G_i,AX)$. Lemma~\ref{lem:sAtMinlSepLargerThanMinsAtSep_s}
shows that the relevant connectivity strictly increases (i.e., \rev{$\kappa_{sAX,t}(G)<\kappa_{sAX,t}(G_i)$} for every $x_i\in L^*$) in these edge-insertion
calls, and Lemma~\ref{claim:recursion-height} uses this to prove that the
potential strictly decreases.

It remains to describe the final branch, beginning at
Line~\ref{line2e:thirdPartStart}. In this branch,
\[
L^t_{sX,t}(G)\in \safeSeps{s}{t}(G,A)
\qquad\text{and}\qquad
L^*\notin \safeSeps{s}{t}(G,A),
\]
where $L^*=\closestMinSep{sAX,t}(G)$. As shown above, at this point
\[
\kappa_{sX,t}(G)=\kappa_{sAX,t}(G)=f_{s,t}(G,AX).
\]
Moreover, (Inv) gives precisely the hypothesis required by
Corollary~\ref{corr:mainThmBoundedCardinalityExt}: for every
$Q\subseteq\nodes(G)$ and every $T\in\minsep_{sXQ,t}(G)$, the set $X$ lies in
$C_s(G\sminus T)$. The algorithm therefore defines, in
Line~\ref{line2e:thirdPartStart},
\begin{align*}
\mathcal D
\defeq
\set{C\in\cc(G\sminus L^*) : s\notin C,\ C\cap A\neq\emptyset},&&\text{ and }&&\varUpsilon \eqdef \set{N_G(C):C\in\mathcal D}.
\end{align*}
Using the enumeration procedure of Corollary~\ref{corr:mainThmBoundedCardinalityExt}, the algorithm finds one $Y\in\mathrm{MHS}(\Upsilon)$ such that $\closestMinSep{sAXY,t}(G)\in \safeSepskImp{s}{t}{\min}(G,AX)$.
Since every set in $\varUpsilon$ is a
subset of $L^*$, every vertex of the minimal hitting set $Y$ belongs to $L^*$;
in particular, $|Y|\le |L^*|\le k$.

The recursive calls in Lines~\ref{line:recursiveCommitToSep},
\ref{line:recursiveGy}, and~\ref{line:recursiveGy2} are justified by
Lemma~\ref{lem:thirdPart}. Writing $Y=\set{y_1,\ldots,y_q}$, Lemma~\ref{lem:thirdPart} gives
\[
\safeSepsImp{s}{t}(G,AX)
\subseteq
\safeSepsImp{s}{t}(G,AXY)
\cup
\bigcup_{i=1}^q
\{T\cup\{y_i\}:T\in\safeSepsImp{s}{t}(G\sminus y_i,AX)\}
\cup
\bigcup_{i=1}^q
\safeSepsImp{s}{t}(G_i,AX),
\]
where $G_i$ is obtained from $G$ by adding all edges between $t$ and
$N_G[y_i]{\setminus}\set{t}$. The first family is covered by the recursive call in
Line~\ref{line:recursiveCommitToSep}, which adds $Y$ to $X$. The second family
is covered by the recursive calls in Line~\ref{line:recursiveGy}, where $y_i$
is committed to the separator, deleted from the graph, and the budget is reduced
by one. The third family is covered by the recursive calls in
Line~\ref{line:recursiveGy2}, where the graph is modified by adding the edges
from $t$ to $N_G[y_i]{\setminus}\set{t}$. Lemma~\ref{lem:GenSepsInvariant} proves
that (Inv) is preserved in all these recursive calls.

The correctness of the recursive calls is proved in
Theorem~\ref{thm:GenSepsCompleteness}. The invariant (Inv) is proved in
Lemma~\ref{lem:GenSepsInvariant}. The recursion height is bounded in
Lemma~\ref{claim:recursion-height}. Finally,
Theorem~\ref{thm:detailedRuntimeAnalysis} combines completeness, the recursion
height, the $O(k)$ branching bound, and the cost of applying
Corollary~\ref{corr:mainThmBoundedCardinalityExt}, to obtain the output-size
and running-time bounds.

\def\lemPartZeroOfAlg{
	Let $A\subseteq \nodes(G)$ and let $S\in \minlsepst{G}$. If $A\not\subseteq C_s(G\sminus S)$, then:
	\[
			\safeSepsImp{s}{t}(G,A)\subseteq \bigcup_{x\in S}\safeSepsImp{s}{t}(G,Ax).
	\]
}
\begin{lemma}
	\label{lem:partZeroOfAlg}
	\lemPartZeroOfAlg
\end{lemma}
\begin{proof}
	Let $T\in \safeSepsImp{s}{t}(G,A)$. We show that there exists
	$x\in S$ such that $T\in \safeSepsImp{s}{t}(G,Ax)$.
	Since $T\in \safeSepsImp{s}{t}(G,A)$, in particular
	$T\in \safeSeps{s}{t}(G,A)$. Hence $A\subseteq C_s(G\sminus T)$.
	By assumption, $A\not\subseteq C_s(G\sminus S)$, so choose $a\in A{\setminus} C_s(G\sminus S)$.
	Since $a\in A\subseteq C_s(G\sminus T)$, there is an $s,a$-path $P$
	in $G\sminus T$.
	The path $P$ starts at $s\in C_s(G\sminus S)$ and ends at
	$a\notin C_s(G\sminus S)$. Therefore $P$ must contain a vertex of
	$S$. Let $x\in \nodes(P)\cap S$.
	Since $P$ is contained in $G\sminus T$, we have $x\notin T$. Moreover,
	the subpath of $P$ from $s$ to $x$ is contained in $G\sminus T$, and
	therefore $x\in C_s(G\sminus T)$.
	Together with $A\subseteq C_s(G\sminus T)$, this gives $Ax\subseteq C_s(G\sminus T)$.
	Thus $	T\in \safeSeps{s}{t}(G,Ax)$.

	It remains to prove that $T$ is important in $\safeSeps{s}{t}(G,Ax)$; that is $T\in \safeSepsImp{s}{t}(G,Ax)$.
	Suppose not. Then there exists $T'\in \safeSeps{s}{t}(G,Ax)$
	such that
	\[
	C_s(G\sminus T')\subsetneq C_s(G\sminus T)
	\qquad\text{and}\qquad
	|T'|\le |T|.
	\]
	Since $T'\in \safeSeps{s}{t}(G,Ax)$, we have $Ax\subseteq C_s(G\sminus T')$.
	In particular, $A\subseteq C_s(G\sminus T')$,
	and so $T'\in \safeSeps{s}{t}(G,A)$.
	But then $T'$ contradicts the assumption that
	$T\in \safeSepsImp{s}{t}(G,A)$. Therefore $	T\in \safeSepsImp{s}{t}(G,Ax)$.

	Since $T\in \safeSepsImp{s}{t}(G,A)$ was arbitrary, we conclude that
	\[
	\safeSepsImp{s}{t}(G,A)
	\subseteq
	\bigcup_{x\in S}\safeSepsImp{s}{t}(G,Ax).
	\]
\end{proof}

\def\lemPartZeroOfAlgBetter{
	Let $A\subseteq \nodes(G)$ and let $L^t\in \minsep_{s,t}{G}$ the unique minimum $s,t$-separator closest to $t$. If $A\not\subseteq C_s(G\sminus L^t)$, then:
	\[
	\safeSepsImp{s}{t}(G,A)\subseteq \bigcup_{x\in L^t}\safeSepsImp{s}{t}(G_x,A),
	\]
	where $G_x$ is the graph that results from $G$ by adding all edges between $s$ and $N_G[x]$.
}

For the proofs of Lemmas~\ref{lem:potentialReducedPreFirstPart} and~\ref{lem:sAtMinlSepLargerThanMinsAtSep_s}, we require Lemma~\ref{lem:sAtMinlSepLargerThanMinsAtSep_t} below.
\begin{lemma}
	\label{lem:sAtMinlSepLargerThanMinsAtSep_t}
	Let $A\subseteq \nodes(G){\setminus}\set{s,t}$ where $sA\cap N_G[t]=\emptyset$, and let $L^t \in \minsep_{sA,t}(G)$ be the unique minimum $sA,t$-separator that is closest to $t$. Then, for any $T\in \minlsep{sA,t}{G}$:
	\begin{align*}
		\text{If }&& C_t(G\sminus L^t)\not\subseteq C_t(G\sminus T) && \text{ then } && |T|>\kappa_{sA,t}(G).
	\end{align*} 
\end{lemma}
\begin{proof}
	Suppose, by way of contradiction, that $C_t(G\sminus L^t)\not\subseteq C_t(G\sminus T)$, and that $|T|\leq \kappa_{sA,t}(G)$. By Lemma~\ref{lem:simpAB}, both separators are the boundaries of their target components. Applying submodularity (Lemma~\ref{lem:submodularity}), we obtain:
	\begin{align}
		\underbrace{|T|}_{=\kappa_{sA,t}(G)}+\underbrace{|L^t|}_{=\kappa_{sA,t}(G)} & \geq \underbrace{|N_G(C_t(G\sminus T)\cup C_t(G\sminus L^t))|}_{\geq \kappa_{sA,t}(G)}+  |N_G(C_t(G\sminus T)\cap C_t(G\sminus L^t))| && \label{eq:sAtMinlSepLargerThanMinsAtSep_t_1}
	\end{align}
	Since $|T|\le \kappa_{sA,t}(G)$ and $T$ is an $sA,t$-separator, we have $T\in\minsep_{sA,t}(G)$.
	Since $L^t, T\in \minsep_{sA,t}(G)$, then clearly both $C_t(G\sminus T)\cap C_t(G\sminus L^t)$ and $C_t(G\sminus T)\cup C_t(G\sminus L^t)$ are vertex-sets that contain $t$ and no vertices from $sA$. Hence, $N_G(C_t(G\sminus T)\cap C_t(G\sminus L^t))$ and $N_G(C_t(G\sminus T)\cup C_t(G\sminus L^t))$ are $sA,t$-separators of $G$. In particular,  $|N_G(C_t(G\sminus T)\cup C_t(G\sminus L^t))| \geq \kappa_{sA,t}(G)$.
	From~\eqref{eq:sAtMinlSepLargerThanMinsAtSep_t_1}, we get that  $|N_G(C_t(G\sminus T)\cap C_t(G\sminus L^t))|=\kappa_{sA,t}(G)$. Since $C_t(G\sminus L^t)\not\subseteq C_t(G\sminus T)$, then $C_t(G\sminus T)\cap C_t(G\sminus L^t) \subset  C_t(G\sminus L^t)$, then $L^t$ is not the closest to $t$; a contradiction.
\end{proof}

\def\potentialReducedPreFirstPart{
	Let $X\subseteq \nodes(G)$, and let $L_{sX,t}^t \in \minsep_{sX,t}(G)$ be the unique minimum $sX,t$-separator that is closest to $t$, and let $y\in L_{sX,t}^t(G)$. Then, $\kappa_{sXy,t}(G)> \kappa_{sX,t}(G)$.
}
\begin{lemma}
	\label{lem:potentialReducedPreFirstPart}
	\potentialReducedPreFirstPart
\end{lemma}
	\begin{proof}
		If $y\in N_G(t)$, then $\kappa_{sXy,t}(G)=\infty$, and the claim trivially holds. So, assume that $sXy \cap N_G[t]=\emptyset$.
		\rev{Suppose, for a contradiction, that $\kappa_{sXy,t}(G)\le\kappa_{sX,t}(G)$. Monotonicity gives equality. Choose $T\in\minsep_{sXy,t}(G)$; it is therefore also a minimum, and hence minimal, $sX,t$-separator.} By definition, $y\notin T\cup C_t(G\sminus T)$, and hence $N_G[y]\cap C_t(G\sminus T)=\emptyset$. Since $y\in L_{sX,t}^t(G) \in \minsep_{sX,t}(G)$, then, by Lemma~\ref{lem:simpAB},  $N_G[y]\cap C_t(G\sminus L_{sX,t}^t(G))\neq \emptyset$. Consequently, $C_t(G\sminus T) \not\supseteq C_t(G\sminus L_{sX,t}^t(G))$. By Lemma~\ref{lem:sAtMinlSepLargerThanMinsAtSep_t}, it holds that $\kappa_{sXy,t}(G)> \kappa_{sX,t}(G)$, \rev{a contradiction}.  
\end{proof}

\def\sAtMinlSepLargerThanMinsAtSep_s{
	Let $A\subseteq \nodes(G){\setminus}\set{s,t}$.
	\begin{enumerate}
		\item If $x\in \closestMinSep{sA,t}(G)$ then $\kappa_{sA,tx}(G)> \kappa_{sA,t}(G)$.
		\item Let $L_{sA,t}^t(G)\in \minsep_{sA,t}(G)$ that is closest to $t$. If $x\in L_{sA,t}^t(G)$ then $\kappa_{sAx,t}(G)> \kappa_{sA,t}(G)$.
	\end{enumerate}
}
\begin{lemma}
	\label{lem:sAtMinlSepLargerThanMinsAtSep_s}
	\sAtMinlSepLargerThanMinsAtSep_s
\end{lemma}
\begin{proof}
	By Lemma~\ref{lem:MinlsASep}, it holds that $\minlsep{sA,t}{G}=\minlsepst{H}$ where $H$ is the graph that results from $G$ by adding all edges between $s$ and $N_G[A]$.  In particular, $\kappa_{sA,t}(G)=\kappa_{s,t}(H)$, $\closestMinSep{sA,t}(G)=\closestMinSep{s,t}(H)$, and $L_{sA,t}^t(G)=L_{s,t}^t(H)$. 
	
\rev{For item~1, if $\kappa_{sA,tx}(G)=\infty$, the claim holds.
Otherwise suppose, for a contradiction, that
$\kappa_{sA,tx}(G)\le\kappa_{sA,t}(G)$. By monotonicity these values
are equal. Hence any $T\in\minsep_{sA,tx}(G)$ is also a minimum
$sA,t$-separator in $G$, and a minimum $s,t$-separator in $H$.}
	Let $T\in \minsep_{sA,tx}(G)$, and hence $T\in \minsep_{s,tx}(H)$. By definition, $x\notin T\cup C_s(H\sminus T)$, and hence $N_H[x]\cap C_s(H\sminus T)=\emptyset$. Since $x\in \closestMinSep{s,t}(H)$, then by Lemma~\ref{lem:fullComponents}, $N_H[x]\cap C_s(H\sminus \closestMinSep{s,t}(H))\neq \emptyset$. Consequently, $C_s(H\sminus T) \not\supseteq C_s(H\sminus \closestMinSep{s,t}(H))$. We apply Lemma~\ref{lem:sAtMinlSepLargerThanMinsAtSep_t} with the roles of $s$ and $t$ interchanged; equivalently, we use its symmetric form for the $s$-component. By Lemma~\ref{lem:sAtMinlSepLargerThanMinsAtSep_t}, $\kappa_{sA,tx}(G)=\kappa_{s,tx}(H)>\kappa_{s,t}(H)=\kappa_{sA,t}(G)$.
	
\rev{The strict inequality contradicts our assumption, proving item~1.
Item~2 follows from Lemma~\ref{lem:potentialReducedPreFirstPart} with $X=A$.}
\end{proof}

\def\partTwoOfAlg{
	Let $A\subseteq \nodes(G){\setminus}\set{s,t}$, and let
\[
L\defeq \closestMinSep{sA,t}(G)=\set{x_1,\ldots,x_\ell}
\]
be the unique minimum $sA,t$-separator closest to $t$. Suppose that
$L\in \safeSeps{s}{t}(G,A)$. For every $i\in\set{1,\ldots,\ell}$, let
$G_i$ be the graph obtained from $G$ by adding all edges between $t$
and $N_G[x_i]$, and all edges between $s$ and
$\set{x_1,\ldots,x_{i-1}}$. Then
\[
\safeSepsImp{s}{t}(G,A)
\subseteq
\set{L}
\cup
\bigcup_{i=1}^{\ell}\safeSepsImp{s}{t}(G_i,A).
\]
}
\def\partTwoOfAlgReduced{
	Let $A\subseteq V(G)\setminus\set{s,t}$, and let
	\[
	L\defeq L^*_{sA,t}(G)=\set{x_1,\ldots,x_\ell}
	\]
	be the unique important minimum $sA,t$-separator.
	Suppose that
	$L\in \safeSeps{s}{t}(G,A)$. For every $i\in\set{1,\ldots,\ell}$, let
	$G_i$ be the graph obtained from $G$ by adding all edges between $t$
	and $N_G[x_i]$. Then
	\[
	\safeSepsImp{s}{t}(G,A)
	\subseteq
	\set{L}
	\cup
	\bigcup_{i=1}^{\ell}\safeSepsImp{s}{t}(G_i,A).
	\]
}
\begin{lemma}
	\label{lem:partTwoOfAlg}
	\partTwoOfAlgReduced
\end{lemma}
\begin{proof}
	Let $S\in \safeSepsImp{s}{t}(G,A)$. If $S=L$, then there is nothing
	to prove. Assume therefore that $S\neq L$. We prove that
	$S\in \safeSepsImp{s}{t}(G_i,A)$ for some $i\in\set{1,\ldots,\ell}$.
	
	Set $C\defeq C_s(G\sminus S)$.
	Since $S\in \safeSepsImp{s}{t}(G,A)$, we have $S\in \safeSeps{s}{t}(G,A)$,
	and hence $	S\in \minlsepst{G}$ and $A\subseteq C$.
	By assumption, $L\in \safeSeps{s}{t}(G,A)$,
	so in particular $	L\in \minlsepst{G}$ and $A\subseteq C_s(G\sminus L)$.

	We first show that there exists $x_i\in L$ such that
	\[
	x_i\notin S\cup C_s(G\sminus S). \tag{1}
	\]
	Suppose not. Then $	L\subseteq S\cup C_s(G\sminus S)$.
	Since $L,S\in \minlsepst{G}$, Lemma~\ref{lem:inclusionCsCt} gives
	\[
	C_s(G\sminus L)\subseteq C_s(G\sminus S). \tag{2}
	\]
	Moreover, $L$ is a minimum $sA,t$-separator, while $S$ is an
	$sA,t$-separator because $A\subseteq C_s(G\sminus S)$. Therefore
	\[
	|L|\le |S|. \tag{3}
	\]
	If the inclusion in~(2) is strict, then $L\in \safeSeps{s}{t}(G,A)$, $	C_s(G\sminus L)\subsetneq C_s(G\sminus S)$, and $|L|\le |S|$, contradicting the assumption that
	$S\in \safeSepsImp{s}{t}(G,A)$. Hence equality holds in~(2): $	C_s(G\sminus L)=C_s(G\sminus S)$.
	Since $L,S\in \minlsepst{G}$, Lemma~\ref{lem:fullComponents} implies $L=N_G(C_s(G\sminus L))
	=N_G(C_s(G\sminus S))=S$
	contradicting $S\neq L$. Therefore some $x_i\in L$ satisfies~(1); that is there exists an $x_i\in L$ such that $x_i\notin S\cup C_s(G\sminus S)$.
	Fix such an index $i$ where $x_i\in L$. We claim that
	\[
	N_G[x_i]\cap C_s(G\sminus S)=\emptyset. \tag{4}
	\]
	Indeed, $x_i\notin C_s(G\sminus S)$ by~(1). Moreover, if some neighbor
	$u\in N_G(x_i)$ belonged to $C_s(G\sminus S)$, then, since
	$x_i\notin S$ by~(1), the edge $ux_i$ would imply
	$x_i\in C_s(G\sminus S)$, a contradiction. Thus~(4) holds.
	
	Let $D_i\eqdef N_G[x_i]{\setminus}\set{t}$.
	Since $s\in C_s(G\sminus S)$, equation~(4) also implies
	$s\notin N_G[x_i]$. Hence $D_i\subseteq \nodes(G){\setminus}\set{s,t}$.
	Furthermore, adding all edges between $t$ and $N_G[x_i]$ is the same
	as adding all edges between $t$ and $D_i$, since the possible edge from
	$t$ to itself is irrelevant. Thus $G_i$ is precisely the graph obtained
	from $G$ by adding all edges between $t$ and $D_i$.
	
	By~(4), $	D_i\cap C_s(G\sminus S)=\emptyset$.
	Since $S\in \safeSepsImp{s}{t}(G,A)$, Lemma~\ref{corr:MinlsBtSepSafe}
	applied with $D=D_i$ gives $S\in \safeSepsImp{s}{t}(G_i,A)$.
	Thus every $S\in \safeSepsImp{s}{t}(G,A)$ either equals $L$, or
	belongs to $\safeSepsImp{s}{t}(G_i,A)$ for some
	$i\in\set{1,\ldots,\ell}$. Therefore
	\[
	\safeSepsImp{s}{t}(G,A)
	\subseteq
	\set{L}
	\cup
	\bigcup_{i=1}^{\ell}\safeSepsImp{s}{t}(G_i,A),
	\]
	as required.
\end{proof}

\begin{corollary}
	\label{corr:partTwoOfAlg}
	\rev{Let $A,X\subseteq \nodes(G){\setminus}\set{s,t}$; these sets may overlap.} Assume that for every $T\in \minsep_{sAX,t}(G)$, it holds that $X\subseteq C_s(G\sminus T)$. If $A\subseteq C_s(G\sminus \closestMinSep{sAX,t}(G))$, then:
	\[
		\safeSepsImp{s}{t}(G,AX)\subseteq \set{\closestMinSep{sAX,t}(G)}\cup \mediumbigcup_{x\in \closestMinSep{sAX,t}(G)}\safeSepsImp{s}{t}(G_x,AX)
	\]
	where $G_x$ is the graph that results from $G$ by adding all edges between $t$ and $N_G[x]$.
\end{corollary}
\begin{proof}
	Since, by definition $\closestMinSep{sAX,t}(G)\in \minsep_{sAX,t}(G)$, then by the assumption of the lemma $X\subseteq C_s(G\sminus \closestMinSep{sAX,t}(G))$. Combined with the assumption that  $A\subseteq C_s(G\sminus \closestMinSep{sAX,t}(G))$, then $AX\subseteq C_s(G\sminus \closestMinSep{sAX,t}(G))$. Therefore, $\closestMinSep{sAX,t}(G)\in \safeSeps{s}{t}(G,AX)$. By Lemma~\ref{lem:partTwoOfAlg}, we get that
	\[
	\safeSepsImp{s}{t}(G,AX)\subseteq \set{\closestMinSep{sAX,t}(G)}\cup \mediumbigcup_{x\in \closestMinSep{sAX,t}(G)}\safeSepsImp{s}{t}(G_x,AX)
	\]
	as required.
\end{proof}

\def\thirdpart{
Let $A,X\subseteq \nodes(G){\setminus}\set{s,t}$, and let $L^*\eqdef \closestMinSep{sXA,t}(G)$.
Let $Y=\set{y_1,\dots,y_\ell} \subseteq L^*$. Then:
\begin{align*}
	\safeSepsImp{s}{t}(G,AX)\subseteq \safeSepsImp{s}{t}(G,AXY)\cup \bigcup_{i=1}^\ell\set{T\cup \set{y_i}: T\in \safeSepsImp{s}{t}(G\sminus y_i, AX)} \cup \bigcup_{i=1}^\ell\safeSepsImp{s}{t}(G_i,AX)
\end{align*}
where $G_i$ is the graph that results from $G$ by adding the edges $\set{(u,t): u\in N_G[y_i]{\setminus}\set{t}}$.
}
\begin{lemma}
	\label{lem:thirdPart}
	\thirdpart
\end{lemma}
\begin{proof}
	Let $S\in \safeSepsImp{s}{t}(G,AX)$. We show that $S$ belongs to one
	of the families on the right-hand side.
	If $Y\subseteq C_s(G\sminus S)$,
	then, since $AX\subseteq C_s(G\sminus S)$, we have $AXY\subseteq C_s(G\sminus S)$.
	Thus $S\in \safeSeps{s}{t}(G,AXY)$. Moreover, $S$ is important in
	$\safeSeps{s}{t}(G,AXY)$: indeed, if some
	$S'\in \safeSeps{s}{t}(G,AXY)$ satisfied
	\[
	C_s(G\sminus S')\subsetneq C_s(G\sminus S)
	\qquad\text{and}\qquad
	|S'|\le |S|,
	\]
	then $S'\in \safeSeps{s}{t}(G,AX)$ as well, contradicting
	$S\in \safeSepsImp{s}{t}(G,AX)$. Hence $	S\in \safeSepsImp{s}{t}(G,AXY)$,
	and we are done.
	
	Assume from now on that $Y\not\subseteq C_s(G\sminus S)$.
	Let $i$ be the smallest index such that $y_i\notin C_s(G\sminus S)$.
	
	There are two cases. First suppose that $y_i\in S$.	Set $T\eqdef S\sminus\set{y_i}$.
	We prove that $T\in \safeSepsImp{s}{t}(G\sminus y_i,AX)$.
	Since $S\in \safeSeps{s}{t}(G,AX)$ and $y_i\in S$, the set $T$
	separates $s$ from $t$ in $G\sminus y_i$, and $(G\sminus y_i)\sminus T = G\sminus S$.
	Therefore
	\[
	C_s((G\sminus y_i)\sminus T)=C_s(G\sminus S). \tag{1}
	\]
	In particular, $AX\subseteq C_s((G\sminus y_i)\sminus T)$.
	
	We next verify minimality. Let $w\in T$. Since
	$S\in \minlsepst{G}$, Lemma~\ref{lem:simpAB} gives that $w$ has a neighbor in
	$C_s(G\sminus S)$ and a neighbor in $C_t(G\sminus S)$. Since
	$G\sminus S=(G\sminus y_i)\sminus T$, the same two witness components
	show, again by Lemma~\ref{lem:simpAB}, that $w$ is necessary for separating $s$
	from $t$ in $G\sminus y_i$. Hence $T\in \minlsepst{G\sminus y_i}$.
	Together with (1), this gives $T\in \safeSeps{s}{t}(G\sminus y_i,AX)$.

	It remains to prove importance. Suppose, toward a contradiction, that
\[T\notin \safeSepsImp{s}{t}(G\sminus y_i,AX).\]
	Then there exists $	T'\in \safeSeps{s}{t}(G\sminus y_i,AX)$
	such that
	\[
	C_s((G\sminus y_i)\sminus T')
	\subsetneq
	C_s((G\sminus y_i)\sminus T)
	\qquad\text{and}\qquad
	|T'|\le |T|. \tag{2}
	\]
	Let $	C'\defeq C_s((G\sminus y_i)\sminus T')$.
	By (1) and (2),
	\[
	C'\subsetneq C_s(G\sminus S). \tag{3}
	\]
	In particular, since $y_i\in S$, we have $y_i\notin C_s(G\sminus S)$,
	and hence
	\[
	y_i\notin C'. \tag{4}
	\]
	Since $T'\in \safeSeps{s}{t}(G\sminus y_i,AX)$, we have
	\[
	AX\subseteq C'. \tag{5}
	\]
	Also, $T'$ is an $s,t$-separator in $G\sminus y_i$, so
	$T'\cup\set{y_i}$ is an $s,t$-separator in $G$. Define $W\defeq N_G(C')$.
	Because $C'$ is the $s$-component of $G\sminus (T'\cup\set{y_i})$; that is, $C'\eqdef C_s(G\sminus (\{y_i\}\cup T'))$,	every neighbor of $C'$ in $G$ belongs to $T'\cup\set{y_i}$. Hence $W\subseteq T'\cup\set{y_i}$, and therefore
	\[
	|W|\le |T'|+1\le |T|+1=|S|. \tag{6}
	\]
	Moreover, by definition of $W$ and because $t\notin C'$, $W$ is an $s,t$-separator in $G$, and
	\[
	C_s(G\sminus W)=C'. \tag{7}
	\]
	We claim that $W\in \minlsepst{G}$. Every vertex of $W=N_G(C')$ has a
	neighbor in $C'=C_s(G\sminus W)$. It remains to show that every vertex
	of $W$ has a neighbor in the $t$-component of $G\sminus W$.
	Let $w\in W$. If $w\in T'$, then since
	$T'\in \minlsepst{G\sminus y_i}$, Lemma~\ref{lem:simpAB} applied in $G\sminus y_i$
	gives that $w$ has a neighbor in the $t$-component of
	$(G\sminus y_i)\sminus T'$. This component is contained in
	$C_t(G\sminus W)$, because $W\subseteq T'\cup\set{y_i}$. Hence $w$
	has a neighbor in $C_t(G\sminus W)$.
	
	It remains to consider the case $w=y_i$. In this case $y_i\in W$, so
	$y_i$ has a neighbor in $C'$. Since $y_i\in S$ and
	$S\in \minlsepst{G}$, Lemma~\ref{lem:simpAB} gives that $y_i$ has a neighbor in
	$C_t(G\sminus S)$. We show that this neighbor lies in
	$C_t(G\sminus W)$. Indeed, by (3), $C'\subseteq C_s(G\sminus S)$, and
	therefore every vertex of $W=N_G(C')$ lies in $C_s(G\sminus S)\cup S$.
	Thus no vertex of $W{\setminus}\set{y_i}$ lies in $C_t(G\sminus S)$.
	Consequently the $t$-component $C_t(G\sminus S)$ is contained in
	$C_t(G\sminus W)$, and the neighbor of $y_i$ in $C_t(G\sminus S)$ is
	also a neighbor of $y_i$ in $C_t(G\sminus W)$.
	Therefore every vertex of $W$ has a neighbor in both
	$C_s(G\sminus W)$ and $C_t(G\sminus W)$. By Lemma~\ref{lem:simpAB}, $W\in \minlsepst{G}$.
	By (5) and (7), $AX\subseteq C_s(G\sminus W)$,
	so $W\in \safeSeps{s}{t}(G,AX)$.
	Finally, by (3), (6), and (7),
	\[
	C_s(G\sminus W)\subsetneq C_s(G\sminus S)
	\qquad\text{and}\qquad
	|W|\le |S|.
	\]
	This contradicts $S\in \safeSepsImp{s}{t}(G,AX)$. Hence $	T\in \safeSepsImp{s}{t}(G\sminus y_i,AX)$,	
	and therefore $S=T\cup\set{y_i}$
	belongs to the $i$-th deletion branch.
	
	It remains to consider the second case: $y_i\notin S$.
	Since $y_i\notin C_s(G\sminus S)$ by the choice of $i$, we have
	\[
	y_i\notin S\cup C_s(G\sminus S). \tag{8}
	\]
	We claim that
	\[
	N_G[y_i]\cap C_s(G\sminus S)=\emptyset. \tag{9}
	\]
	Indeed, $y_i\notin C_s(G\sminus S)$ by (8). If some neighbor
	$u\in N_G(y_i)$ belonged to $C_s(G\sminus S)$, then, since
	$y_i\notin S$, the edge $(u,y_i)$ would imply
	$y_i\in C_s(G\sminus S)$, a contradiction.
	
	Let $D_i\eqdef N_G[y_i]{\setminus}\set{t}$. Then $G_i$ is precisely the graph obtained from $G$ by adding all
	edges between $t$ and $D_i$. By (9), $	D_i\cap C_s(G\sminus S)=\emptyset$.
	Since $S\in \safeSepsImp{s}{t}(G,AX)$,
	Lemma~\ref{corr:MinlsBtSepSafe}, applied with the set $AX$ in place of $A$ and with
	$D=D_i$, gives $S\in \safeSepsImp{s}{t}(G_i,AX)$.
	Thus $S$ belongs to the $i$-th edge-insertion branch.
	
	Combining the cases, every
	$S\in \safeSepsImp{s}{t}(G,AX)$ belongs either to
	$\safeSepsImp{s}{t}(G,AXY)$, or to one of the deletion branches, or to
	one of the edge-insertion branches. Therefore
	\[
	\safeSepsImp{s}{t}(G,AX)
	\subseteq
	\safeSepsImp{s}{t}(G,AXY)
	\cup
	\bigcup_{i=1}^{\ell}
	\set{T\cup\set{y_i}:T\in \safeSepsImp{s}{t}(G\sminus y_i,AX)}
	\cup
	\bigcup_{i=1}^{\ell}
	\safeSepsImp{s}{t}(G_i,AX).
	\]
\end{proof}

\subsection{Correctness Proof}
\begin{lemma}[Invariant of Algorithm~\ref{alg2e:GenSeps}]
	\label{lem:GenSepsInvariant}
	Every recursive call $\algname{Gen{-}Seps}(G,s,t,A,X,Z,k)$
	made by Algorithm~\ref{alg2e:GenSeps} satisfies the following invariant:
	\[
	\tag{Inv}
	\text{for every } Q\subseteq \nodes(G)
	\text{ and every } T\in \minsep_{sXQ,t}(G),
	\quad
	X\subseteq C_s(G\sminus T).
	\]
\end{lemma}

\begin{proof}
	We maintain a slightly stronger bookkeeping invariant. Along any root-to-current-call recursion
	path, look only at the part of the path since the most recent call in which the parameter $X$ was
	set to $\emptyset$. During this suffix of the path, the algorithm may enlarge $X$ several times.
	Let $X_1,\ldots,X_\ell$
	be the nonempty sets added to $X$ in these extension steps (lines~\ref{line2e:increaseX} and~\ref{line:forceYIn}), in chronological order. Thus
	$\ell$ is exactly the number of times $X$ was enlarged since its most recent reset to $\emptyset$.
	If no such enlargement occurred, then $\ell=0$ and the sequence is empty.
	
	For $i\in\{0,\ldots,\ell\}$, define
	\[
	\mc{X}^0\eqdef \emptyset,
	\qquad
	\mc{X}^i\eqdef X_1\cup\cdots\cup X_i
	\quad\text{for } i\ge 1.
	\]
	We prove by induction on the recursion depth that every recursive call satisfies
	\[
	\tag{H1}
	X=\mc{X}^\ell,
	\]
	and, for every $i\in\{1,\ldots,\ell\}$,
	\[
	\tag{H2}
	X_i\subseteq  \minstVertices{s\mc{X}^{i-1},t}(G).
	\]
	Here $G$ is the graph of the current recursive call.
	We first show that (H1)--(H2) imply the desired invariant (Inv). Let
	$Q\subseteq \nodes(G)$ and let $T\in \minsep_{sXQ,t}(G)$.
	By (H1), $X=\mc{X}^\ell$, and hence $T\in \minsep_{s\mc{X}^\ell Q,t}(G)$.
	Corollary~\ref{corr:XInvariant}, applied to the sequence
	$X_1,\ldots,X_\ell$ and to the auxiliary set $Q$, gives $\mc{X}^\ell\subseteq C_s(G\sminus T)$,
	because (H2) is precisely the hypothesis required by that corollary. Using (H1) again, we obtain $X\subseteq C_s(G\sminus T)$.
	Thus (Inv) follows from (H1)--(H2).
	
	It remains to prove (H1)--(H2). In the initial call, $X=\emptyset$. We take $\ell=0$, so the
	sequence is empty. Then (H1) holds and (H2) is vacuous.
	Now consider a recursive call $\algname{Gen{-}Seps}(G,s,t,A,X,Z,k)$
	for which (H1)--(H2) hold, witnessed by the sequence
	$X_1,\ldots,X_\ell$. We prove that every recursive call also satisfies (H1)--(H2).
	
	First suppose that the recursive call resets the parameter $X$ to $\emptyset$. This happens in
	Lines~\ref{line2e:secondPart1}, \ref{line2e:secondPart1End},
	\ref{line:recursiveGy}, and~\ref{line:recursiveGy2}.
	For these calls, the sequence resets, and we take $\ell=0$ and the empty sequence. Then
	(H1)--(H2) hold immediately.
	
	We now consider the recursive calls that extend $X$.
	First consider Line~\ref{line2e:increaseX} where the recursive call is $\algname{Gen{-}Seps}(G,s,t,A,X\cup\set{x_i},Z,k)$,
	where $x_i\in L^t_{sX,t}(G)$.
	Since $L^t_{sX,t}(G)\in \minsep_{sX,t}(G)$, we have $x_i \in \minstVertices{sX,t}(G)$.
	Moreover, $x_i\notin X$, because an $sX,t$-separator is disjoint from $sX$.
	For this child, the sequence is extended by $X_{\ell+1}\eqdef \set{x_i}$.
	The new accumulated set is $\mc{X}^{\ell+1}=\mc{X}^\ell\cup X_{\ell+1}=X\cup\set{x_i}$,
	so (H1) holds for the recursive call. Condition (H2) continues to hold for $X_1\cup \cdots X_\ell$ because the graph $G$ is unchanged in this call. The only new
	condition is $X_{\ell+1}\subseteq \minstVertices{s\mc{X}^{\ell},t}(G)$.
	Since $\mc{X}^\ell=X$, this is exactly $\set{x_i}\subseteq \minstVertices{sX,t}(G)$,
	which follows because $x_i\in L^t_{sX,t}(G)\subseteq \minstVertices{sX,t}(G)$. Thus (H1)--(H2) hold for the recursive call in Line~\ref{line2e:increaseX}.
	
	It remains to consider the recursive call in Line~\ref{line:recursiveCommitToSep}. Let $L^*\eqdef \closestMinSep{sAX,t}(G)$.
	The recursive call is $\algname{Gen{-}Seps}(G,s,t,A,XY,Z,k)$, where $Y\in\MHS(\varUpsilon)$.
	
	We prove that
	\[
	\tag{1}
	Y\subseteq \minstVertices{sX,t}(G).
	\]
	Since (H1)--(H2) hold for the current call, the first part of the proof already gives (Inv) holds for $X$. The algorithm reaches Line~\ref{line:recursiveCommitToSep} only if $L^t_{sX,t}(G)\in \safeSeps{s}{t}(G,A)$ (see line~\ref{line2e:prefirstIf}).
	By the invariant, applied with $Q=\emptyset$ and $T=L^t_{sX,t}(G)\in \minsep_{sX,t}(G)$, we have $X\subseteq C_s(G\sminus L^t_{sX,t}(G))$.
	Since $L^t_{sX,t}(G)$ is CP with respect to $A$ (i.e., $L^t_{sX,t}(G)\in \safeSeps{s}{t}(G,A)$), we also have $A\subseteq C_s(G\sminus L^t_{sX,t}(G))$.
	Therefore $AX\subseteq C_s(G\sminus L^t_{sX,t}(G))$.
	Thus $L^t_{sX,t}(G)\in \safeSeps{s}{t}(G,AX)$ is a CP $s,t$-separator with respect to $AX$ of size
	$\kappa_{sX,t}(G)$, and hence $f_{s,t}(G,AX)\le \kappa_{sX,t}(G)$.
	Conversely, every CP $s,t$-separator with respect to $AX$ is an $sAX,t$-separator, and every
	$sAX,t$-separator is an $sX,t$-separator. Hence $\kappa_{sX,t}(G)\le \kappa_{sAX,t}(G)\le f_{s,t}(G,AX)$.	Consequently,
	\[
	\tag{2}
	\kappa_{sX,t}(G)=\kappa_{sAX,t}(G)=f_{s,t}(G,AX).
	\]
	Since $L^*\in \minsep_{sAX,t}(G)$ is a minimum $sAX,t$-separator, (2) implies $	|L^*|=\kappa_{sAX,t}(G)=\kappa_{sX,t}(G)$.
	Every $sAX,t$-separator is also an $sX,t$-separator, so $L^*\in \minsep_{sX,t}(G)$ is a minimum $sX,t$-separator.
	
	By definition, $\varUpsilon=\set{N_G(C): C\in\D}$,
	where $\D$ is a family of connected components of $G\sminus L^*$ (lines~\ref{line2e:thirdPartStart}-\ref{line2e:thirdPartComputeMHSX}). Hence, for every
	$C\in\D$, $N_G(C)\subseteq L^*$.
	Since $Y$ is a minimal hitting set of $\varUpsilon$, every vertex of $Y$ belongs to some set
	$N_G(C)$ with $C\in\D$; otherwise that vertex could be removed while preserving the hitting
	property. Therefore $Y\subseteq \bigcup_{C\in\D}N_G(C)\subseteq L^*$.
	Since $L^*\in \minsep_{sX,t}(G)$ is a minimum $sX,t$-separator, it follows that $Y\subseteq \minstVertices{sX,t}(G)$, proving (1).

	Also, $Y\cap X=\emptyset$, because $Y\subseteq L^*$ and $L^*$ is an $sAX,t$-separator, hence
	disjoint from $AX$.
	For this recursive call $\algname{Gen{-}Seps}(G,s,t,A,XY,Z,k)$, extend the sequence by the new block $X_{\ell+1}\eqdef Y$.
	The new accumulated sequence is 
	\[
	\mc{X}^{\ell+1}=\mc{X}^\ell\cup X_{\ell+1}=X\cup Y=XY,
	\]
	so (H1) holds for the recursive call. Conditions (H2) for the old blocks are unchanged. The only new
	condition is $X_{\ell+1}\subseteq \minstVertices{s\mc{X}^\ell,t}(G)$.
	Since $\mc{X}^\ell=X$, this condition is exactly (1). Hence (H1)--(H2) hold for the recursive call in
	Line~\ref{line:recursiveCommitToSep}.
	
	We have checked every recursive call. Therefore, by induction on the recursion depth,
	(H1)--(H2) hold in every recursive call. As shown at the beginning of the proof, this implies
	(Inv) in every recursive call.
\end{proof}

\begin{lemma}[Height of the recursion tree]
	\label{claim:recursion-height}
	Consider a call
	\[
	\algname{Gen{-}Seps}(G,s,t,A,X,Z,k)
	\]
	satisfying the invariant of Algorithm~\ref{alg2e:GenSeps}. Define
	\[
	\lambda(G,A,X,k)
	\eqdef
	(k+1)\bigl(2k-\kappa_{sAX,t}(G)\bigr)
	+
	\bigl(k-\kappa_{sX,t}(G)\bigr).
	\]
	Then every root-to-leaf path in the recursion subtree rooted at this call has length at most
	\rev{$4k^2+6k+3$ calls. In particular, the height of the recursion tree is $O(k^2+1)$.}
\end{lemma}
\begin{proof}
	Fix a root-to-leaf path in the recursion subtree rooted at this call.
	Throughout the proof, when we compare a parent call to a child call, the symbols
	$G,A,X,k$ refer to the parent call, and primed symbols refer to the child call.
	
	We first record two elementary facts that will be used repeatedly. First, if a recursive call creates
	at least one child, then the tests in Lines~1 and~3 have failed, and hence $\kappa_{sAX,t}(G)\le k$ and $\kappa_{sX,t}(G)\le k$.
	Consequently, for every non-terminal call on the path,
	\[
	\tag{1}
	0\le \lambda(G,A,X,k)\le 2k^2+3k.
	\]
	Indeed, the lower bound follows from
	$\kappa_{sAX,t}(G)\le k$ and $\kappa_{sX,t}(G)\le k$, while the upper bound follows from
	$\kappa_{sAX,t}(G),\kappa_{sX,t}(G)\ge 0$:
	\[
	\lambda(G,A,X,k)
	\le (k+1)\cdot 2k+k
	=
	2k^2+3k.
	\]
	Second, enlarging the source side cannot decrease the corresponding minimum separator size.
	That is, if $P\subseteq P'$, then $	\kappa_{sP,t}(G)\le \kappa_{sP',t}(G)$.

\rev{Call a node \emph{live} if it passes the two initial feasibility
tests. A nonlive node is a leaf and can only be the last node of our
path. In the potential comparisons below we consider only edges whose
endpoints are live; we count a possible final edge to a nonlive leaf
separately.}

	We now classify recursive edges into two types. A recursive edge is called \emph{ordinary} if it is
	not the edge created by Line~\ref{line:recursiveCommitToSep}. We show first that every ordinary
	recursive edge strictly decreases $\lambda$.
	Consider Line~\ref{line2e:increaseX}. The recursive call is $\algname{Gen{-}Seps}(G,s,t,A,X\cup\set{x_i},Z,k)$,
	where $x_i\in L^t_{sX,t}(G)$. By Lemma~\ref{lem:potentialReducedPreFirstPart}, we have $\kappa_{sXx_i,t}(G)>\kappa_{sX,t}(G)$.
	Moreover, $\kappa_{sAXx_i,t}(G)\ge \kappa_{sAX,t}(G)$. Therefore the recursive call in line~\ref{line2e:increaseX} results in a strict reduction of the potential function:	$\lambda(G,A,X\cup\set{x_i},k) < \lambda(G,A,X,k)$.

	Consider next the recursive call in Line~\ref{line2e:secondPart1}. Let $L^*\eqdef L^*_{sAX,t}(G)$
	and set $b\eqdef |L^*|$. The recursive call is
	\[
	\algname{Gen{-}Seps}
	(G\sminus L^*,s,t,AX,\emptyset,Z\cup L^*,k-b).
	\]
	Since $L^*\in \minsep_{sAX,t}(G)$ is a minimum $sAX,t$-separator, we have $	b=\kappa_{sAX,t}(G)$.
	Moreover, in the graph $G\sminus L^*$ the empty set separates $sAX$ from $t$, and hence $\kappa_{sAX,t}(G\sminus L^*)=0$.
	Thus the potential associated with the recursive call is $2(k-b)^2+3(k-b)$.
	Let
	\[
	a\eqdef \kappa_{sAX,t}(G)=b,
	\qquad
	c\eqdef \kappa_{sX,t}(G).
	\]
	Since every $sAX,t$-separator is also an $sX,t$-separator, we have $c\le a=b$. Hence
	\[
	\begin{aligned}
		&\lambda(G,A,X,k)-\lambda(G\sminus L^*,AX,\emptyset,k-b)\\
		&=
		\bigl((k+1)(2k-b)+(k-c)\bigr)
		-
		\bigl(2(k-b)^2+3(k-b)\bigr)\\
		&=
		b(3k+2-2b)-c\\
		&\ge
		b(3k+2-2b)-b\\
		&=
		b(3k+1-2b).
	\end{aligned}
	\]
	Since $1\le b\le k$, the last expression is strictly positive. Therefore the recursive call in
	Line~\ref{line2e:secondPart1} strictly decreases $\lambda$.
	
	Consider now the recursive calls in Lines~\ref{line2e:secondPart1End} and~\ref{line:recursiveGy2}.
	In both cases, the call has the form $\algname{Gen{-}Seps}(G_i,s,t,AX,\emptyset,Z,k)$,
	where $G_i$ is obtained from $G$ by adding edges from $t$ to the closed neighborhood of
	a vertex in $\closestMinSep{sAX,t}(G)$. By Lemma~\ref{lem:sAtMinlSepLargerThanMinsAtSep_s}, it holds that $\kappa_{sAX,t}(G_i)>\kappa_{sAX,t}(G)$.
	The first term of the potential is therefore decreased by at least $k+1$. The second term may increase,
	because the child has $X'=\emptyset$, but it can increase by at most $k$: the old second term is
	nonnegative, and the new second term is at most $k$. Hence the total potential decreases by at least one. Thus the recursive calls in lines~\ref{line2e:secondPart1End} and~\ref{line:recursiveGy2} strictly decrease $\lambda$.
	
	Consider finally the recursive calls in Line~\ref{line:recursiveGy}, which have the form
	\[
	\algname{Gen{-}Seps}(G\sminus y_i,s,t,AX,\emptyset,Z\cup\set{y_i},k-1).
	\]
	Let $a\eqdef \kappa_{sAX,t}(G)$.
	In the branch where Line~\ref{line:recursiveGy} is reached, it holds that $	\kappa_{sX,t}(G)=\kappa_{sAX,t}(G)=a$.
	Deleting one vertex can decrease vertex-connectivity by at most one, and therefore $\kappa_{sAX,t}(G\sminus y_i)\ge a-1$.
	Thus the first term of the potential function in the recursive call is at most
	\[
	k\bigl(2(k-1)-(a-1)\bigr)
	=
	k(2k-a-1),
	\]
	and the second term of the child potential is at most $k-1$. Hence
	\[
	\lambda(G\sminus y_i,AX,\emptyset,k-1)
	\le
	k(2k-a-1)+(k-1).
	\]
	On the other hand, the parent potential is
	\[
	\lambda(G,A,X,k)
	=
	(k+1)(2k-a)+(k-a).
	\]
	Therefore
	\[
	\begin{aligned}
		&\lambda(G,A,X,k)-\lambda(G\sminus y_i,AX,\emptyset,k-1)\\
		&\ge
		(k+1)(2k-a)+(k-a)
		-
		\bigl(k(2k-a-1)+(k-1)\bigr)\\
		&=
		3k+1-2a.
	\end{aligned}
	\]
	Since $a\le k$, this is at least $k+1>0$. Hence the recursive call in
	Line~\ref{line:recursiveGy} strictly decreases $\lambda$.

	It remains to handle the recursive call in Line~\ref{line:recursiveCommitToSep}. The call in that line is \[\algname{Gen{-}Seps}(G,s,t,A,XY,Z,k).\]
	Here the graph, the set $A$, and the budget $k$ are unchanged, and only the set $X$ is enlarged.
	Therefore
	\[
	\kappa_{sAXY,t}(G)\ge \kappa_{sAX,t}(G)
	\quad\text{and}\quad
	\kappa_{sXY,t}(G)\ge \kappa_{sX,t}(G).
	\]
	Consequently,
	\[
	\tag{2}
	\lambda(G,A,XY,k)\le \lambda(G,A,X,k).
	\]
	Thus Line~\ref{line:recursiveCommitToSep} never increases the potential, but it need not decrease
	it strictly.
	
	We next show that Line~\ref{line:recursiveCommitToSep} cannot occur twice consecutively on a
	root-to-leaf path. Indeed, by the choice of $Y$ in Line~\ref{line2e:thirdPartComputeMHSX} and
	Corollary~\ref{corr:mainThmBoundedCardinalityExt}, the separator $\closestMinSep{sAXY,t}(G)$ (see line~\ref{line2e:thirdPartComputeMHSX})
	has cardinality $\kappa_{sX,t}(G)\leq f_{s,t}(G,AX)$ and is CP for $AX$. That is, 
	$\closestMinSep{sAXY,t}(G)\in \minsep_{sX,t}(G)\cap \safeSeps{s}{t}(G,AX)$. By Corollary~\ref{corr:mainThmBoundedCardinalityExt}, $\closestMinSep{sAXY,t}(G)\in \safeSepskImp{s}{t}{\min}(G,AX)$.
	In particular, it is CP with respect to $A$ and is an $sAXY,t$-separator of size
	$\kappa_{sX,t}(G)$. Since enlarging the source side from $sX$ to $sXY$ cannot decrease
	connectivity, we get $\kappa_{sXY,t}(G)=\kappa_{sX,t}(G)$.
	Therefore $\closestMinSep{sAXY,t}(G)$ is a minimum $sXY,t$-separator that is CP with respect to $A$.
	It follows that the unique closest-to-$t$ minimum separator $L^t_{sXY,t}(G)$ (see Lemma~\ref{lem:uniqueMinImportantSep}) is also CP with respect to
	$A$: the closest-to-$t$ separator leaves an inclusionwise largest $s$-side among minimum $sXY,t$-separators, by Lemma~\ref{lem:inclusionCsCt} and the definition of $L^t$.
	
	Hence, in the child created by Line~\ref{line:recursiveCommitToSep}, the test in
	Line~\ref{line2e:prefirstIf} fails. Moreover, the separator $L^*_{sAXY,t}(G)$ itself is CP with
	respect to $A$, so the algorithm enters the branch beginning at Line~\ref{line2e:elseIfBegin}. Thus,
	if the path continues after a Line~\ref{line:recursiveCommitToSep} edge, its next recursive edge is
	one of the ordinary recursive edges already treated above. In particular, two
	Line~\ref{line:recursiveCommitToSep} edges cannot be consecutive.
	
\begin{revision}
We can now bound the length of the path. If its root is not live, it
consists of one call. Otherwise let $M\eqdef 2k^2+3k$, using the root
budget. The potential at the root is at most $M$, remains nonnegative
on live calls, and decreases strictly along every ordinary edge
between live calls. Thus there are at most $M$ such ordinary edges.
The only other edges are those created by
Line~\ref{line:recursiveCommitToSep}, and no two of these are consecutive.
There are at most $M+1$ of them. The live part of the path therefore
contains at most $2M+2$ calls. A possible final nonlive leaf contributes
at most one more, giving at most $2M+3=4k^2+6k+3$ calls.
\end{revision}
	This proves the claimed height bound.
\end{proof}

\begin{theorem}[Completeness of Algorithm~\ref{alg2e:GenSeps}]
	\label{thm:GenSepsCompleteness}
	Let $(G_0,s,t,A_0,k_0)$ be the input instance, and let
	$\mc C(G_0,A_0,k_0)$ be the family of candidate separators generated by
	Algorithm~\ref{alg2e:GenSeps}. Then
	\[
	\safeSepskImp{s}{t}{k_0}(G_0,A_0)
	\subseteq
	\mc C(G_0,A_0,k_0).
	\]
\end{theorem}
\begin{proof}
	We prove a stronger statement for every recursive call $\algname{Gen{-}Seps}(G,s,t,A,X,Z,k)$
	made by Algorithm~\ref{alg2e:GenSeps}. Let $\mc C(G,A,X,Z,k)$
	denote the family of candidates output by the recursion subtree rooted at this call. We prove that
	\[
	\tag{$\star$}
	\set{Z\cup S : S\in \safeSepskImp{s}{t}{k}(G,AX)}
	\subseteq
	\mc C(G,A,X,Z,k).
	\]
	The theorem follows immediately from $(\star)$ applied to the initial call, where
	$G=G_0$, $A=A_0$, $X=\emptyset$, $Z=\emptyset$, and $k=k_0$.
	
	We prove $(\star)$ by induction on the height of the recursion subtree rooted at the considered
	call. This height is finite by Lemma~\ref{claim:recursion-height}. Moreover, every recursive call
	satisfies the invariant of Algorithm~\ref{alg2e:GenSeps} by Lemma~\ref{lem:GenSepsInvariant}.
	
	Fix a recursive call $\algname{Gen{-}Seps}(G,s,t,A,X,Z,k)$,
	and assume that $(\star)$ holds for all strict descendants of this call. We prove $(\star)$ for the
	current call.
	
	First suppose that the algorithm returns in Line~\ref{line2e:L_sAtt>k}. Then $	\kappa_{sAX,t}(G)>k$.
	Every separator in $\safeSepskImp{s}{t}{k}(G,AX)$ is, in particular, an $sAX,t$-separator. Hence
	every such separator has size at least $\kappa_{sAX,t}(G)>k$, and therefore $\safeSepskImp{s}{t}{k}(G,AX)=\emptyset$.
	Thus $(\star)$ holds.
	
	Next suppose that the algorithm returns in Line~\ref{line2e:L_stt>k}. Then $	|L^t_{sX,t}(G)|=\kappa_{sX,t}(G)>k$.
	Every separator in $\safeSepskImp{s}{t}{k}(G,AX)$ is an $sX,t$-separator, because it preserves
	$AX$ on the $s$-side. Hence every such separator has size at least $\kappa_{sX,t}(G)>k$, and again $\safeSepskImp{s}{t}{k}(G,AX)=\emptyset$.
	Thus $(\star)$ holds.
	
	Now suppose that the algorithm reaches Line~\ref{line2e:returnZEmptyIf}, namely $L^t_{sX,t}(G)=\emptyset$.
	Since $\emptyset\in L_{sX,t}(G)$, the invariant (Inv) from Lemma~\ref{lem:GenSepsInvariant} applied with $Q=\emptyset$ and
	$T=\emptyset$ gives $X\subseteq C_s(G)$.
	If $A\not\subseteq C_s(G)$, then $\emptyset$ is not CP with respect to $AX$. Since
	$\emptyset$ already separates $sX$ from $t$, every nonempty $s,t$-separator is not minimal.
	Hence there is no minimal CP $s,t$-separator with respect to $AX$, and $	\safeSepskImp{s}{t}{k}(G,AX)=\emptyset$.
	The algorithm returns without outputting a candidate, so $(\star)$ holds.
	
	If $A\subseteq C_s(G)$, then $AX\subseteq C_s(G)$. Since $s$ and $t$ are already separated by
	the empty set, the only minimal $s,t$-separator is $\emptyset$. Therefore
	$\safeSepskImp{s}{t}{k}(G,AX)=\set{\emptyset}$.
	The algorithm outputs $Z$ in Line~\ref{line2e:returnZEmpty}, and hence
 	$Z=Z\cup\emptyset\in \mc C(G,A,X,Z,k)$. Thus $(\star)$ holds in this case as well.

	We may therefore assume from now on that the algorithm does not return in Lines~\ref{line2e:L_sAtt>k}--\ref{line2e:returnZEmpty2}.
	Consider the branch beginning at Line~\ref{line2e:prefirstIf}. In this case $L^t_{sX,t}(G)\notin \safeSeps{s}{t}(G,A)$.
	Let $L^t_{sX,t}(G)=\set{x_1,\ldots,x_\ell}$.
	By Lemma~\ref{lem:partZeroOfAlg}, applied with $AX$ as the connectivity set, we have
	\[
	\tag{1}
	\safeSepskImp{s}{t}{k}(G,AX)
	\subseteq
	\bigcup_{i=1}^{\ell}
	\safeSepskImp{s}{t}{k}(G,AX\cup\set{x_i}).
	\]
	For every $i\in\{1,\ldots,\ell\}$, the algorithm makes the recursive call $	\algname{Gen{-}Seps}(G,s,t,A,X\cup\set{x_i},Z,k)$ in line~\ref{line2e:increaseX}.
	By the induction hypothesis applied to this child,
	\[
	\set{Z\cup S :
		S\in \safeSepskImp{s}{t}{k}(G,AX\cup\set{x_i})}
	\subseteq
	\mc C(G,A,X\cup\set{x_i},Z,k).
	\]
	Since the outputs of this child are included in the outputs of the current recursion subtree, taking
	the union over all $i$ and using (1) proves $(\star)$ for the current call.
	
	It remains to consider the case where the algorithm does not enter the if-block in Line~\ref{line2e:prefirstIf}. Thus $L^t_{sX,t}(G)\in \safeSeps{s}{t}(G,A)$.
	Suppose first that the algorithm enters the if-block at Line~\ref{line2e:elseIfBegin}, namely
	$\closestMinSep{sAX,t}(G) \in \safeSeps{s}{t}(G,A)$.
	Let $L^*\eqdef \closestMinSep{sAX,t}(G) $.
	By Corollary~\ref{corr:partTwoOfAlg},
	\[
	\tag{2}
	\safeSepsImp{s}{t}(G,AX)
	\subseteq
	\{L^*\}
	\cup
	\bigcup_{x\in L^*}
	\safeSepsImp{s}{t}(G_x,AX),
	\]
	where $G_x$ is obtained from $G$ by adding all edges between $t$ and $N_G[x]$.
	Intersecting (2) with separators of size at most $k$ gives the corresponding budgeted inclusion.
	
	We first handle the separator $L^*$. Since $L^*\in \minsep_{sAX,t}(G)$, the invariant (Inv) from Lemma~\ref{lem:GenSepsInvariant} applied with
	$Q=A$ and $T=L^*$ gives $X\subseteq C_s(G\sminus L^*)$.
	The assumption of the current branch gives $A\subseteq C_s(G\sminus L^*)$.
	Hence $AX\subseteq C_s(G\sminus L^*)$.
	Consider the recursive call in Line~\ref{line2e:secondPart1}:
	\[
	\algname{Gen{-}Seps}
	(G\sminus L^*,s,t,AX,\emptyset,Z\cup L^*,k-|L^*|).
	\]
	In this child instance, the empty set is a CP important $s,t$-separator with respect to $AX$:
	indeed, $L^*$ separates $sAX$ from $t$ in $G$, and after deleting $L^*$ the set $AX$ lies in the
	$s$-component (i.e., $AX\subseteq C_s(G\sminus L^*)$). Thus $\emptyset \in \safeSepskImp{s}{t}{k-|L^*|}(G\sminus L^*,AX)$.
	Applying the induction hypothesis to the subtree rooted at the recursive call of Line~\ref{line2e:secondPart1} gives
	\[
	(Z\cup L^*)\cup\emptyset
	=
	Z\cup L^*
	\in
	\mc C(G\sminus L^*,AX,\emptyset,Z\cup L^*,k-|L^*|).
	\]
	Therefore $Z\cup L^*$ is output by the recursion subtree of the current call.
	
	Now let $x\in L^*$. The algorithm makes the recursive call $\algname{Gen{-}Seps}(G_x,s,t,AX,\emptyset,Z,k)$ in line~\ref{line2e:secondPart1End}.
	By the induction hypothesis applied to this child,
	\[
	\set{Z\cup S :
		S\in \safeSepskImp{s}{t}{k}(G_x,AX)}
	\subseteq
	\mc C(G_x,AX,\emptyset,Z,k).
	\]
	Combining these inclusions for all $x\in L^*$ with (2), and using the already proved coverage of
	the separator $L^*$ itself, proves $(\star)$ in the branches of the tree generated in the if-block in line Line~\ref{line2e:elseIfBegin}.
	
	We are left with the final branch of the algorithm. Thus
	\[
	L^t_{sX,t}(G)\in \safeSeps{s}{t}(G,A)
	\quad\text{and}\quad
	L^*_{sAX,t}(G)\notin \safeSeps{s}{t}(G,A).
	\]
	By Lemma~\ref{lem:GenSepsInvariant}, the invariant holds in the current call. Hence the hypotheses
	needed for Corollary~\ref{corr:mainThmBoundedCardinalityExt} are satisfied in this branch.
	Let
	\[
	\D=\{C\in\cc(G\sminus L^*_{sAX,t}(G)) : s\notin C,\ C\cap A\neq\emptyset\},
	\qquad
	\varUpsilon=\{N_G(C):C\in\D\}.
	\]
	The algorithm chooses $Y\in\MHS(\varUpsilon)$
	as in Line~\ref{line2e:thirdPartComputeMHSX}, and then makes the recursive calls in Lines~\ref{line:recursiveCommitToSep},~\ref{line:recursiveGy}, and~\ref{line:recursiveGy2}. Write $	Y=\set{y_1,\ldots,y_q}$.
	By Lemma~\ref{lem:thirdPart},
	\[
	\tag{3}
	\safeSepsImp{s}{t}(G,AX)
	\subseteq
	\safeSepsImp{s}{t}(G,AXY)
	\cup
	\bigcup_{i=1}^{q}
	\{T\cup\set{y_i}:T\in\safeSepsImp{s}{t}(G\sminus y_i,AX)\}
	\cup
	\bigcup_{i=1}^{q}
	\safeSepsImp{s}{t}(G_i,AX),
	\]
	where $G_i$ is obtained from $G$ by adding all edges between $t$ and
	$N_G[y_i]{\setminus}\set{t}$.
	
	Let $S\in \safeSepskImp{s}{t}{k}(G,AX)$.
	We prove that $Z\cup S$ is output by one of the children in the this branch.
	If $S\in \safeSepsImp{s}{t}(G,AXY)$, then, since $|S|\le k$, $S\in \safeSepskImp{s}{t}{k}(G,AXY)$.
	The recursive call in line~\ref{line:recursiveCommitToSep} is $	\algname{Gen{-}Seps}(G,s,t,A,XY,Z,k)$.
	By the induction hypothesis applied to this child, $Z\cup S$ is output.
	
	Otherwise, by (3), $S=T\cup \set{y_i}$ where $T\in \safeSepskImp{s}{t}{k}(G\minus y_i,AX)$ for some $y_i\in Y$, or $S\in \safeSepskImp{s}{t}{k}(G_i,AX)$.
	Suppose first that, for some $i\in\set{1,\ldots,q}$, $	S=T\cup\set{y_i}$
	where $T\in\safeSepsImp{s}{t}(G\sminus y_i,AX)$.
	Since $T$ is a separator in $G\sminus y_i$, it does not contain $y_i$. Hence $|T|=|S|-1\le k-1$.
	Therefore $T\in\safeSepskImp{s}{t}{k-1}(G\sminus y_i,AX)$.
	The corresponding recursive call in Line~\ref{line:recursiveGy} is
	\[
	\algname{Gen{-}Seps}
	(G\sminus y_i,s,t,AX,\emptyset,Z\cup\set{y_i},k-1).
	\]
	By the induction hypothesis applied to this child, $(Z\cup\set{y_i})\cup T=	Z\cup S$
	is output.
	
	Finally suppose that, for some $i\in\set{1,\ldots,q}$, $S\in \safeSepsImp{s}{t}(G_i,AX)$.
	Since $|S|\le k$, this means $S\in \safeSepskImp{s}{t}{k}(G_i,AX)$.
	The corresponding recursive call in Line~\ref{line:recursiveGy2} is \[\algname{Gen{-}Seps}(G_i,s,t,AX,\emptyset,Z,k).\]
	By the induction hypothesis applied to this child, $Z\cup S$ is output.
	
	Thus every $S\in \safeSepskImp{s}{t}{k}(G,AX)$
	is covered by one of the recursive children in the final branch, and therefore $(\star)$ holds in this case.
	
	All possible execution cases have been considered. Hence, by induction on the recursion subtree
	height, $(\star)$ holds for every recursive call. Applying $(\star)$ to the initial call gives
	\[
	\safeSepskImp{s}{t}{k_0}(G_0,A_0)
	\subseteq
	\mc C(G_0,A_0,k_0),
	\]
	as required.
\end{proof}

\begin{theorem}
	\label{thm:detailedRuntimeAnalysis}
	Let $G$ be an undirected graph with $n$ vertices and $m$ edges, let
	$s,t\in \nodes(G)$, $A\subseteq \nodes(G){\setminus}\set{s,t}$, and let $k\in\mathbb N_{\geq 0}$.
	\rev{Let $T(n,m)$ be as defined in the introduction.} Algorithm~\algname{Gen-Seps} (Algorithm~\ref{alg2e:GenSeps}), on input
	$(G,s,t,A,\emptyset,\emptyset,k)$, outputs a family of vertex sets $\mc C(G,A,k)\subseteq 2^{\nodes(G)}$
	such that
	\begin{align*}
		\safeSepskImp{s}{t}{k}(G,A) \subseteq 	\mc C(G,A,k), \text{ and }	|\mc C(G,A,k)| \in 2^{O(k^2\rev{\log k})}.
	\end{align*}
	Moreover, the total running time of the algorithm is $O\!\left(2^{O(k^2\rev{\log k})}\cdot n\cdot T(n,m)\right)$.
\end{theorem}
\begin{proof}
	The containment $\safeSepskImp{s}{t}{k}(G,A)\subseteq \mc C(G,A,k)$
	is exactly Theorem~\ref{thm:GenSepsCompleteness}.
	
	It remains to bound the size of the output family and the running time. By Lemma~\ref{lem:GenSepsInvariant},
	every recursive call made by Algorithm~\ref{alg2e:GenSeps} satisfies the invariant (Inv).
	Therefore Lemma~\ref{claim:recursion-height} applies to every recursive call. In particular, every
	\rev{root-to-leaf path in the recursion tree contains at most $4k^2+6k+3=O(k^2+1)$ calls.}
	
	We next bound the branching factor. If a call enters the if-block in line~\ref{line2e:prefirstIf}, then it generates $|L^t_{sX,t}(G)|\leq k$ children. If a call enters the if-block at Line~\ref{line2e:elseIfBegin}, then it creates one child in Line~\ref{line2e:secondPart1} and
	$|L^*_{sAX,t}(G)|$ in the loop in lines \ref{line2e:secondPart1LoopBegin}--\ref{line2e:secondPart1End}. Overall, a call entering this block will generate at most $1+	|L^*_{sAX,t}(G)|=1+\kappa_{sAX,t}(G)\leq k+1$ children.

	Finally, in the last branch, the algorithm chooses
	$Y\in\MHS(\varUpsilon)$. By the construction in Lines~\ref{line2e:thirdPartStart}--\ref{line2e:thirdPartComputeMHSX}, every set in $\varUpsilon$ is a
	subset of $L^*_{sAX,t}(G)$, and every vertex of the minimal hitting set $Y$ belongs to some set in $\varUpsilon$. Hence $Y\subseteq L^*_{sAX,t}(G)$, and therefore $|Y|\le |L^*_{sAX,t}(G)|=\kappa_{sAX,t}(G)\le k$.
	Thus the else-block in lines~\ref{line2e:thirdPartStart}--\ref{line:recursiveGy2} makes $1+2|Y|\le 2k+1$ recursive calls, generating at most $2k+1$ child nodes.
	Hence every node has $O(k)$ children.
	
	It follows from the height bound of Lemma~\ref{claim:recursion-height}, and the branching bound that the total number of nodes in the
	recursion tree is at most $2^{O(k^2\log k)}$.
	The algorithm outputs at most one candidate at a leaf, and hence $	|\mc C(G,A,k)|\le 2^{O(k^2\rev{\log k})}$.
	
\rev{When $k=0$, the initial tests and the empty-separator case
terminate without branching. For $k=1$, the height and branching factor
are bounded by constants. For $k\ge2$, the preceding estimates give
the stated asymptotic bound. Every graph modification deletes vertices
or adds edges incident with $s$ or $t$, so all graphs used in the
recursion have at most $n$ vertices and $O(n+m)$ edges, where $n,m$
refer to the original input. On each such graph, minimum separators
and the canonical separators can be computed in $O(n(n+m))$ time
by Lemma~\ref{lem:uniqueMinImportantSep}, since every finite separator
size is at most $n$; adjacent endpoint sets are detected in $O(n+m)$
time and have separator value $+\infty$. Because $T(n,m)$ includes
input reading, $T(n,m)=\Omega(n+m)$, so these computations take
$O(n\cdot T(n,m))$ time. This accounts for the added edges without
changing the meaning of the second argument of $T$.}

	It remains to bound the work done at a single node. Apart from
	Line~\ref{line2e:thirdPartComputeMHSX}, each node performs a constant
	number of minimum-separator computations and connectivity tests.
	Constructing its $O(n)$ child instances takes $O(n(n+m))$ time.
	In Line~\ref{line2e:thirdPartComputeMHSX},
	Corollary~\ref{corr:mainThmBoundedCardinalityExt} is used to find the required
	minimal hitting set $Y$ and the corresponding separator. Its implementation
	uses $O(n(n+m))$ time per candidate, also on the augmented graphs,
	and takes $O(2^{\kappa_{sX,t}(G)}\cdot n(n+m))$
	$\subseteq O(2^{\kappa_{sX,t}(G)}\cdot n\cdot T(n,m))$ time.
	Since $\kappa_{sX,t}(G)\le k$, every node (recursive call) is processed in time $O(2^k\cdot n\cdot T(n,m))$, up to polynomial factors already dominated by this bound.
	
	Multiplying the per-node cost by the number of nodes gives total running time
	\[
	2^{O(k^2\rev{\log k})}\cdot O(2^k\cdot n\cdot T(n,m))
	=
	O\!\left(2^{O(k^2\rev{\log k})}\cdot n\cdot T(n,m)\right).
	\]
	This proves both the claimed output-size bound and the claimed running-time bound.
\end{proof}

\begin{reptheorem}{\ref{thm:singleExponentialSafeImportantOfSizek}}
	\singleExponentialSafeImportantOfSizek
\end{reptheorem}
\begin{proof}
	Apply Algorithm~\ref{alg2e:GenSeps} ($\algname{Gen{-}Seps}$) to the instance $(G,s,t,A,k)$, with initial sets $X=\emptyset$ and $Z=\emptyset$. Let $\mathcal C(G,A,k)$ be the family of
	sets output by the algorithm.
	
	By Theorem~\ref{thm:GenSepsCompleteness}, the algorithm is complete: every
	connectivity-preserving important $s,t$-separator with respect to $A$ of size at
	most $k$ belongs to $\mathcal C(G,A,k)$. In other words, $\safeSepskImp{s}{t}{k}(G,A)\subseteq \mathcal C(G,A,k)$.
	By Theorem~\ref{thm:detailedRuntimeAnalysis}, the output family satisfies $	|\mathcal C(G,A,k)| \in 2^{O(k^2\rev{\log k})}$
	and can be computed in time $O\!\left(2^{O(k^2\rev{\log k})}\cdot n\cdot T(n,m)\right)$.

	Therefore, $|\safeSepskImp{s}{t}{k}(G,A)|\in 2^{O(k^2\rev{\log k})}$.
	
\begin{revision}
For exact enumeration, let $F$ be the number of generated candidates,
counting repetitions. In the original graph, first discard each set $S$
that has size greater than $k$, meets $\{s,t\}\cup A$, fails to separate
$s$ from $t$, or fails $A\subseteq C_s(G\sminus S)$. Also discard it
unless it is an inclusionwise minimal $s,t$-separator. This last test
uses Lemma~\ref{lem:fullComponents}: both $C_s(G\sminus S)$ and
$C_t(G\sminus S)$ must have neighborhood exactly $S$.
Remove repetitions and store $R(S)\eqdef C_s(G\sminus S)$ for every
remaining candidate. Finally, discard $S$ if another remaining
candidate $T$ satisfies $|T|\le |S|$ and $R(T)\subsetneq R(S)$.

Every important separator survives these tests. Conversely, suppose
a remaining CP separator $S$ is not important. There is a CP separator
of no greater size with a strictly smaller source component. Repeating
this choice must terminate, since source components strictly decrease
in a finite graph. It terminates at an important separator $U$ with
$|U|\le |S|$ and $R(U)\subsetneq R(S)$. Completeness places $U$ in
$\mathcal C(G,A,k)$, and the feasibility and minimality tests retain it.
Thus the dominance test discards $S$. The final family is therefore
exactly $\safeSepskImp{s}{t}{k}(G,A)$, with no repetitions.

The feasibility and minimality tests take $O(F(n+m))$ time. Deduplication
and pairwise dominance tests take $O(F^2n)$ time using incidence vectors
for the vertex sets. As $F=2^{O(k^2\log k)}$ and $T(n,m)$ includes
input reading, these costs fit
$O(2^{O(k^2\log k)}\cdot n\cdot T(n,m))$.
This proves the exact enumeration bound.
\end{revision}
\end{proof}
\clearpage
\DontPrintSemicolon
\newcommand\mycommfont[1]{\footnotesize\ttfamily\textcolor{blue}{#1}}
\SetCommentSty{mycommfont}
\newcommand*{\tikzmk}[1]{\tikz[remember picture,overlay,] \node (#1) {};\ignorespaces}
\newcommand{\boxit}[1]{\tikz[remember picture,overlay]{\node[yshift=3pt,fill=#1,opacity=.25,fit={(A)($(B)+(.92\linewidth,.8\baselineskip)$)}]
		{};}\ignorespaces}
\colorlet{mypink}{red!40}
\colorlet{myblue}{cyan!50}
\colorlet{myyellow}{yellow!50}

\SetKwFor{ForNoEnd}{for}{do}{}  

\begin{algocf}[H]
	\SetAlgoLined
	\caption{{$\boldsymbol{\algname{Gen{-}Seps}}$}: \rev{Generating candidates for CP-important separators} \label{alg2e:GenSeps}}
	\setcounter{AlgoLine}{0}  
	\SetKwProg{Algorithm2}{Algorithm1}{:}{}
	\KwIn{Graph $G$, $s,t\in \nodes(G)$, $A,X\subseteq\nodes(G)\setminus\{s,t\}$, previously selected vertices $Z$, and $k\in \mathbb{N}_{\geq 0}$.}	
	\KwOut{\rev{A family containing $Z\cup S$ for every $S\in\safeSepskImp{s}{t}{k}(G,AX)$; exact filtering is performed in the original graph.}}
	\setcounter{AlgoLine}{0} 
	\tcp*[l]{Invariant: for every $Q\subseteq \nodes(G)$, and $T\in \minsep_{sXQ,t}(G)$, it holds that $X\subseteq C_s(G\sminus T)$.}
	\lIf{$\kappa_{sAX,t}(G)> k$}{\Return  \label{line2e:L_sAtt>k}}
	Compute $L_{sX,t}^t(G)\in \minsep_{sX,t}(G)$ closest to $t$. \label{line2e:L_stt}  \tcp*[l]{By Lem.~\ref{lem:uniqueMinImportantSep}, $L_{sX,t}^t$ is unique and computed in $O(n\cdot T(n,m))$.}
	\lIf{$|L_{sX,t}^t(G)| > k$}{\Return  \label{line2e:L_stt>k}}
	\uIf{$L_{sX,t}^t(G) = \emptyset$\label{line2e:returnZEmptyIf}}{
		\lIf{$A\subseteq C_s(G)$}{Output $Z$} \label{line2e:returnZEmpty}
		\Return \label{line2e:returnZEmpty2}
	}
	Let $L_{sX,t}^t(G)=\set{x_1,\dots,x_\ell}$ \label{line2e:L_{s,t}^tNormal} \; 
	
	\uIf{$L_{sX,t}^t(G) \notin \safeSeps{s}{t}(G,A)$  \label{line2e:prefirstIf}}{
		\For{$i=1$ to $i=\ell$}{
			{$\boldsymbol{\algname{Gen{-}Seps}}$}$(G, s,t,A, X{\cup} \set{x_i}, Z,k)$  \label{line2e:increaseX} \tcp*[l]{\rev{Lemmas~\ref{lem:potentialReducedPreFirstPart}, \ref{lem:partZeroOfAlg}; Corollary~\ref{corr:XInvariant}.}}

		} \label{line2e:prefirstIfEnd}
	}
	\Else{
	\uIf{$\closestMinSep{sAX,t}(G) \in \safeSeps{s}{t}(G,A)$ \label{line2e:elseIfBegin}}{
			{$\boldsymbol{\algname{Gen{-}Seps}}$} $(G\sminus \closestMinSep{sXA,t}(G), s,t,AX,\emptyset,Z\cup \closestMinSep{sAX,t}(G),k\sminus |\closestMinSep{sAX,t}(G)|)$ \;  \label{line2e:secondPart1} 
			
			Let $\closestMinSep{sAX,t}(G)=\set{x_1,\dots,x_\ell}$ \; \tcp*[l]{By Lem.~\ref{lem:uniqueMinImportantSep}, $\closestMinSep{sAX,t}(G)$ is unique and computed in $O(n\cdot T(n,m))$.} 
			
			\For{$i=1$ to $i=\ell$\label{line2e:secondPart1LoopBegin}}{
				Let $G_i$ be the graph that results from $G$ by adding all edges between $t$ and $N_G[x_i]\rev{{\setminus}\set{t}}$  \; \tcp*[l]{\rev{Corollary~\ref{corr:partTwoOfAlg}.}}
				
				{$\boldsymbol{\algname{Gen{-}Seps}}$}$(G_i, s,t,AX, \emptyset, Z,k)$  \label{line2e:secondPart1End} \tcp*[l]{$\kappa_{sAX,t}(G_i){\geq} \kappa_{sAX,tx_i}(G){>} \kappa_{sAX,t}(G)$ (Lem.~\ref{lem:sAtMinlSepLargerThanMinsAtSep_s} (1)).}
			}
		}
		\uElse{
				\tcp*[l]{\rev{The premise of Corollary~\ref{corr:mainThmBoundedCardinalityExt} now holds.}}
			$\begin{aligned}[t]
				\D \eqdef \set{C\in \cc(G\sminus \closestMinSep{sAX,t}(G)): s\notin C, C\cap A \neq \emptyset}
			\end{aligned}$\label{line2e:thirdPartStart}\;
			
			$\varUpsilon \eqdef \set{N_{G}(C): C\in \D}$\;
			
			Let $Y\in \MHS(\varUpsilon)$ such that $\closestMinSep{sAXY,t}(G)\in \safeSepskImp{s}{t}{\min}(G,AX)$ 
		 \label{line2e:thirdPartComputeMHSX} \tcp*[f]{\rev{Use the exact enumeration in Corollary~\ref{corr:mainThmBoundedCardinalityExt}.}}\;
			
		
		$\boldsymbol{\algname{Gen{-}Seps}}(G,s,t,A,XY,Z,k)$ \label{line:recursiveCommitToSep}\; \tcp*[f]{\rev{The next call enters Phase~2.}\label{line:forceYIn}}

		Let $Y=\set{y_1,\ldots,y_q}$\;
		
		\For{$i=1$ to $i=q$}{
				$\boldsymbol{\algname{Gen{-}Seps}}(G\sminus y_i,s,t,AX,\emptyset,Z\cup\set{y_i},k-1)$\; \label{line:recursiveGy} \tcp*[l]{See Lemma~\ref{lem:thirdPart}.} 
				
			 Let $G_i$ be obtained from $G$ by adding all edges between $t$ and $N_G[y_i]{\setminus}\set{t}$\;
			$\boldsymbol{\algname{Gen{-}Seps}}(G_i,s,t,AX,\emptyset,Z,k)$\; \label{line:recursiveGy2} \tcp*[l]{$\kappa_{sAX,t}(G_i){\geq} \kappa_{sAX,ty_i}(G){>} \kappa_{sAX,t}(G)$ (Lem.~\ref{lem:sAtMinlSepLargerThanMinsAtSep_s} (1)). Also, see Lemma~\ref{lem:thirdPart}} 
		}
			
		}
		
		}
		
	\end{algocf}
	
		\eat{	
		{$\boldsymbol{\algname{Gen{-}Seps}}$}$(G, s, t, AX^*Y,Z,k)$ \label{line2e:thirdPartRecurseWithX} \tcp*[f]{By Corollary~\ref{corr:mainThmBoundedCardinalityExt}, $\closestMinSep{sX^*YA,t}(G)\in \safeSepskImp{s}{t}{\min}(G,AX^*Y)$, so progress will definitely be made in the next iteration (in lines~\ref{line2e:elseIfBegin}-\ref{line2e:secondPart1End}).}\;
		
		Let $Y=\set{y_1,\dots,y_\ell}$,  \label{line2e:thirdPartRest}\; 
		
		\ForNoEnd{$i=1$ to $i=\ell$}{	
			$Y_{i-1}\eqdef \set{y_1,\dots,y_{i-1}}$ \;
			
			{$\boldsymbol{\algname{Gen{-}Seps}}$}$(G\sminus y_i, s, t, AX^*Y_{i-1} ,Z\cup \set{y_i},k-1)$ \label{line2e:thirdPartSecondRecursiveCall} \;
			
			Let $G_i$ be the graph where $\edges(G_i)=\edges(G)\cup \set{(u,t):u\in N_G[y_i]}$ \tcp*[l]{See Lemma~\ref{lem:thirdPart}} \; 
			\hspace{-7pt}{$\boldsymbol{\algname{Gen{-}Seps}}$}$(G_i, s,t, AX^*Y_{i-1},Z,k)$ \label{line2e:thirdPartEnd} \tcp*[l]{$\kappa_{sAX^*Y_{i-1},t}(G_i){\geq} \kappa_{sA,ty_i}(G){>} \kappa_{sA,t}(G)$ (Lem.~\ref{lem:sAtMinlSepLargerThanMinsAtSep_s} (1)).} 
		}
	}

\section{\apprev{Source-component optimization}}
\label{sec:proofsFromNMWCUApplication}

\subsection{\apprev{Source-component optimization}}

\def\dominatedByCPImportant{
	Let $s,t\in \nodes(G)$, and let
	$A,B\subseteq \nodes(G){\setminus}\set{s,t}$. Suppose that the set 
	\[
	\set{S\in \safeSeps{s}{t}(G,A) :
		B\cap C_s(G\sminus S)= \emptyset}\neq \emptyset
	\]
	is nonempty, and let $T^* \in \arg\min_{|S|}\set{S\in \safeSeps{s}{t}(G,A) : B\cap C_s(G\sminus S)=\emptyset}$.
	Then there exists $T\in \safeSepsImp{s}{t}(G,A)$ such that
	$C_s(G\sminus T)\subseteq C_s(G\sminus T^*)$ and
	$|T|\le |T^*|$.
}
\begin{lemma}
	\label{lem:dominatedByCPImportant}
	\dominatedByCPImportant
\end{lemma}
\begin{proof}
	If $T^*\in \safeSepsImp{s}{t}(G,A)$, then the claim holds with $T=T^*$.
	Otherwise, by the definition of CP-importance, there exists
	$T_1\in \safeSeps{s}{t}(G,A)$ such that
	\[
	C_s(G\sminus T_1)\subsetneq C_s(G\sminus T^*)
	\qquad\text{and}\qquad
	|T_1|\le |T^*|.
	\]
	If $T_1$ is CP-important, we are done. Otherwise, we repeat the same
	argument. This gives a sequence $T^*=T_0,T_1,T_2,\ldots$ of separators in
	$\safeSeps{s}{t}(G,A)$ such that, for every $i$,
	\[
	C_s(G\sminus T_{i+1})\subsetneq C_s(G\sminus T_i)
	\qquad\text{and}\qquad
	|T_{i+1}|\le |T_i|.
	\]
	The sequence cannot continue indefinitely, because the sets
	$C_s(G\sminus T_i)$ form a strictly descending chain of subsets of
	$\nodes(G)$. Hence the process stops at some separator $T$ that is
	CP-important by definition, i.e., $T\in \safeSepsImp{s}{t}(G,A)$. By
	construction,
	\[
	C_s(G\sminus T)\subseteq C_s(G\sminus T^*)
	\qquad\text{and}\qquad
	|T|\le |T^*|.
	\]
\end{proof}

\begin{repcorollary}{\ref{cor:targetedCPIsolation}}
	\targetedCPIsolation
\end{repcorollary}
\begin{proof}
	Let $\mc{F}\eqdef \set{S\in \safeSepsk{s}{t}{k}(G,A) : B\cap C_s(G\sminus S)=\emptyset}$.
	These are exactly the feasible minimal $s,t$-separators of size at most $k$
	appearing in the statement of the corollary.
	
	Define the candidates retained by the algorithm as
\[
\mc Z=\{S\in\safeSepskImp{s}{t}{k}(G,A):B\cap C_s(G\sminus S)=\emptyset\}.
\]
We first show that $\mc{Z}=\emptyset$ if and only if $\mc{F}=\emptyset$.
	Since $\safeSepskImp{s}{t}{k}(G,A)\subseteq
	\safeSepsk{s}{t}{k}(G,A)$, we have $\mc{Z}\subseteq \mc{F}$.
	Conversely, if $\mc{F}\neq\emptyset$, let
	$T^*\in\arg\min_{|S|}\mc{F}$. Since $|T^*|\le k$, the choice of $T^*$ in
	$\mc{F}$ also makes it minimum-cardinality among all separators in
	$\safeSeps{s}{t}(G,A)$ satisfying
	$B\cap C_s(G\sminus S)=\emptyset$: any smaller such separator would also
	have size at most $k$ and hence would belong to $\mc{F}$.
	By Lemma~\ref{lem:dominatedByCPImportant}, there exists
	$T\in\safeSepsImp{s}{t}(G,A)$ such that
	$C_s(G\sminus T)\subseteq C_s(G\sminus T^*)$ and
	$|T|\le |T^*|$. Since $|T^*|\le k$, we have
	$T\in\safeSepskImp{s}{t}{k}(G,A)$. Moreover,
	$B\cap C_s(G\sminus T)=\emptyset$, because
	$B\cap C_s(G\sminus T^*)=\emptyset$ and
	$C_s(G\sminus T)\subseteq C_s(G\sminus T^*)$. Hence $T\in\mc{Z}$, so
	$\mc{Z}\neq\emptyset$.
	
	It remains to prove optimality. If $\mc{F}\neq\emptyset$, let
	$T^*\in\arg\min_{|S|}\mc{F}$. The same argument gives a separator
	$T\in\mc{Z}$ with $|T|\le |T^*|$. Since $\mc{Z}\subseteq\mc{F}$, no
	separator in $\mc{Z}$ has size smaller than $|T^*|$. Therefore the
	minimum-cardinality separator in $\mc{Z}$ has size exactly $|T^*|$, and is
	an optimum feasible separator. If $\mc{Z}=\emptyset$, then
	$\mc{F}=\emptyset$, and returning $\bot$ is correct.
	
	The running time is dominated by the enumeration of
	$\safeSepskImp{s}{t}{k}(G,A)$, which by
	Theorem~\ref{thm:singleExponentialSafeImportantOfSizek} is
	$O(2^{O(k^2\log k)}\cdot n\cdot T(n,m))$. \rev{Using vertex splitting and the deterministic almost-linear exact
	maximum-flow algorithm of van den Brand et al.~\cite{DBLP:journals/corr/abs-2309-16629},
	we may take $T(n,m)=(n+m)^{1+o(1)}$. The split graph has $O(n+m)$
	vertices and edges and polynomially bounded integral capacities. Thus the
	running time is $2^{O(k^2\log k)}n(n+m)^{1+o(1)}$.}
\end{proof}

\begin{applicationrevision}
\section{Proof of Theorem~\ref{thm:mixed-two-sets}}
\label{sec:mixed-two-sets-appendix}

Let $G$ be the input undirected graph, with disjoint nonempty terminal
classes $A_1,\ldots,A_r$ and budgets $k,\ell\geq0$.
Write $n=|\nodes(G)|$, $m=|\edges(G)|$, and $p=k+\ell$.
Let $W\subseteq\nodes(G)$ contain all terminals and any additional
vertices forbidden for deletion, and let
$E_{\rm forb}\subseteq\edges(G)$ contain the edges forbidden for
edge deletion. A solution is a pair $(X,F)$ with
\[
X\subseteq\nodes(G)\setminus W,\qquad
F\subseteq\edges(G)\setminus E_{\rm forb},\qquad
|X|\leq k,\qquad |F|\leq\ell,
\]
such that each $A_i$ lies in one component after these deletions, and
different classes lie in different components.
Deleting a vertex removes all its incident edges.
If a forbidden edge must survive even vertex deletions, put both its
endpoints in $W$ as well.

\paragraph*{Preprocessing and subdivision.}
If any class meets different components of $G$, reject: deletions cannot
connect them. For $r=1$, the empty solution now suffices.
For $p\leq1$, test the empty deletion and every permitted single-vertex
or single-edge deletion that respects both budgets. Return a feasible
solution if one exists, and otherwise reject. These tests take
$O((n+m+1)^2)$ time.
Henceforth $r\geq2$ and $p\geq2$; $G$ need not be connected.

For every edge $e=(u,v)$, introduce a distinct vertex $z_e$ and replace
$e$ by the path $u,z_e,v$.
Let $H$ be the resulting graph, and put
\[
Z=\{z_e:e\in\edges(G)\}.
\]
The initially undeletable vertices in $H$ are
\[
W_H=W\cup\{z_e:e\in E_{\rm forb}\}.
\]
Then $\nodes(H)=\nodes(G)\mathbin{\dot\cup} Z$,
$|\nodes(H)|\leq n+m$, and $|\edges(H)|\leq2m$.
The same construction permits parallel edges in $G$, using a separate
subdivision vertex for each edge; $H$ is still simple.
Loops, if supplied, can first be discarded, because their deletion
does not affect connectivity.

A pair $(X,F)$ corresponds to the vertex set
\[
D=X\cup\{z_e:e\in F\}
\]
in $H$. In particular, the two budgets become
\[
|D\cap\nodes(G)|\leq k,\qquad |D\cap Z|\leq\ell,
\qquad D\cap W_H=\emptyset.
\tag{*}\label{eq:mixed-two-sets-budgets}
\]
For any two original vertices outside $X$, connectivity after deleting
$X$ and $F$ from $G$ is equivalent to connectivity in $H\sminus D$.
Indeed, a path in the former graph becomes a path in the latter by
subdividing its edges.
Conversely, a simple path in $H\sminus D$ between original vertices
alternates between original and subdivision vertices.
Replacing every segment $u,z_e,v$ by $e=(u,v)$ gives a path using no
vertex of $X$ and no edge of $F$.
This argument also covers edges incident to a deleted original vertex:
the surviving subdivision vertex of such an edge cannot be used to
pass through the deleted endpoint.

\paragraph*{One family of connected sets for each class.}
Choose a representative $b_i\in A_i$ for every $i$.
For each class separately, form $H_i$ from $H$ by adding a new target
$t_i$ adjacent to every terminal outside $A_i$.
Apply Theorem~\ref{thm:singleExponentialSafeImportantOfSizek} to
$(H_i,b_i,t_i,A_i)$ with bound $p$, and keep the source components:
\[
\mathcal R_i=
\bigl\{C_{b_i}(H_i\sminus Q):
Q\in\safeSepskImp{b_i}{t_i}{p}(H_i,A_i)\bigr\}.
\]
If $b_i,t_i$ are already disconnected, the only minimal separator is
empty, so take $\mathcal R_i=\{C_{b_i}(H_i)\}$ directly.
Every $R_i\in\mathcal R_i$ contains $A_i$ and induces a connected
subgraph of $H$: it is connected in $H_i$ and does not contain the only
new vertex $t_i$.
Each family has at most $f(p)=2^{O(p^2\log p)}$ members.
The temporary targets are used only in these enumerations.

The following lemma is the reason we can choose the connected sets for
all classes at once. Its deletion set $D$ is the same for every class.

\begin{lemma}
\label{lem:mixed-two-sets-regions}
If the instance has a solution, there are pairwise disjoint sets
$R_i\in\mathcal R_i$, $1\leq i\leq r$, and a feasible deletion set
$D$ in $H$ with $D\cap\bigcup_{i=1}^rR_i=\emptyset$.
\end{lemma}

\begin{proof}
Start with a feasible deletion set in $H$, and remove vertices from it
while feasibility is preserved, obtaining an inclusion-minimal feasible
set $D$. Taking a subset preserves both budgets and avoids $W_H$.
Let $C_i$ be the component of $H\sminus D$ containing $A_i$, and put
$S_i=N_H(C_i)$, the vertices outside $C_i$ with a neighbor in $C_i$.
Since $C_i$ is a component after deleting $D$, we have $S_i\subseteq D$.

We first show that every $x\in S_i$ has a neighbor in some $C_j$ with
$j\ne i$. Restoring $x$ merges precisely the components of $H\sminus D$
that have a neighbor of $x$. One of these is $C_i$.
If none contained another terminal class, this would merge $C_i$ only
with components containing no terminals. Every class would remain
connected and different classes would remain separated, contradicting
the inclusion-minimality of $D$.

It follows that $S_i$ is an ordinary minimal $b_i,t_i$-separator in
$H_i$. Indeed, $t_i$ has no neighbor in $C_i$, and deleting $S_i$
leaves exactly $C_i$ reachable from $b_i$. For minimality, restore any
$x\in S_i$. There is now a path from $b_i$ through $C_i$, then $x$,
then a neighboring component $C_j$ to a terminal of $A_j$, and finally
to $t_i$. All vertices of $C_i$ and $C_j$ avoid $S_i\subseteq D$,
so this path uses no other vertex of $S_i$.
Also $A_i\subseteq C_i$ and $|S_i|\leq|D|\leq p$.
Thus $S_i$ is a CP minimal separator for this enumeration.

Starting from $S_i$, repeatedly replace the separator by a CP minimal
separator of no greater size whose source component is properly contained
in the current one,
whenever such a replacement exists.
The component strictly decreases, so this process ends at a CP-important
separator $Q_i$ with
\[
A_i\subseteq R_i=C_{b_i}(H_i\sminus Q_i)\subseteq C_i,
\qquad |Q_i|\leq|S_i|\leq p.
\]
Consequently $R_i$ occurs in $\mathcal R_i$.
Do this independently for each $i$, always starting from the boundary
$S_i$ of the component of the original $H\sminus D$.
The resulting sets are disjoint and avoid the same $D$, because
$R_i\subseteq C_i$ and the $C_i$ are distinct components of $H\sminus D$.
\end{proof}

The algorithm tries every tuple $(R_1,\ldots,R_r)$ and rejects tuples
whose sets intersect. For each remaining tuple, work in the original
subdivided graph $H$ and forbid deleting any vertex of
\[
P=W_H\cup\bigcup_{i=1}^r R_i.
\]
It remains to separate $b_1,\ldots,b_r$ pairwise within the two budgets.
Keeping each connected subgraph $H[R_i]$ intact already connects all of
$A_i$ to $b_i$, so separating the representatives suffices.
The lemma guarantees a tuple for which the original feasible $D$ is
still allowed.

The CP calls use only the total bound $p$; their separators $Q_i$ need
not respect the two budgets separately. This causes no problem:
we protect their source components, but never delete the $Q_i$ or charge
them to either budget. Protection includes subdivision vertices in each
$R_i$. Every path in $H[R_i]$ between original terminals survives and
projects to an original path, by the subdivision correspondence above.

\paragraph*{Encoding the two budgets by arc deletions.}
Fix a disjoint tuple and its protected set $P$.
For each vertex $v$ of $H$, create vertices $v_{\rm in},v_{\rm out}$
in a directed graph $J$. Add exactly two types of arcs:
\begin{enumerate}
\item For every $v$, add the \emph{split arc}
$a_v=(v_{\rm in},v_{\rm out})$. Deleting this arc represents deleting $v$.
\item For every edge $\{u,v\}$ of $H$, add the \emph{adjacency arcs}
$(u_{\rm out},v_{\rm in})$ and $(v_{\rm out},u_{\rm in})$.
These permit traversal of the edge in either direction.
\end{enumerate}
Only arcs in the following two disjoint classes may be deleted:
\[
E_1=\{a_v:v\in\nodes(G)\setminus P\},\qquad
E_2=\{a_z:z\in Z\setminus P\},
\]
with budgets $k$ and $\ell$, respectively.
Thus every adjacency arc and every protected split arc $a_v$, $v\in P$,
is forbidden for deletion. The directed problem deletes arcs only;
vertex protection in $H$ is encoded entirely by excluding split arcs
from the allowed classes. In particular, deleting an original edge $e$
is represented by deleting $a_{z_e}$, not by deleting an adjacency arc.

For any allowed arc set $L\subseteq E_1\cup E_2$, let
\[
D_L=\{v:a_v\in L\}.
\]
Then $D_L\cap P=\emptyset$, and
\[
|D_L\cap\nodes(G)|=|L\cap E_1|,
\qquad |D_L\cap Z|=|L\cap E_2|.
\]
For distinct representatives $b_h,b_j$, a path between them in
$H\sminus D_L$ gives a directed path from $(b_h)_{\rm out}$ to
$(b_j)_{\rm in}$ in $J\sminus L$: use one adjacency arc per edge and
the surviving split arc at each internal vertex.
Conversely, such a directed path alternates adjacency arcs and split
arcs. Passing through an internal vertex $v$ requires $a_v$, so the
projected walk uses no vertex of $D_L$. Its endpoints also avoid $D_L$
because the representatives belong to $P$.
The walk contains a path in $H\sminus D_L$. Thus
\[
b_h,b_j\text{ are connected in }H\sminus D_L
\quad\Longleftrightarrow\quad
(b_j)_{\rm in}\text{ is reachable from }(b_h)_{\rm out}
\text{ in }J\sminus L.
\]
Every allowed vertex-deletion set $D$ maps back to
$L_D=\{a_v:v\in D\}$, with exactly the same two costs.

\paragraph*{Separating all representatives.}
We now give the directed algorithm, including its dependence on $r$.
A \emph{skew cut} for ordered pairs
$(\sigma_1,\tau_1),\ldots,(\sigma_q,\tau_q)$ deletes arcs so that
there is no directed $\sigma_i,\tau_j$-path whenever $i\geq j$.
Its multi-budgeted version allows only arcs in the specified classes
and enforces each class's budget~\cite[Section~4.1]{MultiBudgetedCuts2020}.
For our graph $J$, set $q=r-1$ and
\[
\sigma_i=(b_{i+1})_{\rm out},\qquad
\tau_i=(b_i)_{\rm in},\qquad 1\leq i\leq q.
\]
The skew requirement for $i\geq j$ forbids a connection from
$b_{i+1}$ to $b_j$. These are exactly all pairs of representatives
with larger index first: a pair $b_h,b_j$ with $h>j$ corresponds to
$i=h-1\geq j$. Since connectivity in $H$ is undirected, the path
equivalence above proves that these requirements separate every pair.
All directed sources and targets are distinct split copies.

We use multi-budgeted important cuts as defined by
Kratsch et al.~\cite[Definition~1]{MultiBudgetedCuts2020}.
These are permitted minimal cuts for which no other permitted minimal
cut makes at least the same vertices reachable from the source while
using no more arcs from either budget class.
This maximizes the source side, whereas our earlier CP step minimizes it.
Their Theorem~7 enumerates, in $2^{O(b^2\log b)}$ times the graph
size, a family of at most $2^{O(b^2\log b)}$ permitted minimal cuts
respecting budgets of total $b$, containing every multi-budgeted important
cut for a specified source and target set.
Their Section~4.1, Lemma~2 states that every feasible skew instance has
a feasible solution containing such an important cut from its last source
$\sigma_q$ to all remaining targets $\{\tau_1,\ldots,\tau_q\}$.

Apply these results by the following recursion. Each call carries the
remaining graph, the remaining ordered pairs, and the remaining budgets.
\begin{enumerate}
\item If no pairs remain, return the empty additional cut.
\item If $\sigma_q$ reaches none of $\tau_1,\ldots,\tau_q$, drop the
last pair and recurse without spending either budget.
\item Otherwise enumerate the family from Theorem~7 for source
$\sigma_q$ and target set $\{\tau_1,\ldots,\tau_q\}$.
For each returned cut $C$, delete its arcs, subtract
$|C\cap E_1|$ and $|C\cap E_2|$ from the respective budgets,
and recurse on the first $q-1$ pairs.
\item Return $C$ together with a successful recursive answer if one
exists; if every branch fails, reject.
\end{enumerate}
Dropping a pair removes only its requirements: all graph vertices remain.
Restrict each arc class to the surviving arcs at every call.
The cited enumeration takes positive budgets, so whenever a budget is
zero, forbid deletion of its remaining arcs and omit that class from the call.
If no deletable arcs remain, accept with the empty additional cut exactly
when all remaining skew requirements already hold; otherwise reject.
These rules also cover $k=0$ or $\ell=0$ in the original instance.

For soundness, the selected cut $C$ permanently separates $\sigma_q$
from every remaining target. Subsequent arc deletions cannot undo that
separation. The other requirements are precisely those of the first
$q-1$ pairs: no earlier source must be separated from $\tau_q$, because
its index is less than $q$. Thus a successful recursive answer, together
with $C$, solves the current instance. The deleted arc sets are disjoint
across calls, and their costs are subtracted from both budgets exactly.
The same argument justifies dropping an already separated last source.

For completeness, suppose the current instance is feasible.
The cited Lemma~2 gives a feasible solution $L^*$ containing an important
last-source cut $C^*$, which appears in the enumerated family.
In that branch, $L^*\setminus C^*$ satisfies the remaining requirements
and fits both reduced budgets. Induction on the number of pairs gives
a successful recursive branch. Extra cuts in the enumerated family are
harmless: all are permitted cuts respecting the current budgets, and
soundness applies to every branch.

Each branching step drops one pair and deletes a nonempty cut.
The cut is nonempty because branching occurs only when a required path
exists. Since each deletable arc belongs to exactly one class, such a
step also decreases the total remaining budget by at least one.
There are therefore at most $\min\{p,r-1\}$ branching steps on a path
of the recursion tree, each with $2^{O(p^2\log p)}$ choices.
Nonbranching calls merely remove pairs. Since $r\leq n$ and $J$ has
$O(n+m)$ vertices and arcs, this gives deterministic time
\[
2^{O(\min\{p,r-1\}\,p^2\log p)}(n+m+1)^{O(1)}
\]
for the residual problem. This is the dependence on the number of pairs
in the recursion above; we do not use the published cubic-in-$p$
exponent for unrestricted skew instances.

\paragraph*{Correctness and total running time.}
For each disjoint tuple, run the directed algorithm and, if it returns
$L$, recover
\[
X=D_L\cap\nodes(G),\qquad
F=\{e\in\edges(G):z_e\in D_L\}.
\]
The cost equalities give $|X|\leq k$ and $|F|\leq\ell$.
Avoiding $P\supseteq W_H$ respects all deletion prohibitions and preserves
each connected $H[R_i]$. The skew constraints separate the representatives,
so the preserved sets connect each class and keep different classes apart.
The subdivision correspondence transfers these properties to $G$.
Thus every returned pair is a valid mixed solution.
If an inclusion-minimal solution is required, repeatedly restore a selected
vertex or edge whenever feasibility is preserved. This takes polynomial
time and never increases either budget.
Conversely, Lemma~\ref{lem:mixed-two-sets-regions} supplies, for every
feasible instance, a tuple whose protected sets avoid one feasible $D$.
Its split arcs form a feasible skew cut, so the directed algorithm
succeeds for that tuple. If all tuples fail, the instance is infeasible.

Each $H_i$ has $O(n+m)$ vertices and edges: its new target has at most
$n$ neighbors. The $r$ CP enumerations therefore take
$r\,2^{O(p^2\log p)}(n+m+1)^{O(1)}$ time and produce at most
\[
\prod_{i=1}^r|\mathcal R_i|
\leq f(p)^r=2^{O(rp^2\log p)}
\]
tuples. Multiplying by the residual bound, and using
$\min\{p,r-1\}\leq r$, gives total time
\[
2^{O(rp^2\log p)}(n+m+1)^{O(1)}.
\]
The hidden constants and the polynomial degree are independent of $r$.
In particular, every fixed number of classes gives the stated
$2^{O(p^2\log p)}$ parameter dependence.

\section{Annotated Odd Cycle Transversal with terminal connectivity}
\label{sec:annotated-oct-appendix}

We give the full proof of Theorem~\ref{thm:annotated-oct-connectivity}
and derive a comparison algorithm using treewidth reduction.
The input consists of a finite simple undirected graph $G$, a nonempty
set $A\subseteq\nodes(G)$, a map $\phi:A\to\{0,1\}$, and an integer
$k\geq 0$. We seek $X\subseteq\nodes(G)\setminus A$ with $|X|\leq k$
such that $G\sminus X$ is bipartite, all vertices of $A$ belong to one
component of $G\sminus X$, and there is a proper coloring
$\chi:\nodes(G)\setminus X\to\{0,1\}$ satisfying
$\chi(a)=\phi(a)$ for every $a\in A$.
Here a coloring is proper if the endpoints of every edge have different
colors. In particular, $G[A]$ need not be connected, and the size of $A$
is not a parameter.

\subsection{The algorithm using CP-important separators}

\begin{reptheorem}{\ref{thm:annotated-oct-connectivity}}
Annotated Odd Cycle Transversal with terminal connectivity can be solved
deterministically in
$2^{O(k^2\log k)}(n+m+1)^{O(1)}$ time, where
$n=|\nodes(G)|$ and $m=|\edges(G)|$.
\end{reptheorem}

\begin{proof}
The algorithm will find a connected set $U$ containing $A$ whose coloring
agrees with $\phi$, then solve ordinary OCT while protecting $U$ from
deletion. The main task is to generate candidates for $U$ so that, if a
solution exists, some candidate is disjoint from a feasible solution.
For $k=0$, return $X=\emptyset$ if it is feasible, and reject otherwise.
If $A$ meets different components of $G$, reject. Retain the other
components: their odd cycles must also be destroyed.

\paragraph*{Two copies of a bipartite component.}
Construct $H$ with vertices $v^0,v^1$ for each $v\in\nodes(G)$, and
edges $u^0v^1,u^1v^0$ for each edge $uv$ of $G$.
Write $\overline{v^i}=v^{1-i}$. Thus every step along an edge changes
the copy index, and exchanging all copies preserves edges.
Choose $a_0\in A$ and set
\[
s=a_0^{\phi(a_0)},\qquad t=a_0^{1-\phi(a_0)},\qquad
A_\phi=\{a^{\phi(a)}:a\in A\}.
\]
To see which separator we need, suppose a solution exists. Choose an
inclusion-minimal feasible $X$, let $C$ be the vertex set of the component
of $G\sminus X$ containing $A$, and fix a proper coloring $\chi$
of $G\sminus X$ with
$\chi(a)=\phi(a)$ for all $a\in A$. Define
\[
D=\{v^{\chi(v)}:v\in C\},\qquad
\overline D=\{v^{1-\chi(v)}:v\in C\}.
\]
These sets choose opposite copies of each vertex, so they are disjoint.
For an edge $uv$ of $G[C]$, properness gives $\chi(u)\ne\chi(v)$;
its two edges in $H$ therefore join the chosen copies within $D$ and
within $\overline D$, respectively. Conversely, every edge between
copies comes from an edge of $G$. Thus each induced subgraph has
exactly the vertices and edges of $G[C]$, with vertices renamed.
In particular, every path in $G[C]$ has a copy in each set, so both
sets are connected. Also, $A_\phi\subseteq D$, $s\in D$, and
$t\in\overline D$.

\paragraph*{Why restoring any vertex of $X$ creates an odd cycle.}
The old coloring $\chi$ agrees with $\phi$ on $A$.
Suppose restoring $x\in X$ leaves a bipartite graph. Choose a
bipartition of each component of this enlarged graph. In the component
containing $C$, give the side containing $a_0$ color $\chi(a_0)$ and
the other side color $1-\chi(a_0)$; color other components either way.
This colors every vertex, including $x$, with different colors across
every edge. On connected $C$, it reproduces $\chi$: both colorings
start with $\chi(a_0)$ and alternate along paths from $a_0$. Hence every
$a\in A$ retains color $\chi(a)=\phi(a)$.
Restoring $x$ also preserves the paths connecting $A$, so
$X\setminus\{x\}$ would be feasible, contradicting minimality.
Thus $X$ is an inclusion-minimal ordinary OCT: restoring any $x\in X$
creates an odd cycle through $x$, since $G\sminus X$ is bipartite.

\paragraph*{A minimal separator of size at most $2k$.}
Let $S=N_H(D)$, the vertices outside $D$ with a neighbor in $D$.
Every vertex of $S$ is a copy of a vertex of $X$: edges within $C$
stay within the two copies just described, and any neighbor of $C$
outside $C$ belongs to $X$, since $C$ is a component of $G\sminus X$.
Consequently, $|S|\leq 2|X|\leq 2k$.
Deleting $S$ isolates the connected set $D$, so
$C_s(H\sminus S)=D$. It leaves $\overline D$ intact as well.
Thus $S$ separates $s$ from $t$ and preserves $A_\phi$ in the
component of $s$.

We must also prove ordinary minimality. It suffices to show that
restoring any $y\in S$ joins $D$ to $\overline D$, since these
connected sets contain $s$ and $t$. Write $y=x^i$ and choose a
neighbor $z\in D$ of $y$.
Figure~\ref{fig:oct-boundary} illustrates the two cases below: follow
the thick route in each panel to see how restoring $y$ reconnects
$D$ to $\overline D$.
If $\overline y\in S$ (Figure~\ref{fig:oct-boundary}(a)), choose a neighbor $w\in D$ of $\overline y$.
Exchanging copies in the edge $\overline yw$ gives the edge
$y\overline w$. Hence $z,y,\overline w$ is a path from $D$ to
$\overline D$ after restoring $y$; it does not use $\overline y$.

If $\overline y\notin S$ (Figure~\ref{fig:oct-boundary}(b)), traverse the odd cycle through $x$
established above, starting at $y$ in $H$.
An odd number of steps changes the copy index, so this gives a path
from $y$ to $\overline y$. Its internal vertices are copies of
vertices outside $X$, hence avoid $S$. Its endpoint $\overline y$
also avoids $S$ by this case assumption.
The edge $yz$ gives the edge $\overline y\overline z$.
We can therefore go from $z\in D$ to $y$, along the lifted cycle
to $\overline y$, and then to $\overline z\in\overline D$.
This walk survives when only $y$ is restored.
Both cases give an $s,t$-path in $H\sminus(S\setminus\{y\})$.
Thus $S$ is a CP minimal separator of size at most $2k$.

\begin{figure}[htbp]
\centering
\resizebox{\linewidth}{!}{
\begingroup
\colorlet{octink}{.}
\begin{tikzpicture}[
  x=1cm,y=1cm,font=\small,text=octink,
  vertex/.style={circle,draw=octink,fill=white,minimum size=6mm,
    inner sep=1pt,line width=.55pt},
  restored/.style={vertex,fill=octink!18,line width=1.2pt},
  deleted/.style={vertex,draw=gray!65,text=gray!80},
  edge/.style={draw=octink!45,line width=.55pt},
  missing/.style={draw=gray!65,densely dashed,line width=.55pt},
  route/.style={draw=octink,line width=1.5pt},
  region/.style={draw=octink!50,fill=octink!5,rounded corners=5pt},
  boundary/.style={draw=octink,densely dotted,rounded corners=5pt},
  heading/.style={anchor=west,font=\small\bfseries},
  note/.style={font=\scriptsize,align=center}
]
\node[heading] at (0,.65) {(a) Both copies belong to $S$: $\overline y\in S$.};
\node[note] at (1.25,.08) {Original graph $G$; $X=\{x\}$};
\node[note] at (7.8,.08) {Double cover $H$: restore only $y=x^1$};

\draw[region] (-.05,-2.68) rectangle (2.55,-1.83);
\node[note] at (1.25,-2.95) {$C=\{u,v\}$};
\node[deleted] (ax) at (1.25,-.7) {$x$};
\node[vertex] (au) at (.45,-2.25) {$u$};
\node[vertex] (av) at (2.05,-2.25) {$v$};
\draw[missing] (au)--(ax)--(av);
\draw[edge] (au)--(av);
\draw[gray!80,line width=.8pt] (1.04,-.91)--(1.46,-.49)
  (1.04,-.49)--(1.46,-.91);
\draw[->,draw=octink,line width=.6pt] (2.9,-1.55)--(3.3,-1.55);

\draw[region] (3.5,-2.7) rectangle (5.6,-.55);
\draw[region] (10.0,-2.7) rectangle (11.6,-.55);
\node at (4.8,-.3) {$D$};
\node at (10.8,-.3) {$\overline D$};
\draw[boundary] (6.65,-2.85) rectangle (8.95,-.55);
\node[note] at (7.8,-3.08) {$S=\{x^0,x^1\}$};
\node[vertex] (au0) at (4.8,-1.0) {$u^0$};
\node[vertex] (av1) at (4.8,-2.25) {$v^1$};
\node[vertex] (av0) at (10.8,-1.0) {$v^0$};
\node[vertex] (au1) at (10.8,-2.25) {$u^1$};
\node[restored] (ax1) at (7.8,-1.0) {$x^1$};
\node[deleted] (ax0) at (7.8,-2.25) {$x^0$};
\draw[edge] (au0)--(av1) (av0)--(au1);
\draw[missing] (av1)--(ax0)--(au1);
\draw[route] (au0)--(ax1)--(av0);
\node[note,anchor=east] at (4.43,-1.0) {$s=z$};
\node[note,anchor=west] at (11.17,-1.0) {$\overline w$};
\node[note,anchor=east] at (4.43,-2.25) {$w$};
\node[note,anchor=west] at (11.17,-2.25) {$t$};
\node[note] at (7.8,-1.53) {$y$: restored};
\node[note] at (7.8,-2.66) {$\overline y$: still deleted};
\draw[gray!80,line width=.8pt] (7.59,-2.46)--(8.01,-2.04)
  (7.59,-2.04)--(8.01,-2.46);
\node[note] at (7.8,-3.45)
  {The thick path $z,y,\overline w$ joins $D$ to $\overline D$.};

\begin{scope}[yshift=-4.7cm]
\node[heading] at (0,.65) {(b) Only one copy belongs to $S$: $\overline y\notin S$.};
\node[note] at (1.25,.08) {Original graph $G$; $X=\{x\}$};
\node[note] at (7.8,.08) {Double cover $H$: restore only $y=x^1$};

\draw[region] (-.13,-1.57) rectangle (.85,-.62);
\node[vertex] (ba) at (.36,-1.1) {$a$};
\node[note] at (.36,-1.85) {$C=\{a\}$};
\node[deleted] (bx) at (1.75,-1.1) {$x$};
\node[vertex] (bu) at (1.05,-2.5) {$u$};
\node[vertex] (bv) at (2.5,-2.5) {$v$};
\draw[missing] (ba)--(bx)--(bu) (bx)--(bv);
\draw[edge] (bu)--(bv);
\draw[gray!80,line width=.8pt] (1.54,-1.31)--(1.96,-.89)
  (1.54,-.89)--(1.96,-1.31);
\draw[->,draw=octink,line width=.6pt] (2.9,-1.55)--(3.65,-1.55);

\draw[region] (3.9,-2.43) rectangle (5.08,-1.2);
\draw[region] (10.95,-2.43) rectangle (12.13,-1.2);
\node at (4.49,-.9) {$D$};
\node at (11.54,-.9) {$\overline D$};
\node[vertex] (ba0) at (4.49,-1.7) {$a^0$};
\node[vertex] (ba1) at (11.54,-1.7) {$a^1$};
\node[restored] (bx1) at (5.95,-1.7) {$x^1$};
\node[vertex] (bx0) at (10.08,-1.7) {$x^0$};
\node[vertex] (bu0) at (7.22,-.8) {$u^0$};
\node[vertex] (bv1) at (8.82,-.8) {$v^1$};
\node[vertex] (bv0) at (7.22,-2.6) {$v^0$};
\node[vertex] (bu1) at (8.82,-2.6) {$u^1$};
\draw[edge] (bx1)--(bv0)--(bu1)--(bx0);
\draw[route] (ba0)--(bx1)--(bu0)--(bv1)--(bx0)--(ba1);
\draw[boundary] (5.36,-2.25) rectangle (6.54,-1.17);
\node[note] at (4.49,-2.2) {$s=z$};
\node[note] at (11.54,-2.2) {$t=\overline z$};
\node[note] at (5.95,-2.59) {$S=\{y\}$\\$y$: restored};
\node[note] at (10.15,-2.58) {$\overline y\notin S$\\not deleted};
\node[note] at (8.02,-.23) {odd cycle in $G$ $\longmapsto$ three-edge path in $H$};
\node[note] at (8.02,-3.22)
  {The path $x^1,u^0,v^1,x^0$ uses no other vertex of $S$.};
\end{scope}

\node[note,anchor=west] at (0,-8.42)
  {Shaded regions: $C$ in $G$; $D$ and $\overline D$ in $H$. Dotted boxes: $S$. Crossed vertices: deleted.};
\node[note,anchor=west] at (0,-8.82)
  {Thick edges show a surviving route after restoring $y$; pale dashed edges touch a deleted vertex.};
\end{tikzpicture}
\endgroup}
\caption{\apprev{Two examples with $X=\{x\}$ and $k=1$.
The right-hand drawings show the complete double covers, with $y=x^1$
restored and the other vertices of $S$ still deleted.
In (a), $A=\{u\}$, $\phi(u)=0$, and $S$ contains both copies of $x$,
so $|S|=2|X|=2$.
In (b), $A=\{a\}$, $\phi(a)=0$, and $S=\{x^1\}$: the other copy
$x^0$ is outside $S$ and survives throughout. The thick routes connect $D$ to $\overline D$,
and therefore reconnect $s$ and $t$.
Pale dashed edges indicate edges removed by the displayed deletions.}}
\label{fig:oct-boundary}
\end{figure}

\paragraph*{Generating and testing the candidates.}
First check that $A_\phi$ lies in the component of $s$ in $H$;
otherwise the connected set $D$ above cannot exist, so reject.
If $s,t$ are already disconnected, use $Q=\emptyset$ as the sole
candidate: it is the only ordinary minimal separator. Otherwise
enumerate the CP-important $s,t$-separators
of size at most $2k$ with respect to $A_\phi$ in their component;
vertices outside it cannot belong to a minimal separator.
For each candidate $Q$, set
\[
R=C_s(H\sminus Q),\qquad
U=\{v\in\nodes(G):v^0\in R\text{ or }v^1\in R\}.
\]
Then $A\subseteq U$, and $G[U]$ is connected because paths in $H[R]$
project to walks in $G[U]$.
Check that $G[U]$ is bipartite and has a proper coloring agreeing
with $\phi$: color it starting with color $\phi(a_0)$ at $a_0$,
then check every terminal. Discard $U$ if this fails.
Otherwise solve ordinary OCT in $G$ with budget $k$, forbidding
deletion of $U$. Return any successful solution; reject if all calls
fail. The separator $Q$ only selects $U$ and consumes none of the
deletion budget.

\paragraph*{Implementing the protected OCT call.}
We encode protection by giving each $u\in U$ $k+1$ copies:
a budget of $k$ deletions cannot remove them all.
In the new graph, every vertex outside $U$ has just one copy.
There are \emph{no edges between copies of the same vertex}.
Copies of distinct vertices $v,w$ are adjacent exactly when $vw$
is an edge of $G$.
Thus an original edge joins every copy in one group to every copy
in the other.
We show that an ordinary OCT of size at most $k$ in this graph
exists exactly when $G$ has one avoiding $U$.

First, let $X\subseteq\nodes(G)\setminus U$ be an OCT of size at most $k$.
Delete the single copy of each vertex of $X$.
Since $G\sminus X$ is bipartite, its vertices have two sides with
no edge within either side. Put every surviving copy on the side
of its original vertex. Every edge between copies comes from an
original edge, so its endpoints still lie on opposite sides.
There are no edges within a group of copies. Hence the new graph
is bipartite after these same $|X|$ deletions.

Conversely, let $Y$ be a set of at most $k$ vertices \emph{in total}
in the new graph, whose deletion makes that graph bipartite.
The set $Y$ may include copies of vertices in $U$.
However, it cannot contain all $k+1$ copies of any one $u\in U$,
so every protected vertex still has a surviving copy.
Let $X'$ consist of the vertices outside $U$ whose single copies
belong to $Y$. Thus $X'\cap U=\emptyset$ and $|X'|\leq |Y|\leq k$.

To see that $G\sminus X'$ is bipartite, select exactly one surviving
copy of each $u\in U$, and the single surviving copy of every vertex
outside $U\cup X'$. Consider the subgraph induced by these selected
vertices. We do not change $Y$ or restore any deleted vertex;
we only examine a subgraph of the already bipartite graph.
Its vertices inherit the same two sides, so it is bipartite too.
It has exactly one copy of every vertex of $G\sminus X'$.
By construction, two selected copies are adjacent exactly when
their originals are adjacent. Renaming the selected copies by their
original vertices therefore gives exactly $G\sminus X'$.
Hence $X'$ is the required OCT avoiding $U$.
This argument needs no minimality assumption on $Y$: the call asks
for a vertex set whose deletion makes the graph bipartite, not for
an $s,t$-separator.
The new graph has at most $(k+1)n$ vertices and $(k+1)^2m$ edges;
the deterministic OCT algorithm~\cite{OCTLP2014} therefore implements
each call in $2^{O(k)}(n+m+1)^{O(1)}$ time.

\paragraph*{Correctness.}
Any returned OCT leaves $G[U]$ intact, so all vertices of $A$ remain
in one component.
Color the surviving graph, exchanging colors on the component of $U$
if necessary to agree with $\phi(a_0)$ at $a_0$.
The same path argument used above forces agreement throughout $U$
with its checked coloring, and hence with $\phi$ on $A$.
The returned solution is therefore feasible.

Conversely, let $X,D,S$ be obtained from a feasible solution as above.
If $S$ is not CP-important, replace it by a CP minimal separator of
no greater size whose source component is a strict subset of the
current one.
Repeat until this is impossible; each replacement strictly decreases
the number of vertices in the source component. The resulting
CP-important separator $Q$ satisfies
\[
|Q|\leq 2k,\qquad A_\phi\subseteq C_s(H\sminus Q)\subseteq D.
\]
The algorithm considers this $Q$. Its projected set $U$ lies in $C$,
so $\chi$ restricted to $U$ passes the coloring test and $X$ avoids $U$.
Thus its protected OCT call has the solution $X$ and succeeds.

\paragraph*{Running time.}
The graph $H$ has $2n$ vertices and $2m$ edges.
Theorem~\ref{thm:singleExponentialSafeImportantOfSizek}, with budget
$2k$, generates at most $2^{O(k^2\log k)}$ candidates in
$2^{O(k^2\log k)}(n+m+1)^{O(1)}$ time.
Each candidate needs polynomial-time graph searches and one
$2^{O(k)}(n+m+1)^{O(1)}$ OCT call, giving the claimed total time.
\end{proof}

\subsection{A comparison algorithm from treewidth reduction}
\label{subsec:treewidthReduction}
The following proposition is a consequence of existing treewidth-reduction
methods. We include its derivation to distinguish the improved parameter
dependence above from fixed-parameter tractability itself.
In fact, this comparison algorithm also applies when no colors are
prescribed on $A$.

\begin{proposition}
\label{prop:annotated-oct-treewidth}
Annotated OCT with terminal connectivity can be solved in
$2^{2^{O(k^2)}}(n+m+1)^{O(1)}$ time by combining treewidth reduction with
dynamic programming. The same bound holds if the coloring $\phi$ is
omitted and only bipartiteness and connectivity of $A$ are required.
\end{proposition}

\begin{proof}
We use the small-separator and bipartization constructions of
Marx, O'Sullivan, and Razgon~\cite{DBLP:journals/talg/MarxOR13}.
For a graph $F$ and a vertex set $Z$, the graph
$\operatorname{torso}(F,Z)$ has vertex set $Z$; two distinct vertices
are adjacent if they are adjacent in $F$ or joined by a path whose
internal vertices lie outside $Z$.
A tree decomposition consists of a tree of vertex sets, called bags, such
that every edge has both endpoints in some bag and the bags containing
each vertex form a nonempty connected subtree.
Its width is the maximum bag size minus one; the treewidth is the minimum
width over all such decompositions.
Their treewidth-reduction lemma computes a set containing every
ordinary minimal $s,t$-separator of size at most $k$, such that the
torso on this set has treewidth $2^{O(k^2)}$.
This follows from Lemma~2.11 and the bound in Remark~2.14
of~\cite{DBLP:journals/talg/MarxOR13}.
The construction takes $2^{k^{O(1)}}(n+m+1)^{O(1)}$ time: its recursion
has depth at most $k$, with $2^{O(k)}$ boundary choices per layer, and
the total sizes of the layers at a call are linear in that call's
input size~\cite[proof of Lemma~2.11]{DBLP:journals/talg/MarxOR13}.

We may again handle $k=0$ directly. Use an ordinary OCT
algorithm~\cite{OCTLP2014} to find $W\subseteq\nodes(G)$ with
$|W|\leq k$ and $G\sminus W$ bipartite, or reject if no such $W$
exists. Write $B=G\sminus W$ and fix a proper coloring
$\beta:\nodes(B)\to\{0,1\}$.
There are at most $3^{|W|}$ ways to assign each vertex of $W$ one of
three roles: deleted, retained with color $0$, or retained with color
$1$. For a fixed assignment let $D\subseteq W$ be its deleted
vertices and let $\gamma:W\setminus D\to\{0,1\}$ be its coloring.
Discard the assignment if $\gamma$ is not proper on $G[W\setminus D]$.

For $i\in\{0,1\}$ define
\[
F_i=\{v\in\nodes(B):
\text{some }w\in W\setminus D\text{ satisfies }
vw\in\edges(G)\text{ and }
\beta(v)\oplus 1\oplus\gamma(w)=i\},
\]
where $\oplus$ denotes addition modulo two.
The two sets may overlap.
For $Y\subseteq\nodes(B)$, the coloring $\gamma$ extends to a
proper coloring of $G\sminus(D\cup Y)$ if and only if no component
of $B\sminus Y$ meets both $F_0\setminus Y$ and $F_1\setminus Y$.
Indeed, on each connected component of $B\sminus Y$, a proper
coloring either agrees with $\beta$ or exchanges its two colors.
An edge to a retained vertex of $W$ prescribes exactly the orientation
recorded in $F_i$. This is the bipartization-to-separation reduction
of~\cite[Lemma~4.1]{DBLP:journals/talg/MarxOR13}.

Add two auxiliary vertices adjacent to $F_0$ and $F_1$, respectively.
In the resulting graph, compute a set $Z\subseteq\nodes(B)$ containing
all minimal separators between the auxiliary vertices of size at most
$k$, with torso treewidth at most $2^{O(k^2)}$.
If the minimum separator has size greater than $k$, skip this assignment.
If the auxiliary vertices are already separated, use $Z=\emptyset$,
since the empty set is their only inclusion-minimal separator.
The auxiliary vertices themselves are excluded from $Z$.
Because $B$ is a subgraph of the auxiliary graph,
$\operatorname{torso}(B,Z)$ also has treewidth at most $2^{O(k^2)}$.
Let
\[
C=Z\cup W.
\]
The graph $\operatorname{torso}(G,C)$ has treewidth at most
$2^{O(k^2)}+|W|$: its restriction to $Z$ is
$\operatorname{torso}(B,Z)$, and adding all vertices of $W$ to
every bag of a tree decomposition covers the remaining edges.

At least one assignment produces a set $C$ containing a feasible
solution whenever the input has one.
To see this, choose an inclusion-minimal feasible $X$.
As proved above, $X$ is an inclusion-minimal ordinary OCT; the same
argument holds when $\phi$ is omitted.
Choose the assignment whose deleted set is $D=X\cap W$ and whose
coloring agrees with a proper coloring of $G\sminus X$.
Then $Y=X\setminus W$ separates $F_0$ from $F_1$ in $B$.
Restoring any $y\in Y$ creates an odd cycle in
$G\sminus(X\setminus\{y\})$, so $\gamma$ cannot extend after that
restoration. By the preceding equivalence, $Y$ is an ordinary minimal
separator in the auxiliary graph. Hence $Y\subseteq Z$ and
$X\subseteq C$.
For each assignment we may therefore solve the original problem with
deletions restricted to its set $C$.
The final solution need not respect the assignment on $W$: requiring
the original conditions directly guarantees soundness, and the
assignment is used only to obtain a suitable $C$.

\paragraph*{Preserving parity and terminal connectivity.}
Every component $K$ of $G\sminus C$ is bipartite, since $W\subseteq C$.
Fix its bipartition $(K_0,K_1)$ and replace it by two adjacent
undeletable vertices $p_{K,0},p_{K,1}$.
Join $c\in C$ to $p_{K,i}$ if $c$ has a neighbor in $K_i$.
Keep the original edges of $G[C]$; the additional edges of the torso
are used only for the treewidth bound.
This two-vertex replacement is used in the proof of
\cite[Theorem~5.1]{DBLP:journals/talg/MarxOR13} to preserve coloring
constraints. Here we also retain the terminal information as follows.
Every vertex of $A\cap C$ remains an undeletable terminal.
Mark $p_{K,i}$ as a terminal precisely when $K_i\cap A\neq\emptyset$.
Each $a\in K_i\cap A$ prescribes color $\phi(a)$ on $p_{K,i}$;
if these prescriptions conflict, reject this assignment.
When no $\phi$ is given, there are no color prescriptions.
For a singleton component choose its only vertex as $K_0$ and let
$K_1=\emptyset$; the additional unmarked vertex $p_{K,1}$ is then
an undeletable leaf.
Call the resulting graph $J$.

For every $X\subseteq C\setminus A$, the bipartite colorings of
$G\sminus X$ and $J\sminus X$ correspond.
A coloring of $G\sminus X$ is constant on each side of the connected
bipartite graph $G[K]$, so it gives the required colors to
$p_{K,0},p_{K,1}$. For a singleton, the unused color can be assigned
to its leaf. Conversely, give every vertex of $K_i$ the color of
$p_{K,i}$; all edges inside $K$ and between $K$ and $C$ then have
differently colored endpoints.
The prescriptions on terminals are preserved by construction.
Also, $K$ and its replacement are both connected, are never deleted,
and have the same attachments to $C$.
Thus paths can be projected and expanded as walks, proving that all
vertices of $A$ lie in one component of $G\sminus X$ if and only
if all marked terminals lie in one component of $J\sminus X$.

Let $w\geq2$ be a bound of the form $2^{O(k^2)}$ on the treewidth of
$\operatorname{torso}(G,C)$.
A tree decomposition of width $O(w)$ can be computed in
$2^{O(w)}(n+m+1)^{O(1)}$ time~\cite{TreewidthApprox2016}.
For each component $K$, the set $N_G(K)$ is a clique in the torso
and is therefore contained in some bag of that decomposition.
Attach to such a bag a new bag
$N_G(K)\cup\{p_{K,0},p_{K,1}\}$.
If $N_G(K)=\emptyset$, attach the two-vertex bag arbitrarily;
if $C=\emptyset$, start with a decomposition consisting of one empty bag.
This gives a tree decomposition of $J$ of width $O(w)$.
There are $O(n)$ vertices in $J$, and constructing $J$ and the torso
requires only polynomial additional time.

\paragraph*{Dynamic programming.}
Choose one marked terminal $r$ of $J$ and keep it in every bag,
increasing the width by at most one.
Convert the decomposition to a rooted form with leaf and root bags
$\{r\}$, where each other node introduces or forgets one vertex,
introduces one edge, or joins two children with the same bag.
This conversion takes polynomial time without increasing the width.
Record for each bag vertex whether it
is deleted or has color $0$ or $1$, a partition describing connectivity
among its surviving vertices in the processed graph, a flag indicating
which parts contain a marked terminal, and the number of deletions.
Deletion is allowed only for vertices of $C\setminus A$.
Color prescriptions and properness on edges are checked when their
vertices and edges are introduced. Introducing an edge merges the
corresponding parts. At a join, require the same deletion and color choices
on the common bag, take the equivalence relation generated
by the two child partitions, and combine their terminal flags.
Add the child deletion counts and subtract the number of deleted bag
vertices, which were counted in both children.
When a part loses its last bag vertex, reject it if
it contains a marked terminal: such a part no longer has a path to
$r$, which remains in every bag. Parts without a marked terminal may
be forgotten. Use $\{r\}$ as the final root bag and accept precisely
states whose deletion count is at most $k$.
These transitions enforce exactly bipartiteness, the
color prescriptions, and connectivity of all marked terminals to $r$.

For bags of size $O(w)$, there are
$2^{O(w\log w)}$ partitions, flags, and deletion/color choices,
and $k+1$ possible deletion counts. Processing pairs of states at
joins remains within this bound up to a polynomial factor in $k$.
Consequently the dynamic program takes
$2^{O(w\log w)}(k+1)^{O(1)}(n+m+1)^{O(1)}$ time.
The separator construction, decomposition algorithm, and at most
$3^k$ assignments are all absorbed by
\[
2^{2^{O(k^2)}}(n+m+1)^{O(1)}.
\]
This proves both versions of the proposition.
\end{proof}

Theorem~\ref{thm:annotated-oct-connectivity} improves the parameter
dependence of this treewidth-based algorithm for prescribed terminal colors.

\end{applicationrevision}
\section{\rev{Unboundedness in Directed Graphs}}
\label{sec:cp-important-lower-bound}

\rev{In this section we show that, in directed graphs, the number of connectivity-preserving important separators cannot be bounded by a function of the separator size. In fact, this already holds for separators of size two.}
Throughout this section, reachability and separators are with respect to directed paths. For a directed graph $G$, vertices $s,t\in \nodes(G)$, and a set $S\subseteq \nodes(G)\setminus\set{s,t}$, we write $C_s(G\sminus S)$ for the set of vertices reachable from $s$ in $G\sminus S$ via directed paths. Like in the undirected case, a set $S\subseteq \nodes(G){\setminus}\set{s,t}$ is an $s,t$-separator if $t\notin C_s(G\sminus S)$. It is a minimal $s,t$-separator if no proper subset of $S$ is an $s,t$-separator.

\rev{We allow the prescribed set $A$ to contain $s$: adding $s$ to $A$ does not change any CP constraint, since $s$ is never deleted and is always reachable from itself.}
Recall that, for $A\subseteq \nodes(G)\setminus\rev{\{t\}}$, a minimal $s,t$-separator $S$ is connectivity-preserving important with respect to $A$ if $A\subseteq C_s(G\sminus S)$, and for every minimal $s,t$-separator $T$ satisfying
\[
A\subseteq C_s(G\sminus T)\subsetneq C_s(G\sminus S),
\]
we have $|T|>|S|$. We denote by $\safeSepskImp{s}{t}{k}(G,A)$ the family of all such separators of size at most $k$.

\begin{revision}
\begin{theorem}
\label{thm:cp-important-lower-bound}
For every pair of integers $N\geq2$ and $r\geq1$, there exist a directed graph $G$ with $r(3N-1)+2$ vertices, vertices $s,t\in\nodes(G)$, and a set $A\subseteq\nodes(G)\setminus\set{t}$ containing $s$, such that
\[
\left|\safeSepskImp{s}{t}{2r}(G,A)\right|=N^r.
\]
In particular, when $r=1$, there are exactly $N$ such separators of size at most two. Consequently, in directed graphs the number of CP-important separators of size at most $k$ is not bounded by any function of $k$ alone.
\end{theorem}
\end{revision}

\begin{revision}
\begin{proof}
Fix $N\geq2$. We first prove the assertion for $r=1$. We construct a directed graph consisting of two internally vertex-disjoint directed $s,t$-paths, called rails, together with additional sink vertices. The rails are
\[
s\to u_1\to\cdots\to u_N\to t
\qquad\text{and}\qquad
s\to v_1\to\cdots\to v_N\to t.
\]
For every $x\in\set{1,\ldots,N-1}$, create a sink vertex $a_x$ and add the arcs $(u_x,a_x)$ and $(v_{N-x},a_x)$. Let $A\eqdef\set{s,a_1,\ldots,a_{N-1}}$. No arcs leave any vertex of $A\setminus\set{s}$, so these vertices create no additional directed $s,t$-paths. The graph has $2+2N+(N-1)=3N+1$ vertices.

Every $s,t$-separator must intersect both rails. Thus every separator of size at most two has the form
\[
S_{p,q}\eqdef\set{u_p,v_q},\qquad 1\leq p,q\leq N.
\]
Conversely, every such set is a minimal $s,t$-separator: deleting it breaks both rails, while restoring either deleted vertex leaves an entire rail intact.

After deleting $S_{p,q}$, the reachable vertices on the rails are exactly $u_1,\ldots,u_{p-1}$ and $v_1,\ldots,v_{q-1}$. Consequently,
\[
a_x\in C_s(G\sminus S_{p,q})
\quad\Longleftrightarrow\quad
x<p\ \text{or}\ N-x<q.
\]
It follows that
\[
A\subseteq C_s(G\sminus S_{p,q})
\quad\Longleftrightarrow\quad
p+q\geq N+1.
\]
\apprev{First, assume $p+q\geq N+1$. Take any index
$x\in\{1,\ldots,N-1\}$. If $x<p$, then $u_x$ lies before the
deleted vertex $u_p$, so $s$ reaches $a_x$ through $u_x$.
Otherwise $x\geq p$, and
$N-x\leq N-p<q$. Thus $v_{N-x}$ lies before the deleted vertex
$v_q$, so $s$ reaches $a_x$ through $v_{N-x}$.
Every $a_x$ is therefore reachable; $s$ itself is always reachable.
Hence $A\subseteq C_s(G\sminus S_{p,q})$.}

\apprev{Conversely, assume every vertex of $A$ is reachable.
If $p=N$, then $q\geq1$ immediately gives $p+q\geq N+1$.
If $p<N$, the vertex $a_p$ exists and must be reachable.
Its only incoming arcs are from $u_p$ and $v_{N-p}$.
Since $u_p$ is deleted, $v_{N-p}$ must be reachable.
This requires $N-p<q$, because only vertices before $v_q$
are reachable on the second rail.
Thus $p+q>N$, or equivalently $p+q\geq N+1$, since $p,q,N$
are integers.}

For CP separators, the reachable set consists of $s$, all of $A$, and the two rail prefixes. Hence, for two CP separators $S_{p,q}$ and $S_{p',q'}$,
\[
C_s(G\sminus S_{p',q'})\subseteq C_s(G\sminus S_{p,q})
\quad\Longleftrightarrow\quad
p'\leq p\ \text{and}\ q'\leq q,
\]
and the inclusion is strict exactly when at least one coordinate inequality is strict. If $p+q=N+1$, no strictly smaller pair is CP, so $S_{p,q}$ is CP-important. If $p+q>N+1$, then $p\geq2$, and $S_{p-1,q}$ is a CP separator with a strictly smaller reachable set and the same size. Thus $S_{p,q}$ is not CP-important. There are exactly $N$ pairs with $p+q=N+1$, proving the assertion for $r=1$.

For arbitrary $r\geq1$, take $r$ copies of this construction and identify
their sources into one vertex $s$ and their targets into one vertex $t$.
All other vertices remain distinct. In copy $i$, denote the rail vertices
by $u_{i,1},\ldots,u_{i,N}$ and $v_{i,1},\ldots,v_{i,N}$, and the sinks by
$a_{i,1},\ldots,a_{i,N-1}$. Set
\[
A\eqdef\set{s}\cup
\set{a_{i,x}:1\leq i\leq r,\ 1\leq x\leq N-1}.
\]
The resulting graph has $2+r(2N+N-1)=r(3N-1)+2$ vertices.
The source has no incoming arcs, the target has no outgoing arcs, and
the remaining vertices of $A$ are sinks. Thus identifying the endpoints
does not create paths that switch between copies.

There are $2r$ internally vertex-disjoint $s,t$-rails. Every separator of
size at most $2r$ must therefore contain exactly one vertex from each rail
and no other vertices. It has the form
\[
S_{\mathbf p,\mathbf q}
\eqdef\bigcup_{i=1}^r\set{u_{i,p_i},v_{i,q_i}},
\qquad 1\leq p_i,q_i\leq N.
\]
Every such set is a minimal $s,t$-separator: it breaks all rails, and
restoring any deleted vertex leaves its entire rail intact.
By the two-rail argument, it is CP if and only if
$p_i+q_i\geq N+1$ for every $i$. For a CP separator, the reachable set
consists of $A$ and the prefixes of all $2r$ rails. Hence, for two CP
separators of this form,
\[
C_s(G\sminus S_{\mathbf p',\mathbf q'})
\subseteq C_s(G\sminus S_{\mathbf p,\mathbf q})
\quad\Longleftrightarrow\quad
p_i'\leq p_i\ \text{and}\ q_i'\leq q_i\quad\text{for every }i,
\]
with strict inclusion exactly when at least one inequality is strict.

If $p_i+q_i=N+1$ for every $i$, decreasing any coordinate violates the
CP condition in that copy. No separator of no greater size can therefore
dominate $S_{\mathbf p,\mathbf q}$, since every such separator must have
the displayed form. Conversely, if $p_i+q_i>N+1$ for some $i$, then
$p_i\geq2$, and replacing $p_i$ by $p_i-1$ gives a CP minimal separator
of the same size with a strictly smaller reachable set.
Thus $S_{\mathbf p,\mathbf q}$ is CP-important exactly when
$p_i+q_i=N+1$ for every $i$. Each copy permits exactly $N$ choices,
independently, so there are exactly $N^r$ CP-important separators of
size at most $2r$. This also proves Theorem~\ref{thm:directed-lower-bound}.
\end{proof}
\end{revision}

\rev{The construction also rules out enumerating all directed CP-important
separators in total time $f(k)n^c$, for any function $f$ and constant $c$
independent of $k$. Indeed, fix an integer $r>c$ and set $k=2r$.
As $N$ grows, the graph has $n=r(3N-1)+2$ vertices but $N^r$ outputs.
Even writing these outputs exceeds $f(2r)(r(3N-1)+2)^c$ for sufficiently
large $N$. This is an obstruction to total enumeration time; it does not
rule out an FPT algorithm for deciding existence or finding one separator.}

The construction crucially uses directed edges. The vertices of $\rev{A\setminus\set{s}}$ are sinks, and therefore they can certify reachability constraints without creating new directed $s,t$-paths. In an undirected graph, attaching one constraint vertex to several rails would create cross-rail bypasses and would destroy the one-vertex-per-rail separator structure used in the proof. \rev{Thus this construction separates the directed and undirected settings; it does not establish a lower bound matching the undirected upper bound.}

\section{Research methods and AI assistance}
\label{sec:ai-assistance}
\begin{revision}
	AI tools were used for counterexample search, proof critique and
	verification, literature checks, and presentation.
	Specifically, OpenAI Codex assisted in developing and writing
	the Mixed Multiway Cut-Uncut application in
	Section~\ref{sec:mixed-two-sets-appendix} and the Annotated Odd
	Cycle Transversal application in
	Section~\ref{sec:annotated-oct-appendix}, particularly the
	treewidth-reduction comparison in
	Section~\ref{subsec:treewidthReduction}.
	All scientific claims, results, proofs, and interpretations
	were reviewed and verified by the authors, who remain fully
	responsible for the content of this paper.
\end{revision}
\end{document}